\documentclass[12pt, a4paper, english,oneside]{article}
\usepackage{url} 
\usepackage[round]{natbib}
\usepackage{hyperref}
\usepackage{authblk}
\usepackage{amsmath, amssymb, multicol, amsfonts, amsthm, amscd, enumerate}
\usepackage{lmodern} 
\usepackage[latin1]{inputenc}
\usepackage[T1]{fontenc}
\usepackage{graphicx, tikz}
\usepackage{graphicx}
\usepackage{geometry}
\hypersetup{
    colorlinks,
    linkcolor={red!50!black},
    citecolor={blue!50!black},
    urlcolor={blue!80!black}
}
\usepackage{epstopdf}
\usepackage{amscd}
\usepackage{caption}
\usepackage{multirow}
\usepackage{rotating}
\usepackage{bm}
\usepackage{bbm}
\usepackage{latexsym}
\usepackage{mathrsfs}
\usepackage{amsthm}
\usepackage{mathrsfs}
\usepackage{booktabs}
\usepackage{caption}
\usepackage{subfig}
\usepackage{url}
\usepackage{rotating}
\usepackage{multirow}
\usepackage{lscape}

\newtheorem{theorem}{Theorem}

\newtheorem{corollary}[theorem]{Corollary}

\newtheorem{definition}[theorem]{Definition}

\newtheorem{proposition}[theorem]{Proposition}
\newtheorem{remark}[theorem]{Remark}

\newtheorem{assumption}[theorem]{Assumption}

\newcommand\bbone{\ensuremath{\mathbbm{1}}}

\newcommand{\sgn}{\text{sgn}}

\def\bet{\mbox{\boldmath $\eta$}} 
\def\bmu{\mbox{\boldmath $\mu$}}
\def\bGamma{\mbox{\boldmath $\Gamma$}}
\def\bomega{\mbox{\boldmath $\omega$}} 
\def\bOmega{\mbox{\boldmath $\Omega$}}
\def\bxi{\mbox{\boldmath $\xi$}}

\def\bSigma{\mathbf{\Sigma}}
\def\bsigma{\mathbf{\sigma}}
\def\btau{\mbox{\boldmath $\tau$}}

\def\bPhi{\mbox{\boldmath $\Phi$}}

\def\bb{\mathbf{b}}

\def\bP{\mathbf{P}}

\def\bu{\mathbf{u}} 
\def\bt{\mathbf{t}}
\def\by{\mathbf{y}} 
\def\bY{\mathbf{Y}}
\def\0{\mbox{\bf{0}}}

\def\bZ{\mathbf{Z}}

\def\bX{\mathbf{X}} 

\def\bO{\mathbf{0}}

\def\bv{\mathbf{v}}%
\def\bq{\mathbf{q}}

\def\bD{\mathbf{D}}
\def\bU{\mathbf{U}}
 
\def\diag{\mbox{diag}}
\def\trace{\mbox{tr}}

\def\Var{\mbox{Var}}
\def\Cov{\mbox{Cov}}

\begin{document}

\def\spacingset#1{\renewcommand{\baselinestretch}%
{#1}\small\normalsize} \spacingset{1}

\title{The Sparse Multivariate Method\\ of Simulated Quantiles}

\author[1]{Mauro Bernardi\thanks{Corresponding author: Via C. Battisti, 241, 35121 Padua, Italy. e-mail: mauro.bernardi@unipd.it, Tel.: +39.049.8274165.}}
\author[2]{Lea Petrella\thanks{e-mail: lea.petrella@uniroma1.it.}}
\author[3]{Paola Stolfi\thanks{e-mail: paola.stolfi@uniroma3.it.}}
\affil[1]{Department of Statistical Sciences, University of Padova and Istituto per le Applicazioni del Calcolo ``Mauro Picone'' - CNR, Roma, Italy}
\affil[2]{MEMOTEF Department, Sapienza University of Rome}
\affil[3]{Department of Economics, Roma Tre University}

\date{\today}

\maketitle

\begin{abstract}
\noindent In this paper the method of simulated quantiles (MSQ) of \cite{dominicy_veredas.2013} and \cite{dominicy_etal.2013} is extended to a general multivariate framework (MMSQ) and to provide sparse estimation of the scaling matrix (Sparse--MMSQ).
The MSQ, like alternative likelihood--free procedures, is based on the minimisation of the distance between appropriate statistics evaluated on the true and synthetic data simulated from the postulated model. Those statistics are functions of the quantiles providing an effective way to deal with distributions that do not admit moments of any order like the $\alpha$--Stable or the Tukey lambda distribution. The lack of a natural ordering represents the major challenge for the extension of the method to the multivariate framework. Here, we rely on the notion of projectional quantile recently introduced by \cite{hallin_etal.2010} and \cite{kong_mizera.2012}. We establish consistency and asymptotic normality of the proposed estimator. The smoothly clipped absolute deviation (SCAD) $\ell_1$--penalty of \cite{fan_li.2001} is then introduced into the MMSQ objective function in order to achieve sparse estimation of the scaling matrix which is the major responsible for the curse of dimensionality problem. We extend the asymptotic theory and we show that the sparse--MMSQ estimator enjoys the oracle properties under mild regularity conditions. The method is illustrated and its effectiveness is tested using several synthetic datasets simulated from the Elliptical Stable distribution (ESD) for which alternative methods are recognised to perform poorly. The method is then applied to build a new network--based systemic risk measurement framework. The proposed methodology to build the network relies on a new systemic risk measure and on a parametric test of statistical dominance. 
\end{abstract}

\noindent%


\section{Introduction}
\label{sec:intro}
%
\noindent Model--based statistical inference primarily deals with parameters estimation. Under the usual assumption of data being generated from a fully specified model belonging to a given family of distributions $F_\vartheta$ indexed by a parameter $\vartheta\subset\Theta\in\mathbb{R}^p$, inference on the true unknown parameter $\vartheta_0$ can be easily performed by maximum likelihood. However, in some pathological situations the maximum likelihood estimator (MLE) is difficult to compute either because of the model complexity or because the probability density function is not analytically available. For example, the computation of the log--likelihood may involve numerical approximations or integrations that highly deteriorate the quality of the resulting estimates. 
Moreover, as the dimension of the parameter space increases the computation of the likelihood or its maximisation in a reasonable amount of time becomes even more prohibitive. In all those circumstances, the researcher should resort to alternative solutions. The method of moments or its generalised versions (GMM), \cite{hansen.1982} or (EMM), \cite{gallant_tauchen.1996}, may constitute feasible solutions when expressions for some moment conditions that uniquely identify the parameters of interest are analytically available. When this is not the case, simulation--based methods, such as, the method of simulated moments (MSM), \cite{mcfadden.1989}, the method of simulated maximum likelihood (SML), \cite{gourieroux_monfort.1996} and its nonparametric version \cite{kristensen_shin.2012} or the indirect inference (II) method \cite{gourieroux_etal.1993}, are the only viable solutions to the inferential problem. \cite{jiang_turnbull.2004} give a comprehensive review of indirect inference from a statistical point of view. Despite their appealing characteristics of only requiring to be able to simulate from the specified DGP, some of those methods suffer from serious drawbacks. The MSM, for example, requires that the existence of the moments of the postulated DGP is guaranteed, while, the II method relies on an alternative, necessarily misspecified, auxiliary model as well as on a strong form of identification between the parameters of interests and those of the auxiliary model. The quantile--matching estimation method (QM), \cite{koenker.2005}, exploits the same idea behind the method of moments without requiring any other condition. The QM approach estimates model parameters by matching the empirical percentiles with their theoretical counterparts thereby requiring only the existence of a closed form expression for the quantile function.\newline
\indent This paper focuses on the method of simulated quantiles recently proposed by \cite{dominicy_veredas.2013} as a simulation--based extension of the QM of \cite{koenker.2005}. As any other simulation--based method, the MSQ estimates parameters by minimising a quadratic distance between a vector of quantile--based summary statistics calculated on the available sample of observations and that calculated on synthetic data generated from the postulated theoretical model. Specifically, we extend the method of simulated quantiles to deal with multivariate data, originating the multivariate method of simulated quantiles (MMSQ). The extension of the MSQ to multivariate data is not trivial because it requires the definition of multivariate quantile that is not unique given the lack of a natural ordering in $\mathbb{R}^{n}$ for $n>1$. Indeed, only very recently the literature on multivariate quantiles has proliferated, see, e.g., \cite{serfling.2002} for a review of some extensions of univariate quantiles to the multivariate case. Here we rely on the definition of projectional quantile of \cite{hallin_etal.2010b} and \cite{kong_mizera.2012}, that is a particular version of directional quantile. This latter definition is particularly appealing since it allows to reduce the dimension of the problem by projecting data towards given directions in the plane. Moreover, the projectional quantiles incorporate information on the covariance between the projected variables which is crucial in order to relax the assumption of independence between variables. An important methodological contribution of the paper concerns the choice of the relevant directions to project data in order to summarise the information for the parameters of interest. Although the inclusion of more directions can convey more information about the parameters, it comes at a cost of a larger number of expensive quantile evaluations. Of course the number of quantile functions is unavoidably related to the dimension of the observables and strictly depends upon the considered distribution. We provide a general solution for Elliptical distributions and for those Skew--Elliptical distributions that are closed under linear combinations. We also establish consistency and asymptotic normality of the proposed MMSQ estimator under weak conditions on the underlying true DGP. The conditions for consistency and asymptotic Normality of the MMSQ are similar to those imposed by \cite{dominicy_veredas.2013} with minor changes due to the employed projectional quantiles. Moreover, for the distributions considered in our illustrative examples, full details on how to calculate all the quantities involved in the asymptotic variance--covariance matrix are provided. The asymptotic variance--covariance matrix of the MMSQ estimator is helpful to derive its efficient version, the E--MMSQ.\newline
\indent As any other simulation--based method the MMSQ does not effectively deal with the curse of dimensionality problem, i.e., the situation where the number of parameters grows quadratically or exponentially with the dimension of the problem. Indeed, the right identification of the sparsity patterns becomes crucial because it reduces the number of parameters to be estimated. Those reasonings motivate the use of sparse estimators that automatically shrink to zero some parameters, such as, for example, the--off diagonal elements of the variance--covariance matrix. Several works related to sparse estimation of the covariance matrix are available in literature; most of them are related to the graphical models, where the precision matrix, e.g., the inverse of the covariance matrix, represents the conditional dependence structure of the graph. \cite{friedman_etal.2008} propose a fast algorithm based on coordinate--wise updating scheme in order to estimate a sparse graph using the least absolute shrinkage and selection operator (LASSO) $\ell_1$--penalty of \cite{tibshirani.1996}. \cite{meinshausen_buhlmann.2006} propose a method for neighbourhood selection using the LASSO $\ell_1$--penalty as an alternative to covariance selection for Gaussian graphical models where the number of observations is less than the number of variables. \cite{gao_massam.2015} estimate the variance--covariance matrix of symmetry--constrained Gaussian models using three different $\ell_1$--type penalty functions, i.e., the LASSO, the smoothly clipped absolute deviation (SCAD) of \cite{fan_li.2001} and the minimax concave penalty (MCP) of \cite{zhang.2010}. \cite{bien_tibshirani.2011} proposed a penalised version of the log--likelihood function, using the LASSO penalty, in order to estimate a sparse covariance matrix of a multivariate Gaussian distribution. Previous work show that sparse estimation has been proposed mainly either within the regression framework or in the context of Gaussian graphical models. In boh those cases, sparsity patterns are imposed by penalising a Gaussian log--likelihood.\newline 
%
%
\indent In this paper we handle the lack of the model--likelihood or the existence of valid moment conditions together with the curse of dimensionality problem within a high--dimensional non--Gaussian framework. Specifically, our approach penalises the objective function of the MMSQ by adding a SCAD $\ell_1$--penalisation term that shrinks to zero the off--diagonal elements of the scale matrix of the postulated distribution. Moreover, we extend the asymptotic theory in order to account for the sparsity estimation, and we prove that the resulting sparse--MMSQ estimator enjoys the oracle properties of \cite{fan_li.2001} under mild regularity conditions. Moreover, since the chosen penalty is concave, we deliver a fast and efficient algorithm to solve the optimisation problem. 
\newline
\indent The proposed methods can be effectively used to make inference on the parameters of large--dimensional distributions such as, for example, Stable, Elliptical Stable (\citealt{samorodnitsky_etal.1994}), Skew--Elliptical Stable (\citealt{branco_dey.2001}), Copula (\citealt{oh_patton.2013}), multivariate Gamma (\citealt{mathai_moschopoulos.2015}) and Tempered Stable (\citealt{ismo.1995}). Among those, the Stable distribution allows for infinite variance, skewness and heavy--tails that exhibit power decay allowing extreme events to have higher probability mass than in Gaussian model. To test the effectiveness of the MMSQ and sparse--MMSQ methods several synthetic datasets have been simulated from the Elliptical Stable distribution previously considered by \cite{lombardi_veredas.2009}. For a summary of the properties of the stable distributions see \cite{zolotarev.1964} and \cite{samorodnitsky_etal.1994}, which provide a good theoretical background on heavy--tailed distributions. Univariate Stable laws have been studied in many branches of the science and their theoretical properties have been deeply investigated from multiple perspectives, therefore many tools are now available for estimation and inference on parameters, to evaluate the cumulative density or the quantile function, or to perform fast simulation. Stable distribution plays an interesting role in modelling multivariate data. Its peculiarity of having heavy tailed properties and its closeness under summation make it appealing in the financial contest. Nevertheless, multivariate Stable laws pose several challenges that go further beyond the lack of closed form expression for the density. Although general expressions for the multivariate density have been provided by \cite{abdul-hamid_nolan.1998}, \cite{byczkowski_etal.1993} and \cite{matsui_takemura.2009}, their computations is still not feasible in dimension larger than two. A recent overview of multivariate Stable distributions can be found in \cite{nolan.2008}.\newline 
%
\indent As regards applications to real data, we consider the well--known problem of evaluating the systemic relevance of the financial institutions or banks belonging to a given market. After the Bear Stearns hedge funds collapse in July 2007, and the consequent global financial crisis which originated in the United States and then spread quickly to the rest of the world, the threat of a global collapse of the whole financial system has been becoming the major concern of financial regulators. Systemic risk, as opposed to risk associated with any one individual entity, aims at evaluating to which extent the bankruptcy of a bank or financial institutions may degenerate to a collapse of the system as a consequence of a contagion effect. While individual risks are assessed using individual Value--at--Risks (VaR), one the most employed systemic risk measure has been becoming the Conditional VaR (CoVaR), introduced by \cite{adrian_brunnermeier.2011,adrian_brunnermeier.2016}. Since then, the assessment of financial risk in a multi--institution framework where some institutions are subject to systemic or non--systemic distress events is one of the hot topics which has received large attention from scholars in Mathematical Finance, Statistics, Management, see, e.g., \cite{acharya_etal.2012}, \cite{billio_etal.2012}, \cite{bernardi_etal.2017d}, \cite{girardi_ergun.2013}, \cite{caporin_etal.2013}, \cite{engle_etal.2014}, \cite{hautsch_etal.2014}, \cite{lucas_etal.2014}, \cite{bernardi_catania.2015}, \cite{bernardi_etal.2015}, \cite{sordo_etal.2015},
\cite{bernardi_etal.2016}, \cite{bernardi_etal.2016b}, \cite{bernardi_etal.2016c}, \cite{brownlees_engle.2016} and \cite{salvadori_etal.2016}, just to quote a few of the most relevant approaches. For an extensive and up to date survey on systemic risk measures, see \cite{bisias_etal.2012}, while the recent literature on systemic risk is reviewed by \cite{benoit_etal.2016}. The CoVaR measures the systemic impact on the whole financial system of a distress event affecting an institution by calculating the VaR of the system conditioned to the distress event as measured by the marginal VaR of that institution. As recognised by \cite{bernardi_etal.2017d} this definition of CoVaR fails to consider the institution as a part of a system. Here, we introduce a new definition of CoVaR, the NetCoVaR, that overcomes this drawback by aggregating individual institutions providing a measure of profit and loss of the whole financial market. Despite its appealing definition, the NetCoVaR, as any other risk measure, provide only point estimates of the amount of systemic risk. Within this context, statistical methods aims to assess whether two risk measures are statistically different from each other. As concerns the CoVaR, recently, \cite{castro_ferrari.2014} proposed a nonparametric dominance test where pairwise CoVaRs are compared in order to statistically assess the systemic relevance of the different institutions. Here, we propose a parametric counterpart of the test of \cite{castro_ferrari.2014} and we assume profits--and--losses of the different institutions are Elliptically Stable distributed. The asymptotic distribution of the dominance test is provided under the mild assumption of elliptically contoured distributions for the involved random variables. The dominance test is subsequently used to build a network that represents the interdependence relations among institutions. In this context the ESD distribution plays a relevant role either because data are contaminated by the presence of outliers or because the methodology strongly relies on the presence of heavy--tailed distributions such as the systemic risk assessment.\newline
%
%
%
\indent The remainder of the paper is structured as follows. In Section \ref{sec:MMSQ} we introduce the multivariate Method of Simulated Quantiles, and we establish the basic asymptotic properties. The asymptotic variance of the estimator is necessary to select the optimal weighting matrix for the square distance in order to obtain the efficient MMSQ estimator. Section \ref{sec:msq_sparse} deals with the curse of dimensionality problem, introduces the Sparse--MMSQ estimator that induces sparsity in the scale matrix using the SCAD $\ell_1$--penalty and shows that the Sparse--MMSQ enjoys the oracle properties under mild regularity conditions. The penalised estimator cannot be used to make inference on the parameters shrunk to zero, therefore Section \ref{sec:msq_sparse} ends by proposing a de--sparsified MMSQ estimator. The effectiveness of the method is tested in Section \ref{sec:msq_synthetic_ex}, where several synthetic datasets from the Elliptical Stable distribution are considered. Section \ref{sec:application_port_opt} is devoted to the empirical application that aims to illustrate how the methodological contributions of the paper can be applied to the systemic risk assessment. Section \ref{sec:conclusion} concludes. Technical proofs of the theorems are deferred to Appendix \ref{app:appendix_proofs}.
%
\section{Multivariate method of simulated quantiles}
\label{sec:MMSQ}
%
\noindent In this Section we first recall the basic concepts on directional and projectional quantiles. Then, the multivariate method of simulated quantiles is introduced, and results about the consistency and asymptotic properties of the estimator are proposed. 
%
\subsection{Directional quantiles}
\label{sec:dir_quantiles}
%
\noindent The MMSQ requires the prior definition of the concept of multivariate quantile, a notion still vague until quite recently because of the lack of a natural ordering in dimension greater than one. Here, we relies on the definition of directional quantiles and projectional quantiles introduced by \cite{hallin_etal.2010b}, \cite{paindaveine_siman.2011} and \cite{kong_mizera.2012}.
%
%
%
We first recall the definition of directional quantile given in \cite{hallin_etal.2010b} and then we introduce the main assumptions that we will use to develop MMSQ.
%
%
\begin{definition}
Let $\bY=\left(Y_{1},Y_{2},\dots,Y_{m}\right)$ be a $m$--dimensional random vector in $\mathbb{R}^{m}$, $\bu \in\mathbb{S}^{m-1}$ be a vector in the unit sphere $\mathbb{S}^{m-1}=\left\{\bu\in\mathbb{R}^m\::\:\bu^\prime\bu=1\right\}$ and $\tau\in \left(0,1\right)$. The $\tau\bu$--quantile of $\bY$ is defined as any element of the collection $\Pi_{\tau\bu}$ of hyperplanes 
\begin{align}
\pi_{\tau\bu}=\left\{\bY:\bb_{\tau\bu}^{\prime}\bY-q_{\tau\bu}=0\right\},\nonumber
\end{align}
such that 
\begin{align}\label{eq:dq1}
\left(q_{\tau\bu}, \bb_{\tau\bu}^{\prime}\right)^{\prime}\in \left\{\arg\min_{\left(q,\bb\right)}\Psi_{\tau\bu}\left(q,\bb\right) \quad s.t. \quad \bb^{\prime}\bu=1\right\},
\end{align}
where
\begin{align}
\Psi_{\tau\bu}\left(q,\bb\right)=\mathbb{E}\Big[\rho_{\tau }\left(\bb^{\prime}\bY-q\right)\Big],
\label{eq:dir_quantile_loss}
\end{align}
and $\rho_\tau\left(z\right)=z\left(\tau-\bbone_{\left(-\infty,0\right)}\left(z\right)\right)$ denotes the quantile loss function evaluated at $z\in\mathbb{R}$, $q\in \mathbb{R}$, $\bb\in\mathbb{R}^{m}$ and $\mathbb{E}\left(\cdot\right)$ denotes the expectation operator.
\end{definition}
\noindent The term directional is due to the fact that the multivariate quantile defined above is associated to a unit vector $\bu\in\mathbb{S}^{m-1}$. 
\begin{assumption}\label{as:0}
The distribution of the random vector $\textbf{Y}$ is absolutely continuous with respect to the Lebesgue measure on $\mathbb{R}^{m}$, with finite first order moment, having density $f_{\bY}$ that has connected support.
\end{assumption} 
\noindent Under assumption \ref{as:0}, for any $\tau\in \left(0,1\right)$ the minimisation problem defined in equation \eqref{eq:dq1} admits a unique solution $\left(a_{\tau\bu},\bb_{\tau\bu}\right)$, which uniquely identifies one hyperplane $\pi_{\tau\bu}\in\Pi_{\tau\bu}$.\newline
\indent A special case of directional quantile is obtained by setting $\bb=\bu$; in that case the directional quantile $\left(a_{\tau\bu},\bu\right)$ becomes a scalar value and it inherits all the properties of the usual univariate quantile. This particular case of directional quantile is called projectional quantile, whose formal definition, reported below, is due to \cite{kong_mizera.2012} and \cite{paindaveine_siman.2011}. 
\begin{definition}
Let $\bY=\left(Y_{1},Y_{2},\dots,Y_{m}\right)\in\mathbb{R}^{m}$, $\bu \in\mathbb{S}^{m-1}$ be a vector in the unit sphere $\mathbb{S}^{m-1}$, and $\tau\in \left(0,1\right)$. The $\tau\bu$ projectional quantile of $\bY$ is defined as 
\begin{align}\label{eq:dq2}
q_{\tau\bu}\in \left\{\arg\min_{q\in\mathbb{R}}\Psi_{\tau\bu}\left(q\right)\right\},
\end{align}
where $\Psi_{\tau\bu}\left(q\right)=\Psi_{\tau\bu}\left(q,\bu\right)$ in equation \eqref{eq:dir_quantile_loss}.
%
\end{definition}
\noindent Clearly the $\tau\bu$--projectional quantile is the $\tau$--quantile of the univariate random variable $Z=\bu^\prime\bY$. This feature makes the definition of projectional quantile particularly appealing in order to extend the MSQ to a multivariate setting because, once the direction is properly chosen, it 
reduces to the usual univariate quantile.
%
Given a sample of observations $\left\{\by_{i}\right\}_{i=1}^n$ from $\bY$, the empirical version of the projectional quantile is defined as 
\begin{align}
q_{\tau\bu}^{n}\in \left\{\arg\min_{q}\Psi_{\tau\bu}^{n}\left(q\right)\right\},\nonumber
\end{align}
where $\Psi_{\tau \bu}^{n}\left(q\right)=\frac{1}{n}\sum_{i=1}^{n}\Big[\rho_{\tau }\left(\bu^{\prime}\by_{i}-q\right)\Big]$ denotes the empirical version of the loss function defined in equation \eqref{eq:dir_quantile_loss}.
%
\subsection{The method of simulated quantiles}
\label{sec:msq_method}
%
\noindent The MSQ introduced by \cite{dominicy_veredas.2013} is likelihood--free simulation--based inferential procedure based on matching quantile--based measures, that is particularly useful in situations where either the density function does is not analytically available and/or moments do not exist. Since it is essentially a simulation--based method it can be applied to all those random variables that can be easily simulated. In the contest of MSQ, parameter are estimated by minimising the distance between an appropriately chosen vector of functions of empirical quantiles and their simulated counterparts based on the postulated parametric model. An appealing characteristic of the MSQ that makes it a valid alternative to other likelihood--free methods, such as the indirect inference of \cite{gourieroux_etal.1993}, is that the MSQ does not rely on a necessarily misspecified auxiliary model. Furthermore, empirical quantiles are robust ordered statistics being able to achieve high protection against bias induced by the presence of outlier contamination.\newline
\indent Here we introduce the MMSQ using the notion of projectional quantiles defined in Section \ref{sec:dir_quantiles}. Let $\bY$ be a $d$--dimensional random variable with distribution function $F_{\bY}\left(\cdot,\vartheta\right)$, which depends on a vector of unknown parameters $\vartheta\subset\Theta\in\mathbb{R}^{k}$, and $\by=\left(\by_{1},\by_2,\dots,\by_{n}\right)^{\prime}$ be a vector of $n$ independent realisations of $\bY$. Moreover, let $\bq^{\boldsymbol{\tau},\bu}_{\vartheta}=\left(q^{\tau_{1}\bu}_{\vartheta},q^{\tau_{2}\bu,}_{\vartheta},\dots,q^{\tau_{s}\bu}_{\vartheta}\right)$ be a $m\times s$ matrix of projectional quantiles at given confidence levels $\tau_{k}\in\left(0,1\right)$ with $k=1,2,\dots,s$, and $\bu\in\mathbb{S}^{m-1}$. Let $\bPhi_{\bu,\vartheta}=\bPhi\left(\bq^{\boldsymbol{\tau},\bu}_{\vartheta}\right)$ be a $b\times 1$ vector of quantile functions assumed to be continuously differentiable with respect to $\vartheta$ for all $\bY$ and measurable for $\bY$ and for all $\vartheta\subset\Theta$. Let us assume also that $\bPhi_{\bu,\vartheta}$ cannot be computed analytically but it can be empirically estimated on simulated data; denote those quantities by $\tilde{\bPhi}^r_{\mathbf{u},\vartheta}$. Let $\hat{\bq}^{\boldsymbol{\tau},\bu}=\left(\hat{q}^{\tau_{1}\bu},\hat{q}^{\tau_{2}\bu},\dots,\hat{q}^{\tau_{s}\bu}\right)$ be a $m\times s$ matrix of projectional quantiles with $\bu\in\mathbb{S}^{m-1}$ and $0<\tau_{1}<\dots<\tau_{s}<1$, and let $\hat{\bPhi}_{\bu}=\bPhi\left(\hat{\bq}^{\boldsymbol{\tau},\bu}\right)$ be a $b\times 1$ vector of functions of sample quantiles, that is measurable of $\bY$.\newline
\indent The MMSQ at each iteration $j=1,2,\dots$ estimates $\tilde{\bPhi}_{\bu,\vartheta}$ on a sample of $R$ replication simulated from ${\by_{r,j}^\ast}\sim F_{\bY}\left(\cdot,\vartheta^{(j)}\right)$, for $r=1,2,\dots,R$, as $\tilde{\bPhi}_{\bu,\vartheta_j}^{R}=\frac{1}{R}\sum_{r=1}^{R}\tilde{\bPhi}_{\bu,\vartheta_j}^{r}$, where $\tilde{\bPhi}_{\bu,\vartheta_j}^{r}$ is the function $\bPhi_{\bu,\vartheta}$ computed at the $r$--th simulation path. The parameters are subsequently updated by minimising the distance between the vector of quantile measures calculated on the true observations $\hat{\bPhi}_{\mathbf{u}}$ and that calculated on simulated realisations $\tilde{\bPhi}_{\mathbf{u},\vartheta_j}^{R}$ as follows
\begin{align}\label{eq:mmsq_min_problem}
\hat{\vartheta}=\arg\min_{\vartheta\in \vartheta}\left(\hat{\bPhi}_\bu-\tilde{\bPhi}_{\mathbf{u},\vartheta}^{R}\right)^{\prime}\mathbf{W}_{\vartheta}\left(\hat{\bPhi}_{\mathbf{u}}-\tilde{\bPhi}_{\mathbf{u},\vartheta}^{R}\right),
\end{align}
where $\mathbf{W}_{\vartheta}$ is a $b\times b$ symmetric positive definite weighting matrix. The method of simulated quantiles of \cite{dominicy_veredas.2013} reduces to the selection of the first canonical direction $\bu_1=\left(1,0,\dots,0\right)$ as relevant direction in the projectional quantile.\newline
\indent The vector of functions of projectional quantiles $\bPhi_{\bu,\vartheta}$ should be carefully selected in order to be as informative as possible for the vector of parameters of interest. In their applications, \cite{dominicy_veredas.2013} only propose to use the MSQ to estimate the parameters of univariate Stable law. Toward this end they consider the following vector of quantile--based statistics, as in \cite{mcculloch.1986} and \cite{Kim_white.2004}
\begin{align}
\bPhi_{\vartheta}=\left(\frac{q_{0.95}+q_{0.05}-2q_{0.5}}{q_{0.95}-q_{0.05}}, \frac{q_{0.95}-q_{0.05}}{q_{0.75}-q_{0.25}}, q_{0.75}-q_{0.25}, q_{0.5}\right)^\prime.\nonumber
\end{align}
where the first element of the vector is a measure of skewness, the second one is a measure of kurtosis and the last two measures refer to scale and location, respectively. Of course, the selection of the quantile--based summary statistics depend either on the nature of the parameter and on the assumed distribution. The MMSQ generalises also the MSQ proposed by \cite{dominicy_etal.2013} where they estimate the elements of the variance--covariance matrix of multivariate elliptical distributions by means of a measure of co--dispersion which consists in the in interquartile range of the standardised variables projected along the bisector. The MMSQ based on projectional quantiles is more flexible and it allows us to deal with more general distributions than elliptically contoured distributions because it relies on the construction of quantile based measures on the variables projected along an optimal directions which depend upon the considered distribution. The selection of the relevant direction is deferred to Section \ref{sec:msq_synthetic_ex}.
%
\subsection{Asymptotic theory}
\label{sec:msq_asy}
%
\noindent In this section we establish consistency and asymptotic normality of the proposed MMSQ estimator. The next theorem establish the asymptotic properties of projectional quantiles. 
\begin{theorem}\label{th:proj_quant_asy}
Let $\bY\in\mathbb{R}^m$ be a random vector with cumulative distribution function $F_{\bY}$ and variance--covariance matrix $\bSigma_{\bY}$. Let $\left\{\by_{i}\right\}_{i=1}^{n}$ be a sample of iid observations from $F_{\bY}$. Let $\bu_{1},\bu_{2},\dots,\bu_{K}\in \mathbb{S}^{m-1}$ and $Z_{k}=\bu_{k}'\bY$ be the projected random variable along $\bu_{k}$ with cumulative distribution function $F_{Z_{k}}$, for $k=1,2,\dots,K$. Let $\btau_{k}=\left(\tau_{1,k},\tau_{2,k},\dots,\tau_{s,k}\right)$ where $\tau_{j,k}\in\left(0,1\right)$, $\bq_{\boldsymbol{\tau}_{k},\bu_{k}}=\left(q_{\tau_{1,k}\bu_{k}},q_{\tau_{2,k}\bu_{k}},\dots, q_{\tau_{s,k}\bu_{k}}\right)$ be the vector of directional quantiles along the direction $\bu_{k}$ and suppose $Var\left(Z_{k}\right)<\infty$, for $k=1,2,\dots,K$. Let us assume that $F_{Z_{k}}$ is differentiable in $q_{\tau_{j,k}\bu_{k}}$ and $F_{Z_{k}}^{\prime}\left(q_{\tau_{j,k}\bu_{k}}\right)=f_{Z_{k}}\left(q_{\tau_{j,k}\bu_{k}}\right)>0$, for $k=1,2,\dots,K$ and $j=1,2,\dots,s$. Then 
\begin{enumerate}
\item[{\it (i)}] for a given direction $\bu_{k}$, with $k=1,2,\dots, K$, it holds
\begin{align}
\sqrt{n}\left(\hat{\bq}_{\boldsymbol{\tau}_{k},\bu_{k}}-\mathbf{q}_{\boldsymbol{\tau}_{k},\bu_{k}}\right)\xrightarrow[]{d}\mathcal{N}\left(\bO,\bet\right), \nonumber
\end{align}
as $n\rightarrow\infty$, where $\boldsymbol{\eta}$ denotes a $\left(K\times K\right)$ symmetric matrix whose generic $\left(r,c\right)$ entry is
\begin{align}
\eta_{r,c}=\frac{\tau_{r}\wedge\tau_{c}-\tau_{r}\tau_{c}}{f_{Z_{k}}\left(\mathbf{q}_{\tau_{r},\bu_{k}}\right)f_{Z_{k}}\left(\mathbf{q}_{\tau_{c},\bu_{k}}\right)},\nonumber
\end{align}
for $r,c=1,2,\dots,K$;
\item[{\it (ii)}] for a given level $\tau_{j}$, with $j=1,2,\dots, s$, it holds 
\begin{align}
\sqrt{n}\left(\hat{\bq}_{\tau_{j}}-\bq_{\tau_{j}}\right)\xrightarrow[]{d}\mathcal{N}\left(\bO,\bet\right),\nonumber
\end{align}
as $n\rightarrow\infty$, where $\bq_{\tau_{j}}=\left(q_{\tau_{j}\bu_{1}},\dots,q_{\tau_{j}\bu_{K}}\right)$, 
\begin{align}
\eta_{r,c}&=\left\{
\begin{array}{c}
-\frac{\tau_{j}^2}{f_{Z_{r}}\left(q_{\tau_{j}\bu_{r}}\right)f_{Z_{c}}\left(q_{\tau_{j}\bu_{c}}\right)}+\frac{F_{Z_{r},Z_{c}}\left(\bq_{\tau_{j},r,c},\bSigma_{Z_{r},Z_{c}}\right)}{f_{Z_{r}}\left(q_{\tau_{j}\bu_{r}}\right)f_{Z_{c}}\left(q_{\tau_{j}\bu_{c}}\right)},\quad\text{for} \quad r\neq c \nonumber\\
\frac{\tau_{j}\left(1-\tau_{j}\right)}{f_{Z_{r}}\left(q_{\tau_{j}\bu_{r}}\right)^2},\qquad\qquad\qquad\qquad\quad\qquad\qquad\qquad \text{for} \quad r=c, \nonumber
\end{array}\right.
\end{align}
and $\bSigma_{Z_{r},Z_{c}}$ denotes the variance--covariance matrix of the random variables $Z_{r}$ and $Z_{c}$ and $\mathbf{q}_{\tau_{j},r,c}=\left(q_{\tau_{j}\bu_{r}},q_{\tau_{j}\bu_{c}}\right)$, for $r,c=1,2,\dots,K$; 
\item[{\it (iii)}] given $\tau_{j}$ and $\tau_{l}$ with $j,l=1,2,\dots, s$ and $j\neq l$ and given $\bu_{s}$ and $\bu_{t}$ with $s,t=1,2,\dots,K$ and $s\neq t$, it holds
\begin{align}
\sqrt{n}\left(\hat{q}_{\tau_{j}\bu_{s}}-q_{\tau_{j}\bu_{s}}, \hat{q}_{\tau_{l}\bu_{t}}-q_{\tau_{l}\bu_{t}}\right)\xrightarrow[]{d}\mathcal{N}\left(\bO,\bet\right),\nonumber
\end{align}
as $n\rightarrow\infty$, where 
\begin{align}
\eta_{r,c}=-\frac{\tau_{j}\tau_{l}}{f_{Z_{s}}\left(q_{\tau_{j}}\right)f_{Z_{t}}\left(q_{\tau_{l}}\right)}+\frac{F_{Z_{s},Z_{t}}\left(\left(q_{\tau_{j}\bu_{s}},q_{\tau_{l}\bu_{t}}\right),\bSigma_{Z_{s},Z_{t}}\right)}{f_{Z_{s}}\left(q_{\tau_{j}}\right)f_{Z_{t}}\left(q_{\tau_{l}}\right)},\quad for \quad r\neq c. \nonumber\\
\end{align}
\end{enumerate}
\end{theorem}
\begin{proof}
See Appendix \ref{app:appendix_proofs}.
\end{proof}
\begin{remark}
The expression $a\wedge b $ stands for the minimum of $a$ and $b$. As regards the calculation of the sparsity function $s\left(\tau\right)=f\left(F^{-1}\left(\tau\right)\right)$ we refer to \cite{koenker.2005} and \cite{dominicy_veredas.2013}.
\end{remark}
\noindent To establish the asymptotic properties of the MMSQ estimates we need the following set of assumptions. 
\begin{assumption}\label{as:1}
There exists a unique/unknown true value $\vartheta_{0}\subset\Theta$ such that the sample function of projectional quantiles equal the theoretical one, provided that each quantile--based summary statistic is computed along a direction that is informative for the parameter of interest. That is $\vartheta=\vartheta_{0}\ \Leftrightarrow \ \hat{\bPhi}=\bPhi_{\vartheta_{0}}$.  
\end{assumption}
\begin{assumption}\label{as:2}
$\vartheta_{0}$ is the unique minimiser of $\left(\hat{\bPhi}-\tilde{\bPhi}_{\vartheta}^{R}\right)^{\prime}\mathbf{W}_{\vartheta}\left(\hat{\bPhi}-\tilde{\bPhi}_{\vartheta}^{R}\right)$.
\end{assumption}
\begin{assumption}\label{as:3}
Let $\hat{\bOmega}$ be the sample variance--covariance matrix of $\hat{\bPhi}$ and $\bOmega_{\vartheta}$ be the non--singular variance--covariance matrix of $\bPhi_{\vartheta}$, then  $\hat{\bOmega}$ converges to $\bOmega_{\vartheta}$ as $n$ goes to infinity. 
\end{assumption}
\begin{assumption}\label{as:5}
The matrix $\left(\frac{\partial\boldsymbol{\Phi}_{\vartheta}}{\partial\vartheta^\prime}\mathbf{W}_{\vartheta}\frac{\partial\boldsymbol{\Phi}_{\vartheta}}{\partial\vartheta}\right)$ is non--singular.
\end{assumption}
%
\noindent Under these assumptions we show the asymptotic properties of functions of quantiles. 
 \begin{theorem}\label{th:functions_quant_asy}
Under the hypothesis of Theorem \ref{th:proj_quant_asy} and assumptions \ref{as:1}--\ref{as:3}, we have 
\begin{align}
& \sqrt{n}\left(\hat{\bPhi}-\bPhi_{\vartheta}\right)\xrightarrow[]{d}\mathcal{N}\left(\bO,\bOmega_{\vartheta}\right)\nonumber\\
& \sqrt{n}\left(\tilde{\bPhi}-\bPhi_{\vartheta}\right)\xrightarrow[]{d}\mathcal{N}\left(\bO,\bOmega_{\vartheta}\right),\nonumber
\end{align} 
as $n\rightarrow\infty$, where $\bOmega_{\vartheta}=\frac{\partial\boldsymbol{\Phi}_{\vartheta}}{\partial\mathbf{q}^\prime}\boldsymbol{\eta}\frac{\partial\boldsymbol{\Phi}_{\vartheta}}{\partial\mathbf{q}}$, $\bq=\left(\bq_{\boldsymbol{\tau}_{1},\bu_{1}},\bq_{\boldsymbol{\tau}_{2},\bu_{2}},\dots,\bq_{\boldsymbol{\tau}_{K},\bu_{K}}\right)^\prime$, $\bet$ is the variance--covariance matrix of the projectional quantiles $\mathbf{q}$ defined in Theorem \ref{th:proj_quant_asy} and $\frac{\partial\boldsymbol{\Phi}_{\vartheta}}{\partial\bq}=\diag\left\{\frac{\partial\boldsymbol{\Phi}_{\vartheta}}{\partial\mathbf{q}_{\boldsymbol{\tau}_{1},\bu_{1}}},\frac{\partial\boldsymbol{\Phi}_{\vartheta}}{\partial\mathbf{q}_{\boldsymbol{\tau}_{2},\bu_{2}}},\dots,\frac{\partial\boldsymbol{\Phi}_{\vartheta}}{\partial\bq_{\boldsymbol{\tau}_{K},\bu_{K}}}\right\}$.
\end{theorem}
\begin{proof}
See Appendix \ref{app:appendix_proofs}.
\end{proof}
\noindent Next theorem shows the asymptotic properties of the MMSQ estimator.
\begin{theorem}\label{th:mmsq_asy}
Under the hypothesis of Theorem \ref{th:proj_quant_asy} and assumptions \ref{as:1}--\ref{as:5}, we have 
\begin{align}
\sqrt{n}\left(\hat{\vartheta}-\vartheta\right)\xrightarrow[]{d}\mathcal{N}\left(\bO,\left(1+\frac{1}{R}\right)\bD_{\vartheta}\mathbf{W}_{\vartheta}\bOmega_{\vartheta}\mathbf{W}_{\vartheta}^{\prime}\bD_{\vartheta}^{\prime}\right),\nonumber
\end{align}
as $n\rightarrow\infty$, where $\bD_{\vartheta}=\left(\frac{\partial\boldsymbol{\Phi}_{\vartheta}}{\partial\vartheta^\prime}\mathbf{W}_{\vartheta}\frac{\partial\boldsymbol{\Phi}_{\vartheta}}{\partial\vartheta}\right)^{-1}\frac{\partial\boldsymbol{\Phi}_{\vartheta}}{\partial\vartheta}$.
\end{theorem}
\begin{proof}
See Appendix \ref{app:appendix_proofs}.
\end{proof}
\noindent The next corollary provides the optimal weighting matrix $\mathbf{W}_\vartheta$.
%
%
\begin{corollary}
Under the hypothesis of Theorem \ref{th:proj_quant_asy} and assumptions \ref{as:1}--\ref{as:5}, the optimal weighting matrix is 
\begin{align}
\mathbf{W}_{\vartheta}^{*}=\bOmega_{\vartheta}^{-1}.\nonumber
\end{align}
Therefore, the efficient method of simulated quantiles estimator E--MMSQ has the following asymptotic distribution
\begin{align}
\sqrt{n}\left(\hat{\vartheta}-\vartheta\right)\xrightarrow[]{d}\mathcal{N}\left(\bO,\left(1+\frac{1}{R}\right)\left(\frac{\partial\boldsymbol{\Phi}_{\vartheta}}{\partial\vartheta^\prime}\bOmega_{\vartheta}^{-1}\frac{\partial\boldsymbol{\Phi}_{\vartheta}}{\partial\vartheta}\right)^{-1}\right),\nonumber
\end{align}
as $n\rightarrow\infty$.
 \end{corollary}
\section{Handling sparsity}
\label{sec:msq_sparse}
%
\noindent In this section the MMSQ estimator is extended in order to achieve sparse estimation of the scaling matrix. Specifically, the smoothly clipped absolute deviation (SCAD) $\ell_1$--penalty of \cite{fan_li.2001} is introduced into the MMSQ objective function. Formally, let $\bY\in\mathbb{R}^{m}$ be a random vector and $\bSigma=\left(\sigma_{i,j}\right)_{i,j=1}^{n}$ be its variance--covariance matrix we are interested in providing a sparse estimation of $\bSigma$. To achieve this target we adopt a modified version of the MMSQ objective function obtained by adding the SCAD penalty to the off--diagonal elements of the covariance matrix in line with \cite{bien_tibshirani.2011}. The SCAD function is a non convex penalty function with the following form
\begin{align}
p_{\lambda}\left(\vert\gamma\vert\right)=\begin{cases}\lambda \vert\gamma\vert & \mbox{if } \vert\gamma\vert\leq\lambda \\
\frac{1}{a-1}\left(a\lambda \vert\gamma\vert-\frac{\gamma^2}{2}\right)-\frac{\lambda^2}{2\left(a-1\right)} & \mbox{if }\lambda<\gamma\leq a\lambda\\
\frac{\lambda^2\left(a+1\right)}{2} & \mbox{if } a\lambda<\vert\gamma\vert,
\end{cases}
\end{align} 
which corresponds to quadratic spline function with knots at $\lambda$ and $a\lambda$. The SCAD penalty is continuously differentiable on $(-\infty; 0)\cup(0;\infty)$ but singular at 0 with its derivative equal to zero outside the range $[-a\lambda; a\lambda]$. This results in small coefficients being set to zero, a few other coefficients being shrunk towards zero while retaining the large coefficients as they are. The penalised MMSQ estimator minimises the penalised MMSQ objective function, defined as follows
\begin{align}\label{eq:s1}
\hat{\vartheta}&=\arg\min_{\vartheta}\mathcal{Q}^\star\left(\vartheta\right),
\end{align}
where $\mathcal{Q}^\star\left(\vartheta\right)=\left(\hat{\boldsymbol{\Phi}}_\mathbf{u}-\tilde{\boldsymbol{\Phi}}_{\mathbf{u},\vartheta}^{R}\right)^\prime\mathbf{W}_{\vartheta}\left(\hat{\boldsymbol{\Phi}}_{\mathbf{u}}-\tilde{\boldsymbol{\Phi}}_{\mathbf{u},\vartheta}^{R}\right)+n\sum_{i<j}p_{\lambda}\left(\vert\sigma_{ij}\vert\right)$ is the penalised distance between $\hat{\boldsymbol{\Phi}}_\mathbf{u}$ and $\tilde{\boldsymbol{\Phi}}_{\mathbf{u},\vartheta}^{R}$ and $\hat{\boldsymbol{\Phi}}_\mathbf{u},\tilde{\boldsymbol{\Phi}}_{\mathbf{u},\vartheta}^{R}$ are the quantile--based summary statistics defined in Section \ref{sec:msq_method}.
%
%
As shown in \cite{fan_li.2001}, the SCAD estimator, with appropriate choice of the regularisation (tuning) parameter, possesses a sparsity property, i.e., it estimates zero components of the true parameter vector exactly as zero with probability approaching one as sample size increases while still being consistent for the non--zero components. An immediate consequence of the sparsity property of the SCAD estimator is that the asymptotic distribution of the estimator remains the same whether or not the correct zero restrictions are imposed in the course of the SCAD estimation procedure. They call them the oracle properties.\newline 

\noindent Let 
$\vartheta_{0}=\left(\vartheta_{0}^{1}, \vartheta_{0}^{0}\right)$  be the true value of the unknown parameter $\vartheta$, where $\vartheta_{0}^{1}\in\mathbb{R}^{s}$ is the subset of non--zero parameters and $\vartheta_{0}^{0}=0\in\mathbb{R}^{k-s}$ and let $\mathcal{A}=\left\{\left(i,j\right): i<j, \sigma_{ij,0}\in\vartheta_{0}^{1}\right\}$. The following definition of oracle estimator is given in \cite{zou.2006}.
\begin{definition}
An oracle estimator $\hat{\vartheta}_{\mathrm{oracle}}$ has the following properties:
\begin{itemize}
\item[{\it (i)}] consistent variable selection: $\lim_{n\to\infty}\mathbb{P}\left(\mathcal{A}_{n}=\mathcal{A}\right)=1$, where \\$\mathcal{A}_{n}=\left\{\left(i,j\right): i<j, \hat{\sigma}_{ij}\in\hat{\vartheta}_{\mathrm{oracle}}^{1}\right\}$;
\item[{\it (ii)}] asymptotic normality: $\sqrt{n}\left(\hat{\vartheta}_{\mathrm{oracle}}^{1}-\vartheta_{0}^{1}\right)\xrightarrow[]{d}\mathcal{N}\left(\bO,\bSigma\right)$, as $n\rightarrow\infty$, where $\bSigma$ is the variance covariance matrix of $\vartheta_{0}^{1}$.
\end{itemize}
\end{definition}  
\noindent Following  \cite{fan_li.2001}, in the remaining of this section we establish the oracle properties of the penalised SCAD MMSQ estimator. We first prove the sparsity property. 
\begin{theorem}\label{th:4}
Given the SCAD penalty function $p_{\lambda}\left(\vert\sigma_{ij}\vert\right)$, for a sequence of $\lambda_{n}$ such that $\lambda_{n}\rightarrow0$, and $\sqrt{n}\lambda_{n}\rightarrow\infty$, as $n\rightarrow\infty$, there exists a local minimiser $\hat{\vartheta}$ of $\mathcal{Q}^\star\left(\vartheta\right)$ in \eqref{eq:s1} with $\|\hat{\vartheta}-\vartheta_{0}\|=\mathcal{O}_{p}\left(n^{-\frac{1}{2}}\right)$. Furthermore, we have 
\begin{align}\label{eq:s3}
\lim_{n\rightarrow\infty}\mathbb{P}\left(\hat{\vartheta}^{0}=0\right)=1.
\end{align}
\end{theorem}
\begin{proof}
See Appendix \ref{app:appendix_proofs}.
\end{proof}
\noindent The following theorem establishes the asymptotic normality of the penalised SCAD MMSQ estimator;  we denote by $\vartheta^{1}$ the subvector of $\vartheta$ that does not contain zero off--diagonal elements of the variance covariance matrix and by $\hat{\vartheta}^{1}$ the corresponding penalised MMSQ estimator. 
\begin{theorem}\label{th:5}
Given the SCAD penalty function $p_{\lambda}\left(\vert\sigma_{ij}\vert\right)$, for a sequence $\lambda_{n}\rightarrow0$ and $\sqrt{n}\lambda_{n}\rightarrow\infty$ as $n\rightarrow\infty$, then $\hat{\vartheta}^{1}$ has the following asymptotic distribution:
\begin{align}
\sqrt{n}\left(\hat{\vartheta}^{1}-\vartheta_{0}^{1}\right)\xrightarrow[]{d}\mathcal{N}\left(\bO,\left(1+\frac{1}{R}\right)\left(\frac{\partial\boldsymbol{\Phi}_{\vartheta}}{\partial{\vartheta^{1}}^\prime}\bOmega_{\vartheta_{0}^1}^{-1}\frac{\partial\boldsymbol{\Phi}_{\vartheta}}{\partial\vartheta^{1}}\right)^{-1}\right),
\end{align}
as $n\rightarrow\infty$.
\end{theorem}
\begin{proof}
See Appendix \ref{app:appendix_proofs}.
\end{proof}
%
\subsection{Algorithm}
\label{sec:algorithm}
%
\noindent The objective function of the sparse estimator is the sum of a convex function and a non convex function which complicates the minimisation procedure. Here, we adapt the algorithms proposed by \cite{fan_li.2001} and \cite{hunter_li.2005} to our objective function in order to allow a fast procedure for the minimisation problem.\newline
\indent The first derivative of the penalty function can be approximated as follows
\begin{align}
\big[p_{\lambda}\left(|\sigma_{ij}|\right)\big]^{\prime}=p_{\lambda}^{\prime}\left(|\sigma_{ij}|\right)\sgn\left(\sigma_{ij}\right)\approx\frac{p_{\lambda}^{\prime}\left(|\sigma_{ij,0}|\right)}{|\sigma_{ij,0}|}\sigma_{ij},
\end{align}
when $\sigma_{ij}\neq 0$. We use it in the first order Taylor expansion of the penalty to get 
\begin{align}
p_{\lambda}\left(|\sigma_{ij}|\right)\approx p_{\lambda}\left(|\sigma_{ij,0}|\right)+\frac{1}{2}\frac{p_{\lambda}^{\prime}\left(|\sigma_{ij,0}|\right)}{|\sigma_{ij,0}|}\left(\sigma_{ij}^2-\sigma_{ij,0}^2\right),
\end{align}
for $\sigma_{ij}\approx \sigma_{ij,0}$. 
The objective function $\mathcal{Q}^\star$ in equation \eqref{eq:s1} can be locally approximated, except for a constant term by
\begin{align}
\mathcal{Q}^\star\left(\vartheta\right)&\approx\left(\hat{\Phi}-\tilde{\Phi}_{\vartheta_{0}}^{R}\right)^{\prime}\mathbf{W}_{\bar{\theta}}\left(\hat{\Phi}-\tilde{\Phi}_{\vartheta_{0}}^{R}\right)-\frac{\partial \tilde{\Phi}_{\vartheta_{0}}^{R}}{\partial \vartheta}\mathbf{W}_{\bar{\vartheta}}\left(\hat{\Phi}-\tilde{\Phi}_{\vartheta_{0}}^{R}\right)\left(\vartheta-\vartheta_{0}\right)\nonumber\\
&\qquad+\frac{1}{2}\left(\vartheta-\vartheta_{0}\right)^{\prime}\frac{\partial \tilde{\Phi}_{\vartheta_{0}}^{R}}{\partial \vartheta}\mathbf{W}_{\bar{\vartheta}}\frac{\partial \tilde{\Phi}_{\vartheta_{0}}^{R}}{\partial \vartheta}\left(\vartheta-\vartheta_{0}\right)+\frac{n}{2}\vartheta^{\prime}\bSigma_{\lambda}\left(\vartheta_{0}\right)\vartheta,
\end{align} 
where $\bSigma_{\lambda}\left(\vartheta_{0}\right)=\diag\left\{\bO,\frac{p_{\lambda}^{\prime}\left(|\sigma_{ij,0}|\right)}{|\sigma_{ij,0}|}; i>j, \sigma_{ij,0}\in\vartheta_{0}^1 \right\}$, for which the first order condition becomes
\begin{align}
\frac{\partial \mathcal{Q}^\star\left(\vartheta\right)}{\partial\vartheta}&\approx-\frac{\partial \tilde{\Phi}_{\vartheta_{0}}^{R}}{\partial \vartheta}\mathbf{W}_{\bar{\vartheta}}\left(\hat{\Phi}-\tilde{\Phi}_{\vartheta_{0}}^{R}\right)+\frac{\partial \tilde{\Phi}_{\vartheta_{0}}^{R}}{\partial {\vartheta}^\prime}\mathbf{W}_{\bar{\vartheta}}\frac{\partial \tilde{\Phi}_{\vartheta_{0}}^{R}}{\partial \vartheta}\left(\vartheta-\vartheta_{0}\right)+n\bSigma_{\lambda}\left(\vartheta_{0}\right)\vartheta\nonumber\\
& =-\frac{\partial \tilde{\Phi}_{\vartheta_{0}}^{R}}{\partial \vartheta}\mathbf{W}_{\bar{\vartheta}}\left(\hat{\Phi}-\tilde{\Phi}_{\vartheta_{0}}^{R}\right)+\frac{\partial \tilde{\Phi}_{\vartheta_{0}}^{R}}{\partial {\vartheta}^{\prime}}\mathbf{W}_{\bar{\vartheta}}\frac{\partial \tilde{\Phi}_{\vartheta_{0}}^{R}}{\partial \vartheta}\left(\vartheta-\vartheta_{0}\right)\nonumber\\
&\qquad+n\bSigma_{\lambda}\left(\vartheta_{0}\right)\left(\vartheta-\vartheta_{0}\right)+n\bSigma_{\lambda}\left(\vartheta_{0}\right)\vartheta_{0}\nonumber\\
& =\left(\vartheta-\vartheta_{0}\right)^{\prime}\left[\frac{\partial \tilde{\Phi}_{\vartheta_{0}}^{R}}{\partial {\vartheta}^\prime}\mathbf{W}_{\bar{\vartheta}}\frac{\partial \tilde{\Phi}_{\vartheta_{0}}^{R}}{\partial \vartheta}+\bSigma_{\lambda}\left(\vartheta_{0}\right)\right]-\frac{\partial \tilde{\Phi}_{\vartheta_{0}}^{R}}{\partial \vartheta}\mathbf{W}_{\bar{\vartheta}}\left(\hat{\Phi}-\tilde{\Phi}_{\vartheta_{0}}^{R}\right)\nonumber\\
&\qquad+\bSigma_{\lambda}\left(\vartheta_{0}\right)\vartheta_{0}\nonumber\\
&=0, 
\end{align}
and therefore
\begin{align}
\vartheta &=\vartheta_{0}-\left[\frac{\partial \tilde{\Phi}_{\vartheta_{0}}^{R}}{\partial {\vartheta}^\prime}\mathbf{W}_{\bar{\vartheta}}\frac{\partial \tilde{\Phi}_{\vartheta_{0}}^{R}}{\partial \vartheta}+n\bSigma_{\lambda}\left(\vartheta_{0}\right)\right]^{-1}\nonumber\\
&\qquad\times\left[-\frac{\partial \tilde{\Phi}_{\vartheta_{0}}^{R}}{\partial \vartheta}\mathbf{W}_{\bar{\vartheta}}\left(\hat{\Phi}-\tilde{\Phi}_{\vartheta_{0}}^{R}\right)+n\bSigma_{\lambda}\left(\vartheta_{0}\right)\vartheta_{0}\right].
\end{align}
The optimal solution can be find iteratively, as follows 
\begin{align}
\vartheta^{\left(k+1\right)}&=\vartheta^{\left(k\right)}-\left[\frac{\partial \tilde{\Phi}_{\vartheta^{\left(k\right)}}^{R}}{\partial {\vartheta}^\prime}\mathbf{W}_{\bar{\vartheta}}\frac{\partial \tilde{\Phi}_{\vartheta^{\left(k\right)}}^{R}}{\partial \vartheta}+n\bSigma_{\lambda}\left(\vartheta^{\left(k\right)}\right)\right]^{-1}\nonumber\\
&\qquad\times\left[-\frac{\partial \tilde{\Phi}_{\vartheta^{\left(k\right)}}^{R}}{\partial \vartheta}\mathbf{W}_{\bar{\vartheta}}\left(\hat{\Phi}-\tilde{\Phi}_{\vartheta^{\left(k\right)}}^{R}\right)+n\bSigma_{\lambda}\left(\vartheta^{\left(k\right)}\right)\vartheta^{\left(k\right)}\right],
\end{align}
and if $\vartheta_{j}^{\left(k+1\right)}\approx 0$, then $\vartheta_{j}^{\left(k+1\right)}$ is set equal zero. When the algorithm converges the estimator satisfies the following equation
\begin{align}
-\frac{\partial \tilde{\Phi}_{\vartheta_{0}}^{R}}{\partial \vartheta}\mathbf{W}_{\bar{\vartheta}}\left(\hat{\Phi}-\tilde{\Phi}_{\vartheta_{0}}^{R}\right)+n\bSigma_{\lambda}\left(\vartheta_{0}\right)\vartheta_{0}=0,
\end{align}
that is the first order condition of the minimisation problem of the SCAD MMSQ.\newline
\indent The algorithm used above and introduced by \cite{fan_li.2001} is called local quadratic approximation (LQA). \cite{hunter_li.2005} showed that LQA applied to penalised maximum likelihood is an MM algorithm. Indeed, we define 
\begin{align}
\Psi_{|\sigma_{ij,0}|}\left(|\sigma_{ij}|\right)=p_{\lambda}\left(|\sigma_{ij,0}|\right)+\frac{1}{2}\frac{p_{\lambda}^{\prime}\left(|\sigma_{ij,0}|\right)}{|\sigma_{ij,0}|}\left(\sigma_{ij}^2-\sigma_{ij,0}^2\right),
\end{align} 
since the SCAD penalty is concave it holds
\begin{align}
\Psi_{|\sigma_{ij,0}\vert}\left(\vert\sigma_{ij}\vert\right)\geq p_{\lambda}\left(\vert\sigma_{ij}\vert\right), \quad \forall \vert\sigma_{ij}\vert,
\end{align}
and equality holds when $\vert\sigma_{ij}\vert=\vert\sigma_{ij,0}\vert$. Then $\Psi_{\vert\sigma_{ij,0}\vert}\left(\vert\sigma_{ij}\vert\right)$ majorise $p_{\lambda}\left(\vert\sigma_{ij}\vert\right)$, and it holds
\begin{align}
\Psi_{\vert\sigma_{ij,0}\vert}\left(\vert\sigma_{ij}\vert\right)<\Psi_{\vert\sigma_{ij,0}\vert}\left(\vert\sigma_{ij,0}\vert\right) \Rightarrow p_{\lambda}\left(\vert\sigma_{ij}\vert\right)<p_{\lambda}\left(\vert\sigma_{ij,0}\vert\right),
\end{align}
that is called descendent property. This feature allows us to construct an MM algorithm: at each iteration $k$ we construct $\Psi_{\vert\sigma_{ij}^{\left(k\right)}\vert}\left(\vert\sigma_{ij}\vert\right)$ and then minimize it to get $\sigma_{ij}^{\left(k+1\right)}$, that satisfies $p_{\lambda}\left(\vert\sigma_{ij}^{\left(k+1\right)}\vert\right)<p_{\lambda}\left(\vert\sigma_{ij}^{\left(k\right)}\vert\right)$. Let us consider the following
\begin{align}
S_{k}\left(\vartheta\right)=\left(\hat{\Phi}-\tilde{\Phi}_{\vartheta}^{R}\right)^{\prime}\mathbf{W}_{\bar{\theta}}\left(\hat{\Phi}-\tilde{\Phi}_{\vartheta}^{R}\right)+n\sum_{i>j}\Psi_{\vert\sigma_{ij}^{\left(k\right)}\vert}\left(\vert\sigma_{ij}\vert\right),
\end{align}
then $S_{k}\left(\vartheta\right)$ majorise $\mathcal{Q}^\star\left(\vartheta\right)$; thus we only need to minimise $S_{k}\left(\vartheta\right)$, that can be done as explained above. \cite{hunter_li.2005} proposed an improved version of LQA for penalised maximum likelihood, aimed at avoiding to zero out the parameters too early during the iterative procedure. We present their method applied to SCAD MMSQ as follows 
\begin{align}
p_{\lambda,\epsilon}\left(|\sigma_{ij}|\right)&= p_{\lambda}\left(|\sigma_{ij}|\right)-\epsilon\int_{0}^{\vert\sigma_{ij}\vert}\frac{p_{\lambda}^{\prime}\left(\vert\sigma_{ij,0}\vert\right)}{\epsilon+t}dt\nonumber\\
\mathcal{Q}^\star_{\epsilon}\left(\vartheta\right)&=\left(\hat{\Phi}-\tilde{\Phi}_{\vartheta}^{R}\right)^{\prime}\mathbf{W}_{\bar{\vartheta}}\left(\hat{\Phi}-\tilde{\Phi}_{\vartheta}^{R}\right)+n\sum_{i>j}p_{\lambda, \epsilon}\left(\vert\sigma_{ij}\vert\right)\nonumber\\
\Psi_{\vert\sigma_{ij,0}\vert,\epsilon}\left(\vert\sigma_{ij}\vert\right)&=p_{\lambda,\epsilon}\left(\vert\sigma_{ij,0}\vert\right)+\frac{p_{\lambda}^{\prime}\left(\vert\sigma_{ij,0}\vert\right)}{2\left(\epsilon+\vert\sigma_{ij,0}\vert\right)}\left(\sigma_{ij}^2-\sigma_{ij,0}^2\right)\nonumber\\
S_{k,\epsilon}\left(\vartheta\right)&=\left(\hat{\Phi}-\tilde{\Phi}_{\vartheta}^{R}\right)^{\prime}\mathbf{W}_{\bar{\vartheta}}\left(\hat{\Phi}-\tilde{\Phi}_{\vartheta}^{R}\right)+n\sum_{i>j}\Psi_{\vert\sigma_{ij}^{\left(k\right)}\vert,\epsilon}\left(\vert\sigma_{ij}\vert\right),\nonumber
\end{align}
where $\bar{\vartheta}$ is a consistent estimator of $\vartheta$.
They proved that as $\epsilon\downarrow 0$ the perturbed objective function $\mathcal{Q}^\star_{\epsilon}\left(\vartheta\right)$ converges uniformly to the  not perturbed one $\mathcal{Q}^\star\left(\vartheta\right)$ and that if $\hat{\vartheta}_{\epsilon}$ is a minimiser of $\mathcal{Q}^\star_{\epsilon}\left(\vartheta\right)$ then any limit point of the sequence $\left\{\hat{\vartheta}_{\epsilon}\right\}_{\epsilon\downarrow 0}$ is a minimiser of $\mathcal{Q}^\star\left(\vartheta\right)$. This construction allows to define $\Psi_{\vert\sigma_{i,j}^{\left(k\right)}\vert,\epsilon}\left(\vert\sigma_{i,j}\vert\right)$ even when $\sigma_{i,j}^{\left(k\right)}\approx 0$. The authors also provided a way to choose the value of the perturbation $\epsilon$ and suggested the following
\begin{align}
\epsilon=\frac{\tau}{2np_{\lambda}^{\prime}\left(0\right)}\min\left\{\vert\sigma_{i,j}^{\left(0\right)}\vert:\sigma_{i,j}^{\left(0\right)}\neq 0\right\},
\end{align}
with the following tuning constant $\tau=10^{-8}$.
%
\subsection{Tuning paramenter selection}
\label{sec:msq_tuning_sel}
%
The SCAD penalty requires the selection of two tuning parameters $\left(a,\lambda\right)$. The first tuning parameter is fixed at $a=3.7$ as suggested in \cite{fan_li.2001}, while the parameter $\lambda$ is selected using as validation function
\begin{align}
V\left(\lambda\right) = \frac{1}{n}\left(\hat{\Phi}-\tilde{\Phi}_{\hat{\vartheta}_{\lambda}}^{R}\right)\mathbf{W}_{\hat{\vartheta}_{\lambda}}\left(\hat{\Phi}-\tilde{\Phi}_{\hat{\vartheta}_{\lambda}}^{R}\right),
\end{align}
where $\hat{\vartheta}_{\lambda}$ denotes the parameters estimate when $\lambda$ is selected as tuning parameter. We choose $\lambda^{\ast}=\arg\min_{\lambda} V\left(\lambda\right)$; the minimisation is performed over a grid of values for $\lambda$.\newline
\indent An alternative approach is the $K$--fold cross validation, in which the original sample is divided in $K$ subgroups $T_{k}$, called folds. The validation function is    
\begin{align}
CV\left(\lambda\right) = \sum_{k=1}^{K}\frac{1}{n_{k}}\left(\hat{\Phi}-\tilde{\Phi}_{\hat{\vartheta}_{\lambda,k}}^{R}\right)\mathbf{W}_{\hat{\vartheta}_{\lambda,k}}\left(\hat{\Phi}-\tilde{\Phi}_{\hat{\vartheta}_{\lambda,k}}^{R}\right),
\end{align}
where $\hat{\vartheta}_{\lambda,k}$ denotes the parameters estimate on the sample  $\left(\cup_{i=1}^K T_{k}\right)\setminus T_{k}$ with $\lambda$ as tuning parameter. Then the optimal value is chosen as $\lambda^{*}=\arg\min_{\lambda} CV\left(\lambda\right)$; again the minimisation is performed over a grid of values for $\lambda$.  
%
\subsection{Implementation}
\label{sec:msq_implementation}
%
\noindent The symmetric and positive definiteness properties of the variance--covariance matrix should be preserved at each step of the optimisation process. Preserving those properties is a difficult task since the constraints that ensure the definite positiveness of a matrix are non linear. Therefore, we consider an implementation that is similar to the Graphical Lasso algorithm introduced by \cite{friedman_etal.2008}. We outline the steps of the algorithm below. Let $\bOmega$ be a correlation matrix of dimension $n\times n$ and partition $\bOmega$ as follows
\begin{align}
\bOmega = \left[\begin{matrix}
\bOmega_{11} & \bomega_{12}\\
\bomega_{12}^{\prime} & 1\\
\end{matrix}\right],
\end{align}
where $\bOmega_{11}$ is a matrix of dimension $(n-1)\times(n-1)$ and $\bomega_{12}$ is a vector of dimension $n-1$, and consider the transformation consider the transformation $\bomega_{12}^\star\to\frac{\hat{\boldsymbol{\omega}}_{12}}{1+\hat{\boldsymbol{\omega}}_{12}^{\prime}\boldsymbol{\Omega}_{11}^{-1}\hat{\boldsymbol{\omega}}_{12}}$ where $\hat{\boldsymbol{\omega}}_{12}$ is obtained by applying
a step of the Newton--Raphson algorithm to $\bomega_{12}$ as follows
\begin{align}
\hat{\bomega}_{12}&=\bomega_{12}-\left[\frac{\partial \tilde{\Phi}_{\bomega_{12}}^{R}}{\partial {\bomega_{12}}^\prime}\mathbf{W}_{\bar{\bomega}_{12}}\frac{\partial \tilde{\Phi}_{\bomega_{12}}^{R}}{\partial \bomega_{12}}+n\bSigma_{\lambda}\left(\bomega_{12}\right)\right]^{-1}\nonumber\\
&\qquad\qquad\qquad\times\left[-\frac{\partial \tilde{\Phi}_{\bsigma_{12}}^{R}}{\partial \bsigma_{12}}\mathbf{W}_{\bar{\bsigma}_{12}}\left(\hat{\Phi}-\tilde{\Phi}_{\bsigma_{12}}^{R}\right)+n\bSigma_{\lambda}\left(\bomega_{12}\right)\bomega_{12}\right].
\end{align}
Once we update the last column, we shift the next to the last at the end and repeat the steps described above. We repeat this procedure until convergence. 
%
\section{Synthetic data examples}
\label{sec:msq_synthetic_ex}
%
\noindent As mentioned in the introduction the Stable distribution plays an interesting role in modelling multivariate data. Its peculiarity of heaving heavy tailed properties and its closeness under summation make it appealing in the financial contest. Despite its characteristics, estimation of parameters has been always challenging and this feature greatly limited its use in applied works requiring simulation--based methods. In this section we briefly introduce the multivariate Elliptical Stable distribution (ESD) previously considered by \cite{lombardi_veredas.2009}. 
%
\subsection{Multivariate Elliptical Stable distribution}
\label{sec:esd_multiv}
%
\noindent 
%
A random vector $\bY\in \mathbb{R}^m$ is elliptically distributed if 
\begin{align}
\bY=^{d}\bxi+\mathcal{R}\bGamma\bU,
\end{align}
where $\bxi\in\mathbb{R}^{m}$ is a  vector of location parameters, $\bGamma$ is a matrix such that $\bOmega= \bGamma\bGamma^{\prime}$ is a $m\times m$ full rank matrix of scale parameters, $\bU\in\mathbb{R}^{m}$ is a random vector uniformly distributed in the unit sphere $\mathbb{S}^{m-1}=\left\{\bu\in\mathbb{R}^m\::\:\bu^\prime\bu=1\right\}$ and $\mathcal{R}$ is a non--negative random variable stochastically independent of $\bU$, called generating variate of $\bY$.\newline
\indent If $\mathcal{R}=\sqrt{Z_1}\sqrt{Z_2}$ where $Z_1\sim\chi_{m}^2$ and $Z_2\sim\mathcal{S}_{\frac{\alpha}{2}}\left(\xi,\omega,\delta\right)$ is a positive Stable distributed random variable with kurtosis parameter equal to $\frac{\alpha}{2}$ for $\alpha\in\left(0,2\right]$, location parameter $\xi=0$, scale parameter $\omega=1$ and asymmetry parameter $\delta=1$, stochastically independent of $\chi_{m}^2$, then the random vector $\bY$ has Elliptical Stable distribution, i.e., $\bY\sim\mathcal{ESD}_{m}\left(\alpha,\bxi,\bOmega\right)$,
%
with characteristic function 
\begin{align}
\psi_\bY\left(\bt\right)&=\mathbb{E}\left(\exp\left\{i\bt^\prime\bY\right\}\right)\nonumber\\
&=\exp\left\{i\bt^\prime\bxi-\left(\bt^\prime\bOmega\bt\right)^{\frac{\alpha}{2}}\right\}.
\end{align}
See \cite{samorodnitsky_etal.1994} for more details on the positive Stable distribution and \cite{nolan.2013} for the recent developments on multivariate elliptically contoured stable distributions.\newline
\indent Among the properties that the class of elliptical distribution possesses, the most relevant are the closure with respect to affine transformations, conditioning and marginalisation, see \cite{fang_etal1990}, \cite{embrechts_etal.2005} and \cite{mcneil_etal.2015} for further details. Simulating from an ESD is straightforward, indeed let $\bar\omega_\alpha=\left(\cos\frac{\pi\alpha}{4}\right)^{\frac{2}{\alpha}}$, then $\bY\sim\mathcal{ESD}_m\left(\alpha,\bxi,\bOmega\right)$ if and only if $\bY$ has the following stochastic representation as a scale mixture of Gaussian distributions
\begin{align}
\bY=\bxi+\zeta^{\frac{1}{2}}\bX,
\end{align} 
where $\zeta\sim\mathcal{S}_{\frac{\alpha}{2}}\left(0,\bar\omega_\alpha,1\right)$ and $\bX\sim\mathcal{N}\left(\bO, \bOmega\right)$  independent of $\zeta$. Following the Proposition 2.5.2 of \cite{samorodnitsky_etal.1994}, the characteristic function of $\bY$ is
\begin{align}
\psi_\bY\left(\bt\right)&=\mathbb{E}\left(\exp\left\{i\bt^\prime\bY\right\}\right)\nonumber\\
&=\mathbb{E}_\zeta\mathbb{E}\left(\exp\left\{i\bt^\prime\bxi+i\zeta^{\frac{1}{2}}\bt^\prime\bX\right\}\mid \zeta\right)\nonumber\\
&=\mathbb{E}_\zeta\mathbb{E}\left(\exp\left\{i\bt^\prime\bxi-\frac{\zeta\bt^\prime\bOmega\bt}{2}\right\}\mid \zeta\right)\nonumber\\
&=\exp\left\{i\bt^\prime\bxi-\left(\frac{1}{2}\right)^{\frac{\alpha}{2}}\left(\bt^\prime\bOmega\bt\right)^{\frac{\alpha}{2}}\right\},\quad\qquad \alpha\neq1,
\label{eq:esd_char_fun}
\end{align}
which is the characteristic function of an Elliptical Stable distribution with scale matrix $\bOmega/2$. The last equation follows the fact that the Laplace transform of $\zeta\sim\mathcal{S}_{\frac{\alpha}{2}}\left(0,\bar\omega_\alpha,1\right)$ with $0<\alpha\leq 2$ is
\begin{align}
\psi^\ast_\zeta\left(A\right)&=\mathbb{E}\left(\exp\left\{-A\zeta\right\}\right)\nonumber\\
&=\begin{cases}
\exp\left\{-\frac{\left(\bar{\omega}_\alpha\right)^{\frac{\alpha}{2}}}{\cos\frac{\pi\alpha}{4}}A^{\frac{\alpha}{2}}\right\},&\alpha\neq 1\\
\exp\left\{\frac{2\bar{\omega}_\alpha}{\pi}A\log\left(A\right)\right\},&\alpha=1.
\end{cases}
\end{align}
The Elliptical Stable distribution is a particular case of multivariate Stable distribution so it admits finite moments if $\mathbb{E}\left[\zeta^p\right]<\infty$ for $p<\alpha$. For $\alpha\in\left(1,2\right)$, $\mathbb{E}\left(\zeta^{\frac{1}{2}}\right)<\infty$, so that by the law of iterated expectations $\mathbb{E}\left(\bY\right)=\boldsymbol{\xi}$, while the second moment never exists. Except for few cases, $\alpha=2$ (Gaussian), $\alpha=1$ (Cauchy) and $\alpha=\frac{1}{2}$ (L\'evy), the density function cannot be represented in closed form. Those characteristics of the Stable distribution motivate the use of simulations methods in order to make inference on the parameters of interest.  
%
\subsection{How to choose optimal directions}
\label{sec:msq_direction}
%
Before we turn to illustrate our simulation framework, we should solve an important issue related to the application of the MMSQ that concerns the choice of the directions. Indeed, the easiest solution is to choose an equally spaced grid of directions, an approach that would be computational expensive. Therefore, we choose optimal directions $\bu^*$ according to the following definition \ref{def:1} which allows to maximise the information contained in the chosen measure.  
\begin{definition}
\label{def:1}
Let us consider a given parameter of interest $\vartheta^\star\subset\Theta_k\in\mathbb{R}^k$ and consider the subset $\bY^\star=(Y_1^\star,\dots,Y^\star_l,\dots,Y^\star_h)$ of $h$ variables of $\bY \in\mathbb{R}^{m}$ assumed to be informative for the parameter $\vartheta^\star$, and the projectional quantile  $q^{\tau\bu}$ of $\bY^\star$ at a given $\tau$, with $\bu\in\mathbb{S}^{h-1}$. An optimal direction $\bu^{*}\in\mathbb{S}^{m-1}$ for $\bY^\star$ is defined as the vector whose $i$--th coordinate is 
\begin{align}
u_{i}^{*}=
\bigg\{
\begin{array}{lr}
u_{\max,l} & if \ Y_{i}=Y^\star_l\\
0 & otherwise,
\end{array}\nonumber
\end{align}
where $u_{\max,l}$ is the $l$--th coordinate of the vector
\begin{align}\label{eq:opt_dir}
\bu_{\max}\in \left\{\arg\max_{\bu\in\mathbb{S}^{h-1}}q^{\tau\bu}\right\}.
\end{align}
\end{definition} 
\noindent If for example, $h=2$, then the optimal direction is  
\begin{align}
\bu^*=\left(0,\dots,u_{\max,1}, 0,\dots,0, u_{\max,2},\dots, 0\right),\nonumber
\end{align} 
where $u_{\max,1}$ and $u_{\max,2}$ are the $i$--th and $j$--th coordinate respectively, which is informative for the covariances between $Y_i$ and $Y_j$. The optimal solutions defined in \eqref{eq:opt_dir} are computed using the Lagrangian function as follows
\begin{align}
\mathcal{L}\left(\bu,\lambda\right)=q_{\tau\bu}-\lambda\left(\|\bu\|-1\right),\nonumber
\end{align}
by solving $\nabla \mathcal{L}\left(\bu,\lambda\right)=0$, where $\nabla$ stands for the gradient. This equation can be solved analytically, for instance when $m=h=2$ for ESD distribution as shown in section \ref{sec:simulated_examples_ell_stable}, or numerically.\newline

\noindent Let $\bU^{*}$ collect all the optimal solutions $\bu_{j}^{*}$ for $q^{\tau_{j}\bu}, \ j=1,2,\dots,s$ and all the canonical directions and let 
\begin{align}
\bPhi^{\boldsymbol{\tau},\mathbf{u}^\ast}_{\vartheta}&=\left(\bPhi^{\mathbf{u}_1\otimes\boldsymbol{\tau}^\prime}_{\vartheta},\bPhi^{\mathbf{u}_2\otimes\boldsymbol{\tau}^\prime}_{\vartheta},\dots,\bPhi^{\mathbf{u}_K\otimes\boldsymbol{\tau}^\prime}_{\vartheta}\right)^{\prime}\in\mathbb{R}^{B}\nonumber\\
\tilde{\bPhi}_{\mathbf{u}^\ast,\vartheta}^{R}&=\left(\tilde{\bPhi}_{\bu_{1}^{*},\vartheta}^{R},\tilde{\bPhi}_{\bu_{2}^{*},\vartheta}^{R},\dots, \tilde{\bPhi}_{\bu_{K}^{*},\vartheta}^{R}\right)^{\prime}\in\mathbb{R}^{B}\nonumber\\
\hat{\bPhi}_{\mathbf{u}^\ast}&=\left(\hat{\bPhi}_{\bu_{1}^{*}},\hat{\bPhi}_{\bu_{2}^{*}},\dots,\hat{\bPhi}_{\bu_{K}^{*}}\right)^{\prime}\in\mathbb{R}^{B},\nonumber
\end{align}
where $K$ is the cardinality of $\bU^{*}$, $B=\sum_{i=1}^{K}b_{i}$ and $b_{i}$ is the dimension of $\Phi_{\bu_{i},\vartheta}$ for $i=1,2,\dots, K$, then the MMSQ minimises the square distance defined in equation \eqref{eq:mmsq_min_problem} between $\hat{\bPhi}_{\mathbf{u}^\ast}$ and $\tilde{\bPhi}_{\mathbf{u}^\ast,\vartheta}^{R}$ along the optimal directions $\bU^\ast$.
%
\subsection{Simulation results}
\label{sec:simulated_examples_results}
%
\noindent In this Section we consider simulation examples for the ESD distribution $\bY\thicksim\mathcal{ESD}_{m}\left(\alpha, \bxi, \bOmega\right)$ as defined in section \ref{sec:esd_multiv}. In order to apply the MMSQ, we first need to select the quantile--based measures which are informative for each of the parameters of interest $\left(\alpha,\boldsymbol{\xi},\boldsymbol{\Omega}\right)$ where the shape parameter $\alpha\in\left(0,2\right)$ controls for the tail behaviour of the distribution, while $\boldsymbol{\xi}\in\mathbb{R}^m$ and $\boldsymbol{\Omega}$ denote the location parameter and the positive definite $m\times m$ scaling matrix, respectively. Since the quantile--based measures should be informative for the corresponding parameter, we select for $\alpha$ a measure related to the kurtosis of the distribution, for the locations the median and for the elements of the scaling matrix we opt for a measure of dispersion, and all the measures will be calculated along appropriately chosen directions, as it will be discussed later in this section. Summarising, for kurtosis, location and scale parameters we choose respectively
\begin{align}
\kappa_\bu&=\frac{q_{0.95,\bu}-q_{0.05,\bu}}{q_{0.75,\bu}-q_{0.25\bu}}\nonumber\\
m_\bu&= q_{0.5,\bu}\nonumber\\
\varsigma_\bu&= q_{0.75,\bu}-q_{0.25\bu}\nonumber,
\end{align} 
where $\bu\in\mathcal{S}^{m-1}$ defines a relevant direction. Next, we need to identify the optimal directions. To this end we can consider the relevant properties of the ESD. Specifically, as shown for example by \cite{embrechts_etal.2005}, the ESD is closed under marginalisation, i.e., $Y_{i}\thicksim\mathcal{ESD}_{1}\left(\alpha, \xi_{i}, \omega_{ii}\right)$, for $i=1,2,\dots,m$, where $\omega_{ii}$ is the $i$--th element of the main diagonal of the matrix $\boldsymbol{\Omega}$. By exploiting the closure with respect to marginalisation, we can conclude that the optimal directions for the shape parameter $\alpha$, for the locations $\xi_i$ and for the diagonal elements of the scale matrix $\omega_{ii}$, for $i=1,2,\dots,m$ are the canonical directions. It still remains to consider the optimal directions for the off--diagonal elements of the scale matrix $\omega_{ij}$, with $i,j=1,2,\dots,m$ and $i\neq j$. Again we exploit the closure with respect to marginalisation. Specifically, let $\bZ_{ij}=\left(Y_{i}, Y_{j}\right)$, then $\bZ_{ij}\thicksim\mathcal{ESD}_{2}\left(\alpha, \bxi_{ij}, \bOmega_{ij}\right)$, where
\begin{align}
\bxi_{ij}=\left(\xi_{i},\xi_{j}\right)^\prime,\qquad
\bOmega_{ij}=\begin{pmatrix}\omega_{ii} & \omega_{ij}\\ \omega_{ij} & \omega_{jj}  \end{pmatrix}.\nonumber
\end{align}
Moreover, let $\bu\in \mathbb{S}^{1}$ and $Z_{ij,\bu}=\bu^{\prime}\bZ_{ij}$ be the projection of $\bZ_{ij}$ along $\bu$, then $Z_{ij,\bu}\thicksim\mathcal{ESD}_{1}\left(\alpha, \bu^{\prime}\bxi_{ij}, \bu^{\prime}\bOmega_{ij}\bu\right)$, (see \citealt{embrechts_etal.2005}), from which we have the following representation of the projected ESD random variable
\begin{align}
\label{eq:rap_esd}
Z_{ij,\bu} = \bu^{\prime}\bxi_{ij}+\sqrt{\bu^{\prime}\bOmega_{ij}\bu}Z,
\end{align}    
where $Z\thicksim\mathcal{ESD}_{1}\left(\alpha,0,1\right)$. Following Definition \ref{def:1}, in order to find the optimal directions we need to compute 
\begin{align}
\bu_{\max}=\arg\max_{\bu\in\mathbb{S}^{1}}q_{\tau\bu}\left(\bZ_{ij}\right),
\end{align}
where $q_{\tau\bu}\left(\bZ_{ij}\right)$ is the projectional quantile of $\bZ_{ij}$, i.e., the $\tau$--th level quantile of the random variable $Z_{ij,\bu}$. Exploiting representation \eqref{eq:rap_esd},
it holds
\begin{align}
\bu_{\max}=\arg\max_{\bu\in\mathbb{S}^{1}}\bu^{\prime}\bxi_{ij}+\sqrt{\bu^{\prime}\bOmega_{ij}\bu},
\end{align}
which is a quadratic optimisation problem that can be solved using the method of Lagrangian multiplier, as follows
\begin{align}
\mathcal{L}\left(\bu, \lambda\right)=\bu^{\prime}\bxi_{ij}+\sqrt{\bu^{\prime}\bOmega_{ij}\bu}-\lambda\left(\|\bu\|-1\right).
\end{align}
The solution requires to set to zero the gradient of the Lagrangian $\nabla \mathcal{L}\left(\bu, \lambda\right)=0$, that is  
\begin{align}
& \frac{\partial \mathcal{L}}{\partial u_{1}}=\frac{\left(\omega_{ii}^2u_{1}+\omega_{ij} u_{2}\right)}{\sqrt{\omega_{ii}^2u_{1}^2+\omega_{jj}^2u_{2}^2+2\omega_{ij}u_{1}u_{2}}}-2\lambda u_{1}=0\nonumber\\
& \frac{\partial \mathcal{L}}{\partial u_{2}}=\frac{\left(\omega_{jj}^2u_{2}+\omega_{ij}u_{1}\right)}{\sqrt{\omega_{ii}^2u_{1}^2+\omega_{jj}^2u_{2}^2+2\omega_{ij}u_{1}u_{2}}}-2\lambda u_{2}=0\nonumber\\
\label{eq:lagrangian_fd_3}
& \frac{\partial \mathcal{L}}{\partial \lambda} = u_{1}^2+u_{2}^2-1=0,
\end{align}
and from the first two equations, we obtain
\begin{align}
& u_{2}\left(\sigma_{1}^2u_{1}+\omega_{ij} u_{2}\right)-u_{1}\left(\omega_{jj}^2u_{2}+\omega_{ij}u_{1}\right)=0\nonumber\\
& u_{2}^2+u_{2}u_{1}\frac{\omega_{ii}^2-\omega_{jj}^2}{\omega_{ij}}-u_{1}^2=0\nonumber\\
& u_{2}=\frac{u_{1}}{2}\left(-\frac{\omega_{ii}^2-\omega_{jj}^2}{\omega_{ij}}\pm\sqrt{\left(\frac{\omega_{ii}^2-\omega_{jj}^2}{\omega_{ij}}\right)^2+4}\right).\nonumber
\end{align}
By inserting the previous expression for $u_2$ into equation \eqref{eq:lagrangian_fd_3}, we solve for $u_1$
\begin{align}
& u_{1}^2+\frac{u_{1}^2}{4}\left(-\frac{\omega_{ii}^2-\omega_{jj}^2}{\omega_{ij}}\pm\sqrt{\left(\frac{\omega_{ii}^2-\omega_{jj}^2}{\omega_{ij}}\right)^2+4}\right)^2=1\nonumber\\
& u_{1}^2\left[1+\frac{1}{4}\left(-\frac{\omega_{ii}^2-\omega_{jj}^2}{\omega_{ij}}\pm\sqrt{\left(\frac{\omega_{ii}^2-\omega_{jj}^2}{\omega_{ij}}\right)^2+4}\right)^2\right]=1\nonumber\\
& u_{1}=\pm\frac{1}{\sqrt{\left[1+\frac{1}{4}\left(-\frac{\omega_{ii}^2-\omega_{jj}^2}{\omega_{ij}}\pm\sqrt{\left(\frac{\omega_{ii}^2-\omega_{jj}^2}{\omega_{ij}}\right)^2+4}\right)^2\right]}},
\end{align}
where the sign of $u_1$ depends on the sign of $\omega_{ij}$. The optimal direction $\bu_{\max}$ is then plugged into $\bu^{\ast}=\left(0,\dots,u_{1,\max}, \dots, u_{2,\max}, \dots,0\right)$ as explained in Definition \ref{def:1}.\newline 

\noindent To illustrate the effectiveness of the MMSQ we replicate the simulation study considered in \cite{lombardi_veredas.2009}. Specifically, we consider two dimensions of the random vector $\bY$, $m=2,5$ and, for each dimension, we consider three values of the shape parameters $\alpha=\left(1.7,1.9,1.95\right)$, while the location parameter $\bxi$ is always set to zero and the scale matrices are 
\begin{align}\label{eq:esd2_mat1}
\bSigma^s_2=\begin{pmatrix}0.5 & 0.9\\
0.9 & 2
\end{pmatrix},
\end{align}
for $m=2$, and 
\begin{align}\label{eq:esd2_mat2}
\bSigma^s_5=\begin{pmatrix}
0.25 & 0.25 & 0.4 & 0 & 0\\
0.25 & 0.5 & 0.4 & 0 & 0\\
0.4 & 0.4 & 1 & 0 & 0\\
0 & 0 & 0 & 2 & 2.55\\
0 & 0 & 0 & 2.55 & 4  
\end{pmatrix},
\end{align}
for $m=5$. We also consider two different sample sizes $n=500,2000$ and we fix $R=200$. In Table \ref{tab:table_coverage_esd_dim2}--\ref{tab:table_coverage_esd_dim5}, we report estimation results obtained over $1,000$ replications for $m=2$ and $m=5$, with $n=500,2000$, for three different values of the characteristic exponent $\alpha=\left(1.7,1.9,1.95\right)$. Specifically, each table reports the bias (BIAS), the standard error (SSD) and the empirical coverage probability (ECP) of the estimated parameters. Our results show that the MMSQ estimator is always unbiased, indeed the BIAS is always less than $0.25$ in dimension $m=2$ and less that $0.15$ in dimension $m=5$. The SSDs are always small, in particular for $n=500$ it is always less then $0.5$. The empirical coverages are always in line with their expected values for all but the diagonal elements of the scale matrix $\sqrt{\omega_{ii}}$ for $i=1,2,\dots,m$ for which they display lower values than expected, which means that in those cases the asymptotic standard errors are underestimated.\newline 
%
%

\noindent We also illustrate the performance of the Sparse--MMSQ method and compare it with three alternative methods on two simulation examples. The first example considers a sample of $n=500$ observations from a Elliptical Stable distribution of dimension $m=12$, with locations at zero, four different values of the characteristic exponent $\alpha=\left(1.70,1.90,1.95,2.00\right)$ and scale matrix $\bSigma^s_{12}$ equal to that considered in \cite{wang.2015}.
The second simulated example considers a sample of $n=800$ observations from the Elliptical Stable distribution of dimension $27$ with location and characteristic exponent chosen as before and block--diagonal scale matrix $\boldsymbol{\Sigma}^s_{27}=\diag\left\{\boldsymbol{\Sigma}^s_{12},\boldsymbol{\Sigma}^s_{15}\right\}$ and $\boldsymbol{\Sigma}^s_{15}$ is the correlation matrix in Section 4.2 of \cite{wang.2010}. We compare the Sparse--MMSQ with three alternative algorithms: the graphical LASSO (GLASSO) of \cite{friedman_etal.2008}, the graphical LASSO with SCAD penalty (SCAD), the graphical adaptive Lasso (Adaptive Lasso)  of \cite{fan_etal.2009}. The main aim of the proposed simulation example is to compare the performance of the different algorithm for different levels of deviations from the Gaussian assumption which represents the benchmark assumption for the competing algorithms. Results are reported in Table \ref{tab:table_frob_mmsq_dim_12_27} in terms of average Frobenius norm, F1--Score and Kullback--Leibler (KL) divergence over 100 replications and their standard deviations.
The $F_1-\textrm{score}$, see \cite{baldi_etal.2000} assesses the performance of the algorithm by computing
\begin{align}
F_1-\textrm{score}=\frac{2TP}{2TP+FP+FN},
\end{align}
where $TP$, $FP$ and $FN$ are the true positives, false positives and false negatives. The $F_1-\textrm{score}$ lies between $0$ and $1$, where $1$ stands for perfect identification and $0$ for bad identification. The KL divergence computes the divergence between the true $\bOmega$ and the estimated $\hat{\bOmega}$ scale matrices as
\begin{align}
KL\left(\bOmega,\hat{\bOmega}\right)=\frac{1}{2}\left(\trace\left(\bOmega^{-1}\hat{\bOmega}\right)-m-\log\left(\frac{\vert\hat{\bOmega}\vert}{\vert\bOmega\vert}\right)\right),
\end{align}
while the Frobenius norm is the usual matrix norm $\|\bOmega\|_F=\sqrt{\sum_{i,j=1}^m\omega_{ij}^2}$. The Sparse--MMSQ method performs very well with respect to the alternatives in terms of $F_1-\textrm{score}$ for all the considered values of the characteristic exponent $\alpha$. This means that the method correctly identifies the sparse structure of the matrices regardless the amount of the deviation from the Gaussian assumption. This results is confirmed by visual inspection of Figures \ref{fig:smmsq_12_alpha_170_190}--\ref{fig:smmsq_27_alpha_200} reporting the band structure of the true and estimated matrices averaged across the $100$ replications. The Sparse--MMSQ method does a good job also in terms of Frobenius norm but only in  dimension $m=27$. The worst results are reported by the Sparse--MMSQ in terms of KL divergence. A possible explanation for those results would be that maximum likelihood methods essentially minimise the KL divergence, therefore reported values for the alternative methods are the minimum obtainable.
%
%
\begin{table}[!t]\captionsetup{font={footnotesize}}\setlength{\tabcolsep}{5pt}\begin{center}\resizebox{0.95\columnwidth}{!}{\begin{tabular}{lrrrrrrrr}\\\toprule$\alpha$ & 1.70  & 1.90 & 1.95 & 2.00 & 1.70  & 1.90 & 1.95 & 2.00   
\\\cmidrule(lr){1-1}\cmidrule(lr){2-2}\cmidrule(lr){3-3}\cmidrule(lr){4-4}\cmidrule(lr){5-5} \cmidrule(lr){6-6}\cmidrule(lr){6-6}\cmidrule(lr){7-7}\cmidrule(lr){8-8}\cmidrule(lr){9-9} {\bf Frobenius norm}&\multicolumn{4}{c}{{\it Dimension 12}}&\multicolumn{4}{c}{{\it Dimension 27}}\\ 
\cmidrule(lr){2-5}\cmidrule(lr){6-9} 
\multirow{ 2}{*}{GLasso} & 1.595 & 1.0392 & 0.81693 & 0.59958 & 4.5938 & 2.6358 & 1.8644 & 0.76058 \\ & {\it (0.5232)}  & {\it (0.47659)}  & {\it (0.34976)}  & {\it (0.074508)}  & {\it (3.0661)} & {\it (2.071)} & {\it (1.5321)} & {\it (0.045129)}\\ 
\multirow{ 2}{*}{SCAD} &1.5043 & 0.94056 & 0.7378 & 0.58001 & 4.5847 & 2.5318 & 1.7361 & 0.56438 \\ 
& {\it(0.56176) } &{\it (0.51448)}  & {\it(0.39827)}  & {\it(0.11528)}  & {\it(3.3211)} & {\it(2.1012)} & {\it(1.5937) }&{\it(0.063328)}\\ \multirow{ 2}{*}{Adaptive Lasso}&1.4486&0.90578&0.69566&0.50441&4.084&2.2872&1.7195&0.65269\\ 
&{\it (0.5416) } &{\it(0.46181) } &{\it(0.34298)  }&{\it(0.087179) } & {\it(3.2848)} &{\it(1.6282) }& {\it(1.5373) }&{\it(0.047185)}\\ 
\multirow{ 2}{*}{S--MMSQ}&1.6618&1.4111&1.293&1.2417&2.6987&2.449&2.3426&2.1677\\ 
&{\it (0.21718) } &{\it(0.22563) } &{\it(0.22635) } &{\it(0.22013) } &{\it (0.2791) }&{\it (0.26377)} & {\it(0.28864) }&{\it(0.24305)}\\ 
\hline 
{\bf $F_1$--score}&\multicolumn{4}{c}{{\it Dimension 12}} &\multicolumn{4}{c}{{\it Dimension 27}}\\
\cmidrule(lr){2-5}\cmidrule(lr){6-9} 
\multirow{ 2}{*}{GLasso} & 0.1313&0.012143&0.019025&0&0.037952&0.007118&0.0036548&0 \\ 
& {\it(0.23919)}  &{\it(0.079663)}  &{\it(0.1103)}  &{\it(0)}  & {\it(0.10075)} &{\it (0.07118)} & {\it(0.036548)} &{\it(0)}\\ 
\multirow{ 2}{*}{SCAD}&0.26295&0.17153&0.15148&0.23174&0.033123&0.0085093&0.0036548&0.0015072 \\ 
&{\it (0.27865) }  &{\it(0.22789) }  &{\it(0.21994)}  &{\it(0.23612) }  &{\it (0.095177) } &{\it (0.072389)}  &{\it (0.036548) } &{\it(0.015072)} \\ 
\multirow{ 2}{*}{Adaptive Lasso}&0.2431&0.080443&0.057361&0.037187&0.13042&0.040525&0.0075655&0 \\ 
&{\it (0.33484) }  &{\it(0.17254)}   &{\it(0.1628) }  &{\it(0.10126)  } & {\it(0.23048)}  & {\it(0.15814)}  & {\it(0.075655)}  &{\it(0)} \\ 
\multirow{ 2}{*}{S--MMSQ}&0.40246&0.55827&0.62059&0.69567&0.83754&0.75499&0.71847&0.66897 \\ 
&{\it (0.17051)}   &{\it(0.14057) }  &{\it(0.13682)}   &{\it(0.089005)}   & {\it(0.097355)}  & {\it(0.086734) } & {\it(0.079755)}  &{\it(0.048205)} \\ 
\hline 
{\bf KL} & \multicolumn{4}{c}{{\it Dimension 12}} &\multicolumn{4}{c}{{\it Dimension 27}}\\\cmidrule(lr){2-5}  \cmidrule(lr){6-9} 
\multirow{ 2}{*}{GLasso} & 0.68981&0.29197&0.18876&0.10059&6.6643&2.3116&0.98558&0.17044 \\ 
& {\it (0.36107)}   &{\it (0.25353)}   &{\it (0.17321)}   &{\it (0.024998) }  &{\it  (8.8661)}  & {\it (4.2476) } &{\it  (1.7369)}  &{\it (0.021347)} \\ 
\multirow{ 2}{*}{SCAD}&0.63751&0.24506&0.16588&0.09049&6.8927&2.2701&0.92517&0.095768 \\ 
&{\it (0.39392)}  &{\it(0.24673)}  &{\it(0.19988)}  &{\it(0.03358)}  & {\it(8.9943)} & {\it(4.3379)} &{\it (1.8981) }&{\it(0.018791)}\\ 
\multirow{ 2}{*}{Adaptive Lasso}&0.58807&0.2294&0.14527&0.0735&6.627&2.3228&0.96203&0.13577 \\ 
& {\it(0.34298)}  &{\it(0.2109)}  &{\it(0.15541) } &{\it(0.022154)}  & {\it(8.9975) }&{\it (4.6305) }& {\it(2.0405)} &{\it(0.020124)}\\ 
\multirow{ 2}{*}{S--MMSQ}&0.96549&0.77602&0.67512&0.64992&58.4657&53.6626&51.8209&48.6645 \\ 
& {\it(0.20521)}  &{\it(0.22501)}  &{\it(0.21598) } &{\it(0.21598)}  & {\it(7.8325)} & {\it(9.2006)} &{\it (8.7965) }&{\it(8.7965)}\\ 
\bottomrule \end{tabular}} 
\caption{\footnotesize{Frobenius norm, F1--Score and Kullbach--Leibler information between the true scale matrix of the Elliptical Stable distribution and the matrices estimated by alternative methods: the Graphical Lasso of \cite{friedman_etal.2008} (GLasso), the graphical model with SCAD penalty (SCAD), the graphical model with adaptive Lasso of \cite{fan_etal.2009} (Adaptive Lasso) and the S--MMSQ. The measures  are evaluated over 100 replications, we report the mean and the variances in brackets.}} 
\label{tab:table_frob_mmsq_dim_12_27} 
\end{center} 
\end{table} 
%
%
\section{Application to systemic risk}
\label{sec:application_port_opt}
%
\indent In this Section we first introduce the NetCoVaR risk measure that extends the CoVaR approach of \cite{adrian_brunnermeier.2011,adrian_brunnermeier.2016} to account for multiple contemporaneous distress instances. Then we introduce the NetCoVaR dominance test that extends that proposed by \cite{castro_ferrari.2014} to a parametric framework where asset returns are assumed to follow an Elliptical Stable distribution and we describe how it can be used to build a network measuring the tail dependence among institutions. The Elliptical Stable distribution plays a relevant role in systemic risk assessment either because data are contaminated by the presence of outliers or because the methodology strongly relies on the presence of heavy--tailed distributions. Finally, we apply the risk measure and the risk measurement framework to a real dataset of US financial institutions covering the recent global financial crisis of 2008.
%
\subsection{The NetCoVaR risk measure}
\label{sec:sys_risk_network_risk_measure}
%
%
\noindent The Conditional or Comovement Value--at--Risk (CoVaR) has been introduced in the systemic risk literature by \cite{adrian_brunnermeier.2011,adrian_brunnermeier.2016}, and subsequently extended to a parametric framework by \cite{girardi_ergun.2013}. The CoVaR measures the spillover effects between institutions by providing information on the Value--at--Risk of an institution or market, conditional on another institution's distress event. Formally, the $CoVaR^\tau_{j\vert M}$ of institution $j$ belonging to a given market and the market itself $M$ is defined as the VaR of the market $M$ conditional to the institution $j$ being at its VaR level
\begin{equation}
\label{eq:covar_AB_def}
\mathbb{P}\left(X_M\leq VaR_M^\tau\mid X_j=VaR_j^\tau\right)=\tau,
\end{equation}
where $X_j$ and $X_M$ denote the profits and loss of the institution $j$ and the market $M$ respectively, and $VaR_j^\tau$ and $VaR_M^\tau$ denote the individual $\tau$--level VaRs of the institution $j$ and the market, respectively.\newline
\indent Despite its relevance for measuring the impact of a distress event affecting one institution on the overall financial market, the previous definition of CoVaR suffers from two main drawbacks. First, the CoVaR in equation \eqref{eq:covar_AB_def} is not monotonically increasing as a function of the correlation between $X_j$ and $X_M$. As a consequence, it does not preserve the stochastic ordering induced by the bivariate distribution  for the couple of random variables $\left(X_j,X_M\right)$, see, e.g., \cite{mainik_schaanning.2014}, \cite{bernardi_etal.2017} and \cite{bernardi_etal.2017b} for an exhaustive discussion. Second, and more importantly, it only considers the impact of an extreme event affecting an institution on another institution or a market index, failing to account for the presence of any other institution belonging to the same market. Several papers try to overcome this problem by introducing systemic risk measures that account for multiple contemporaneous distress events and investigated their theoretical properties, see, e.g., \cite{bernardi_etal.2017d}, \cite{salvadori_etal.2016} and \cite{bernardi_etal.2016b}. Here, we follow along the same line provided in \cite{bernardi_etal.2016c} and we measure how the distress of one institution affects the overall health of all the remaining ones. Formally, let $\tau\in\left(0,1\right)$ be a confidence level, and let $j$ denote an institution belonging to a given market with $d$ participants, $i=1,2,\dots,d$, then the network CoVaR of institution $j$, denoted by $NetCoVaR_{j}^{\tau}$, satisfies the following equation 
\begin{equation}
\mathbb{P}\left(\sum_{i=1}^dY_{i}\leq NetCoVaR_{j}^{\tau}\mid Y_{j}\leq VaR_{j}^{\tau}\right)=\tau,
\label{eq:netcovar_def}
\end{equation}
for $j=1,2,\dots,d$, where $VaR_{j}^{\tau}$ denotes the marginal Value--at--Risk (VaR) of institution $j$ such that $\mathbb{P}\left(Y_{j}\leq{\rm VaR}_{j}^{\tau}\right)=\tau$. The NetCoVaR of institution $j$ defined in equation \eqref{eq:netcovar_def}, is the quantile of the distribution of the random variables $S=\sum^d_{i=1}Y_{i}$ conditional on an extreme event affecting the return of institution $j$, $Y_{j}$, where such an extreme event is defined as $Y_{j}$ being below its VaR at confidence level $\tau$. The calculation of the NetCoVaR requires the prior evaluation of institution's $j$ marginal VaR and, conditional on $VaR_{j}^{\tau}$, the $NetCoVaR_{j}^{\tau}$ is calculated as the value of $s^\ast$ such that
\begin{equation*}
\mathbb{P}\left(\sum_{i=}^dY_{i}\leq s^\ast,Y_{j}\leq VaR_{j}^{\tau}\right)=\tau^2,
\end{equation*}
for $j=1,2,\dots,d$. Our definition of NetCoVaR in equation \eqref{eq:netcovar_def} is substantially different from that originally introduced by \cite{adrian_brunnermeier.2011} and it overcomes the deficiencies of the original definition mentioned above. Moreover, the closure of the ESD with respect to linear combinations and marginalisation is extremely helpful in evaluating the NetCoVaR.\newline 
\indent The NetCoVaR in equation \eqref{eq:netcovar_def} only provides a point estimate of the systemic impact of institution $j$. A further improvement would be to provide a system of hypothesis to test the systemic dominance of one institution over another one. To this end, \cite{castro_ferrari.2014} recently proposed the following system of hypothesis
\begin{equation}
\label{eq:netcovar_dom_test}
\begin{cases}
H_0\::\:CoVaR_j^\tau=CoVaR_k^\tau\\
H_1\::\:CoVaR_j^\tau\neq CoVaR_k^\tau,
\end{cases}
\end{equation}
for any $j,k=1,2,\dots,d$ with $j\neq k$. Here, we consider a similar dominance test where the CoVaR risk measured is substituted by our NetCoVaR.
The next proposition provides the asymptotic distribution of $NetCoVaR_j^\tau$, for $j=1,2,\dots,d$ which is useful to calculate the asymptotic distribution of the test statistic to perform the NetCoVaR dominance test in equation \eqref{eq:netcovar_dom_test}. Although our results easily extend to any Elliptical distribution, in what follows we consider Elliptically Stable distributed random variables.  
%
%
\begin{proposition}
\label{prop:netcovar_asy_dist}
Let $\bY\sim\mathcal{ESD}_d\left(\alpha,\bmu,\bOmega\right)$, with $\bY=\left(Y_1,Y_2,\dots,Y_d\right)$, then 
\begin{equation}
\left(\begin{matrix}
S\\Y_j\\Y_k
\end{matrix}\right)\sim\mathcal{ESD}_3\left(\alpha,\tilde{\bmu},\tilde{\bOmega}\right),
\end{equation}
where $\tilde{\bmu}=\left(\boldsymbol{\iota}_d^\prime\bmu,\mu_j,\mu_k\right)^\prime$ with
\begin{align}
\tilde{\bOmega}=
\left[\begin{matrix}
\sigma^2_{S} &  \sigma_{S,Y_j}& \sigma_{S,Y_k}\\
\star& \sigma^2_{Y_j}& \sigma_{Y_j,Y_k}\\
\star&\star&\sigma^2_{Y_k}
\end{matrix}\right],
\end{align}
and $\sigma^2_{Y_l}=\omega_{l,l}$, $\sigma^2_{S}=\boldsymbol{\iota}_l^\prime\bOmega\boldsymbol{\iota}_l$, $\sigma_{S,Y_l}=\sigma_{Y_l}^2+\sum_{\substack{s=1\\ s\neq l}}^d\omega_{s,l}$ and $\sigma_{Y_j,Y_k}=\omega_{j,k}$, for $l=j,k$. Furthermore, let $Z_{l}= S\mid Y_{l}\leq VaR_{l}^{\tau}$, for $l=j,k$, then the NetCoVaR$_{j}^{\tau}$ and NetCoVaR$_{k}^{\tau}$ are the $\tau$--level quantile of those variables, that is, they satisfy the following relations $\mathbb{P}\left(Z_{l}\leq NetCoVaR_{l}^{\tau}\right)=\tau$, for $l=j,k$, and
\begin{align}
\sqrt{n}\left(q_{Z_{j}}^{\tau}\left(\hat{\vartheta}\right)-q_{Z_{j}}^{\tau}\left(\vartheta\right),q_{Z_{k}}^{\tau}\left(\hat{\vartheta}\right)-q_{Z_{k}}^{\tau}\left(\vartheta\right) \right)\rightarrow\mathcal{N}\left(\bO, \hat{\bSigma}\right),
\end{align}
where $q_{Z_{l}}^{\tau}\left(\hat{\vartheta}\right)=NetCoVaR_{l}^{\tau}\left(\hat{\vartheta}\right)$, for $l=j,k$, with
\begin{align}
\hat{\bSigma}= \left[\begin{matrix}\nabla_{\vartheta} q_{Z_{j}}^{\tau}\left(\hat\vartheta\right)\\
\nabla_{\vartheta} q_{Z_{k}}^{\tau}\left(\hat\vartheta\right)\end{matrix}\right]^{\prime}
\hat{\bOmega}
\left[\begin{matrix}\nabla_{\vartheta} q_{Z_{j}}^{\tau}\left(\hat\vartheta\right)\\
\nabla_{\vartheta} q_{Z_{k}}^{\tau}\left(\hat\vartheta\right)\end{matrix}\right],
\end{align}
where $\bOmega$ is the asymptotic variance covariance matrix of the vector of estimated parameters $\hat{\vartheta}$ and $\nabla_{\vartheta} q_{Z_{l}}^{\tau}\left(\hat\vartheta\right)$, for $l=j,k$ is the Jacobian of the conditional quantile transformation evaluated at $\hat{\vartheta}$. Therefore, the asymptotic distribution of $NetCoVaR_j^\tau\left(\hat{\vartheta}\right)-NetCoVaR_k^\tau\left(\hat{\vartheta}\right)$ becomes
\begin{align}
NetCoVaR_j^\tau \left(\hat{\vartheta}\right)-NetCoVaR_k^ \tau\left(\hat{\vartheta}\right) \sim \mathcal{N}\left(NetCoVaR_j^\tau \left(\vartheta\right)-NetCoVaR_k^\tau \left(\vartheta\right), \sigma_{jk}^2\right),
\end{align}
where $\sigma_{jk}^{2} = \nabla_{\vartheta} q_{Z_{j}}^{\tau}\left(\hat\vartheta\right)^{\prime}\hat{\bOmega} \nabla_{\vartheta} q_{Z_{k}}^{\tau}\left(\hat\vartheta\right)$.
\end{proposition}
\begin{proof}
The distribution of the random vector $\left(S,Y_j,Y_k\right)^\prime$ follows immediately by applying the closure property of the multivariate Elliptical Stable distribution, while the asymptotic distribution of $NetCoVaR_j^\tau\left(\hat{\vartheta}\right)-NetCoVaR_k^\tau\left(\hat{\vartheta}\right)$ follows from Theorem \ref{th:mmsq_asy} and the application of the Delta method.
\end{proof}
\begin{remark}
The asymptotic distribution of the NetCoVaR in Proposition \eqref{prop:netcovar_asy_dist} requires the evaluation of the Jacobian of the quantile transformation of the distribution of the random variables $Z_{l}= S\mid Y_{l}\leq VaR_{l}^{\tau}$, for $l=j,k$. Under the assumption that $\bY\sim\mathcal{ESD}_d\left(\alpha,\bmu,\bOmega\right)$, with $\bY=\left(Y_1,Y_2,\dots,Y_d\right)$, then the density function $f_{Z_{j}}\left(q_{Z_{j}}^{\tau}\right)$  is easily derived as follows 
\begin{align}
f_{Z_{j}}\left(z_{j}\right) = f_{S\mid Y_{j}\leq VaR_{j}^{\tau}}\left(z_{j}\right) = \frac{f_{S, Y_{j}}\left(s, y_{j}\leq VaR_{j}^{\tau}\right)}{\int_{-\infty}^{VaR_{j}^{\tau}}f_{Y_{j}}\left(y\right)dy}= \frac{\int_{-\infty}^{VaR_{j}^{\tau}}f_{S,Y_{j}}\left(s, y_{j}\right)dy_{j}}{\tau},
\end{align}
where $\tau=\int_{-\infty}^{VaR_{l}^{\tau}}f_{S,Y_{l}}\left(s, y_{l}\right)dy_{l}$. Exploiting again the closure under linear transformation property of the Elliptical Stable distribution, the joint density of the bivariate vector $\left(S,Y_l\right)^\prime$ is $\left(S,Y_l\right)^\prime\sim\mathcal{ESD}_2\left(\alpha,\bmu_{S,Y_l},\bOmega_{S,Y_l}\right)$, with $\bmu_{S,Y_l}=\left(\boldsymbol{\iota}_d^\prime\bmu,\mu_l\right)^\prime$ with
\begin{align}
\bOmega_{S,Y_l}=
\left[\begin{matrix}
\sigma^2_{S} &  \sigma_{S,Y_l}\\
\sigma_{S,Y_l}& \sigma^2_{Y_l}
\end{matrix}\right],
\end{align}
for $l=j,k$. Therefore, the density of $Z_l$ for $l=j,k$ is Extended Skew Elliptical, i.e., $Z_l\sim\mathcal{ESES}_1\left(\alpha,\mu_S,\sigma^2_S,\lambda_0,\lambda_1\right)$, with $\mu_S = \boldsymbol{\iota}_d^\prime\boldsymbol{\mu}$, $\sigma^2_S = \boldsymbol{\iota}_d^\prime\boldsymbol{\Omega}\boldsymbol{\iota}_d$, $\lambda_0=\frac{VaR_l^\tau-\mu_{Y_l}}{\sigma_{Y_l}\sqrt{1-\delta^2}}$, $\lambda_1=\frac{\delta}{\sqrt{1-\delta^2}}$ and $\delta=\frac{\sigma_{S,Y_l}}{\sqrt{\sigma^2_{S}\sigma^2_{Y_l}}}$, for $l=j,k$. See the supplementary material accompanying the paper for an analytical definition of the Extended Skew Elliptical Stable distribution, the probability and cumulative density function, the generating mechanism and the characteristic function.
\end{remark}
%
\subsection{Empirical application and results}
\label{sec:sys_risk_empirical_app}
%
\begin{table}[!h] 
\captionsetup{font={footnotesize}} 
\begin{center} \resizebox{0.9\columnwidth}{!}{ 
\begin{tabular}{llccllcc}\\ 
\toprule 
\multirow{2}{*}{Name} & \multirow{2}{*}{Ticker} & Date of & Date of & \multirow{2}{*}{Name} & \multirow{2}{*}{Ticker}& Date of  & Date of \\ 
&& first observation	& last observation	& && first observation	 &last observation\\
\hline 
\multirow{2}{*}{ZIONS BANCORP} &\multirow{2}{*}{ZIO} & \multirow{2}{*}{02/01/73} & \multirow{2}{*}{30/08/17} & \multirow{2}{*}{\bf\color{red} LEHMAN BROS. HDG.} &\multirow{2}{*}{\bf\color{red} LEHM} & \multirow{2}{*}{\bf\bf\color{red} 02/05/94} & {\bf\bf\color{red} Default}  \\ 
 	& &&	 & 	& 	 & 	 & {\bf\bf\color{red} 01/03/12} \\ 
WELLS FARGO \& CO. &WFG& 02/01/73 & 30/08/17 & LEGG MASON & LEGG & 01/08/83 & 30/08/17 \\
\multirow{2}{*}{\bf WACHOVIA} & \multirow{2}{*}{\bf WAC} & \multirow{2}{*}{\bf 02/01/73}  & {\bf Acquisition} & \multirow{2}{*}{KEYCORP} &\multirow{2}{*}{ KEY }& \multirow{2}{*}{02/01/73} & \multirow{2}{*}{30/08/17}\\ 
& &	& {\bf 30/12/08}	& 	& 	 &\\
US BANCORP & USB & 16/10/84 & 30/08/17 & JP MORGAN CHASE \& CO. &JPM &02/01/73 & 30/08/17 \\
TRAVELERS COS. &TRAV& 02/01/73 & 30/08/17 & FRANKLIN RESOURCES &FRA& 06/01/75 & 30/08/17 \\ 
SYNOVUS FINANCIAL& SYN& 06/01/75 & 30/08/17 & FANNIE MAE &FAN& 02/01/73 & 30/08/17 \\
 \multirow{2}{*}{SUNTRUST BANKS} &\multirow{2}{*}{SUN}&  \multirow{2}{*}{01/07/85} &  \multirow{2}{*}{30/08/17} & \multirow{2}{*}{\bf COUNTRYWIDE FINL.} &\multirow{2}{*}{\bf COU}& \multirow{2}{*}{\bf 31/01/75} & {\bf Acquisition}  \\ 
 & 	& 	& 	&& &	 &{\bf 27/06/08}\\
\multirow{2}{*}{\bf SAFECO} &\multirow{2}{*}{\bf SAF}& \multirow{2}{*}{\bf 02/01/73 }& {\bf Acquisition} &  \multirow{2}{*}{COMERICA} & \multirow{2}{*}{COM}& \multirow{2}{*}{02/01/73} & \multirow{2}{*}{30/08/17} \\
& 	& 	&{\bf 19/09/08} & 	& 	&\\
REGIONS FINL. & REG& 02/01/73 & 30/08/17 & CITIZENS FINANCIAL GROUP &CIT & 24/09/14 & 30/08/17 \\ 
PNC FINL. SVS. GP. &PNC & 02/01/73 & 30/08/17 & CITIGROUP &CTG & 29/10/86 & 30/08/17 \\
NORTHERN TRUST & NORT& 02/01/73 & 30/08/17 & CINCINNATI FINL. & CIN& 02/01/73 & 30/08/17 \\
\multirow{2}{*}{\bf NATIONAL CITY} &\multirow{2}{*}{\bf NTC}&\multirow{2}{*}{\bf 01/05/73} & {\bf Acquisition}  & \multirow{2}{*}{CIGNA} &\multirow{2}{*}{CIG} & \multirow{2}{*}{01/04/82} & \multirow{2}{*}{ 30/08/17} \\
  & &	& {\bf 29/12/08}   & 	  & 	&	\\ 
\multirow{2}{*}{MORGAN STANLEY} & \multirow{2}{*}{MS}& \multirow{2}{*}{23/02/93} & \multirow{2}{*}{30/08/17} & \multirow{2}{*}{\bf CHUBB} & \multirow{2}{*}{\bf CHU}& \multirow{2}{*}{ \bf 02/01/73} & {\bf Acquisition} \\ 
 & 	& 	& &	& &	 &{\bf 13/01/16 }\\
\multirow{2}{*}{\bf MERRILL LYNCH \& CO.}& \multirow{2}{*}{\bf ML}& \multirow{2}{*}{\bf 02/01/73} &  {\bf Acquisition}  &\multirow{2}{*}{\bf\bf\color{red} BEAR STEARNS} &\multirow{2}{*}{\bf\bf\color{red} BST} & \multirow{2}{*}{\bf\bf\color{red} 29/10/85} & {\bf\bf\color{red} Default}  \\ 
& 	&& {\bf 30/12/08} & 	& &	& {\bf \bf\color{red}29/05/08} \\
\multirow{2}{*}{\bf MELLON FINL.}& \multirow{2}{*}{\bf MEL}& \multirow{2}{*}{\bf 02/01/73 }& {\bf Acquisition} & \multirow{2}{*}{BANK OF NEW YORK MELLON } &\multirow{2}{*}{BNY} & \multirow{2}{*}{02/01/73} & \multirow{2}{*}{ 30/08/17} \\ 
& 	& & {\bf 28/06/07} & 	 & 	 & \\
\multirow{2}{*}{\bf MARSHALL \& ILSLEY} &\multirow{2}{*}{\bf M\&I} & \multirow{2}{*}{\bf 02/01/73} & {\bf Acquisition} &  \multirow{2}{*}{BANK OF AMERICA} &\multirow{2}{*}{BOA} &  \multirow{2}{*}{02/01/73} &  \multirow{2}{*}{30/08/17} \\
& 	& & {\bf 04/07/11} & 	 & 	 & \\ 
MARSH \& MCLENNAN & M\&M& 02/01/73 & 30/08/17 & AMERICAN INTL.GP. & INTL& 02/01/73 & 30/08/17 \\
LOEWS & LOE& 02/01/73 & 28/08/17 & AFLAC &AFLAC & 23/08/73 & 30/08/17 \\
LINCOLN NATIONAL &LIN & 02/01/73 & 30/08/17 \\
\bottomrule \end{tabular}} \end{center} 
\caption{\footnotesize{Name and classifications of the $37$ US financial institutions belonging to the Standard \& Poor's 500 Composite Index (S\&P500). Most of the institutions have been excluded because of the limited length of their return series. Wachovia, Countrywide financial, National City, Merrill Lynch \& Co., Mellon Financial, Marshall Isley, Chubb and Safeco Corp. {\it (denoted in bold)} have been acquired by Wells Fargo, Bank of America, PNC Financial Services, Bank of America, Bank of New York, BMO Financial Group, ACE Limited and Liberty Mutual Group, respectively, while Lehman Bros. Holding and Bear Stearns {\it (denoted in bold red)} defaulted before the end of the sample period. For those institutions, the date of death is reported in the third and last columns.}}
\label{tab:USBanks_list}
\end{table}
%
%
\begin{table}[!ht]
\captionsetup{font={footnotesize}}
\begin{center}
\resizebox{0.9\columnwidth}{!}{%
  \begin{tabular}{lcccccccc}
\toprule
   Name & Min & Max & Mean & Std. Dev. & Skewness & Kurtosis & 1\% Str. Lev. & JB\\
    \hline
ZIONS  & -36.879  & 64.239  & 0.125  & 5.843  & 0.790  & 21.045  & -17.798  & 16624.661  \\
WELLS FARGO  & -36.780  & 48.184  & 0.174  & 4.757  & 0.364  & 23.106  & -12.426  & 20508.003  \\
{\bf WACHOVIA}  & -62.861  & 30.805  & -0.115  & 4.510  & -3.158  & 50.817  & -14.763  & 117870.710  \\
US BANCORP & -48.394  & 43.084  & 0.206  & 4.617  & -0.354  & 23.551  & -14.439  & 21424.088  \\
TRAVELERS  & -24.069  & 20.422  & 0.152  & 3.759  & 0.089  & 7.791  & -9.846  & 1164.465  \\
SYNOVUS  & -46.536  & 40.870  & 0.080  & 5.714  & -0.418  & 12.997  & -19.637  & 5098.829  \\
SUNTRUST  & -40.583  & 49.270  & 0.073  & 5.448  & -0.047  & 21.324  & -17.007  & 17012.268  \\
{\bf SAFECO}  & -19.371  & 35.193  & 0.076  & 3.052  & 1.330  & 22.108  & -8.608  & 18857.084  \\
REGIONS  & -39.151  & 52.695  & 0.004  & 5.779  & 0.707  & 21.326  & -18.283  & 17117.767  \\
PNC  & -39.004  & 41.924  & 0.125  & 4.668  & -0.165  & 19.643  & -11.665  & 14040.468  \\
NORTHERN TRUST & -22.142  & 21.475  & 0.179  & 4.037  & -0.052  & 6.484  & -11.602  & 615.502  \\
{\bf NATIONAL CITY}  & -56.247  & 40.547  & -0.164  & 4.295  & -2.929  & 52.434  & -11.837  & 125552.976  \\
MORGAN STANLEY & -90.465  & 68.693  & 0.145  & 6.678  & -1.041  & 41.996  & -17.335  & 77269.027  \\
{\bf MERRILL LYNCH}  & -52.670  & 54.824  & 0.020  & 5.291  & -0.479  & 34.155  & -13.309  & 49223.945  \\
{\bf MELLON}  & -16.055  & 15.901  & 0.128  & 2.894  & -0.169  & 8.688  & -8.287  & 1644.785  \\
{\bf MARSHALL \& ILSLEY} & -40.020  & 53.348  & -0.003  & 5.321  & 0.602  & 27.196  & -19.615  & 29735.332  \\
MARSH \& MCLENNAN & -45.404  & 23.564  & 0.136  & 3.727  & -1.350  & 24.197  & -9.816  & 23134.233  \\
LOEWS & -31.657  & 20.380  & 0.152  & 3.606  & -0.591  & 12.500  & -9.407  & 4643.964  \\
LINCOLN  & -81.372  & 76.955  & 0.101  & 6.743  & -0.312  & 45.347  & -14.865  & 90878.439  \\
{\bf\color{red} LEHMAN } & -280.885  & 109.861  & -0.414  & 12.931  & -8.953  & 205.640  & -33.647  & 2096767.740  \\
LEGG MASON & -57.563  & 34.090  & 0.169  & 5.575  & -0.832  & 16.525  & -15.588  & 9408.928  \\
KEYCORP & -61.427  & 40.280  & 0.013  & 5.500  & -1.136  & 27.486  & -15.735  & 30639.284  \\
JP MORGAN  & -41.684  & 39.938  & 0.171  & 5.089  & -0.159  & 14.323  & -12.307  & 6500.922  \\
FRANKLIN  & -27.834  & 24.178  & 0.193  & 4.789  & -0.167  & 6.804  & -12.702  & 738.884  \\
FANNIE MAE & -225.271  & 135.239  & -0.163  & 12.483  & -3.408  & 109.981  & -28.016  & 582232.354  \\
{\bf COUNTRYWIDE}  & -28.531  & 25.033  & 0.010  & 4.428  & -0.581  & 11.051  & -14.695  & 3352.449  \\
COMERICA & -31.744  & 32.776  & 0.109  & 4.841  & -0.138  & 10.328  & -13.821  & 2724.400  \\
CITIZENS  & -31.511  & 36.900  & 0.031  & 6.066  & 0.197  & 7.083  & -16.128  & 852.359  \\
CITIGROUP & -92.632  & 78.798  & 0.020  & 7.011  & -1.483  & 54.834  & -15.161  & 136574.867  \\
CINCINNATI  & -27.415  & 17.646  & 0.138  & 3.331  & -0.497  & 12.228  & -8.626  & 4364.415  \\
CIGNA & -47.279  & 31.701  & 0.269  & 5.002  & -1.553  & 21.293  & -13.120  & 17442.571  \\
{\bf CHUBB}  & -20.601  & 26.328  & 0.154  & 3.441  & 0.646  & 12.318  & -9.338  & 4483.892  \\
{\bf\color{red} BEAR STEARNS}  & -161.613  & 59.262  & -0.041  & 6.624  & -13.038  & 319.850  & -10.691  & 5121070.902  \\
BANK OF NEW YORK MELLON & -24.609  & 26.004  & 0.162  & 4.404  & -0.047  & 6.931  & -11.726  & 783.219  \\
BANK OF AMERICA & -59.288  & 60.671  & 0.049  & 5.953  & -0.286  & 31.198  & -15.104  & 40302.134  \\
AMERICAN INTL & -114.843  & 92.426  & -0.110  & 7.814  & -1.364  & 71.030  & -20.375  & 234866.237  \\
AFLAC & -48.560  & 31.735  & 0.226  & 4.523  & -0.831  & 19.937  & -11.469  & 14673.793  \\
\bottomrule
\end{tabular}}
\end{center}
\caption{\footnotesize{Summary statistics US financial institutions in the panel, for the period form May 6, 1994 till August 25, 2017. The eight column, denoted by ``1\% Str. Lev.'' is the 1\% empirical quantile of the returns distribution, while the last column, denoted by ``JB'' is the value of the Jarque--Ber\'a test--statistics. Institutions that experienced distress instances are denoted in bold, see Table \ref{tab:USBanks_list}.}}
\label{tab:USBanks_summary_stat}
\end{table}
%
\noindent This section illustrates the practical utility of using the Sparse--MMSQ for building a 
network based on the NetCoVaR dominance test discussed in the previous section. Specifically, we apply the methodology to analyse the US financial system during the period of the Global Financial Crisis. The data considered are weekly log--returns of $37$ financial institutions belonging to the Standard \& Poor's 500 Composite Index (S\&P500), from May 15, 1997 through February 15, 2017. Table \ref{tab:USBanks_list} provides the list of the institutions included in the sample, the tickers, the date of the first available observation and the date of the last observation. 
Eight of the institutions included in the list, e.g., Wachovia, Countrywide financial, National City, Merrill Lynch \& Co., Mellon Financial, Marshall Isley, Chubb and Safeco Corp. experienced distress instances during the period. They have been acquired by Wells Fargo, Bank of America, PNC Financial Services, Bank of America, Bank of New York, BMO Financial Group, ACE Limited and Liberty Mutual Group, respectively. Moreover, we also included Bear Stearns and Lehman Bros. Holding that defaulted on May 2008 and March 2012, respectively. Their inclusion is motivated by the desire to increase understanding of how the proposed NetCoVaR risk measurement framework behaves in practice when two different types of extreme events affect the institutions: acquisition and default.
Table \ref{tab:USBanks_summary_stat} reports the descriptive statistics of the weekly returns of the financials institutions over the whole sampling period. Interestingly, the $1\%$ stress levels in the penultimate column of Table \ref{tab:USBanks_summary_stat} point out that individual risk measures, like the VaR, sometimes fail to detect the systemic relevance of the institutions. For example, Bear Stearns reported a $1\%$ stress level of about $-11$ in line with most of the other institutions that did not experienced bankruptcy.\newline
\indent Here, we would examine whether stock market co--movements have changed over time, with a focus on the period of the recent global financial crisis. Specifically, we estimate the NetCoVaR over two different periods before and after the global financial crisis of 2008, under the assumption of Stable Elliptical returns using the Sparse--MMSQ. The first period begins the May 6, 1997 till the end of December 2007 while the second is from January 4, 2008 to August 25, 2017. Estimation results are not reported to save space but are available upon request to the first author. Then we employed the NetCoVaR systemic dominance test detailed in the previous Section to construct a graph in which the vertexes represent companies and an edges between two vertexes stands for interconnection between the two institution as measured by a p--value greater than $0.05$. The number of edges can be interpreted as the degree of connectedness among stock return series, therefore an increase of this number implies that the market tends to be more integrated and consequently tends to have a higher systemic risk. Figure \ref{fig:netcovar_newtork} reports the result of our analysis. Specifically, Figure \ref{fig:netcovar_newtork_ESD_prior_gfc} reports the network estimated at the end of December 2007 assuming returns are jointly modelled using the Elliptical Stable distribution. For comparison, Figure \ref{fig:netcovar_newtork_GAUSS_prior_gfc} reports the NetCoVaR graphs estimated under the assumption of Normally distributed returns and the parameters are estimated using the GLASSO algorithm of \cite{friedman_etal.2008} where the penalty parameter has been chosen by cross validation.
%
\begin{sidewaysfigure}[p]
\captionsetup{font={footnotesize}}
\begin{center}

\subfloat[ESD, prior to GFC]{\label{fig:netcovar_newtork_ESD_prior_gfc}\includegraphics[width=0.5\textwidth]{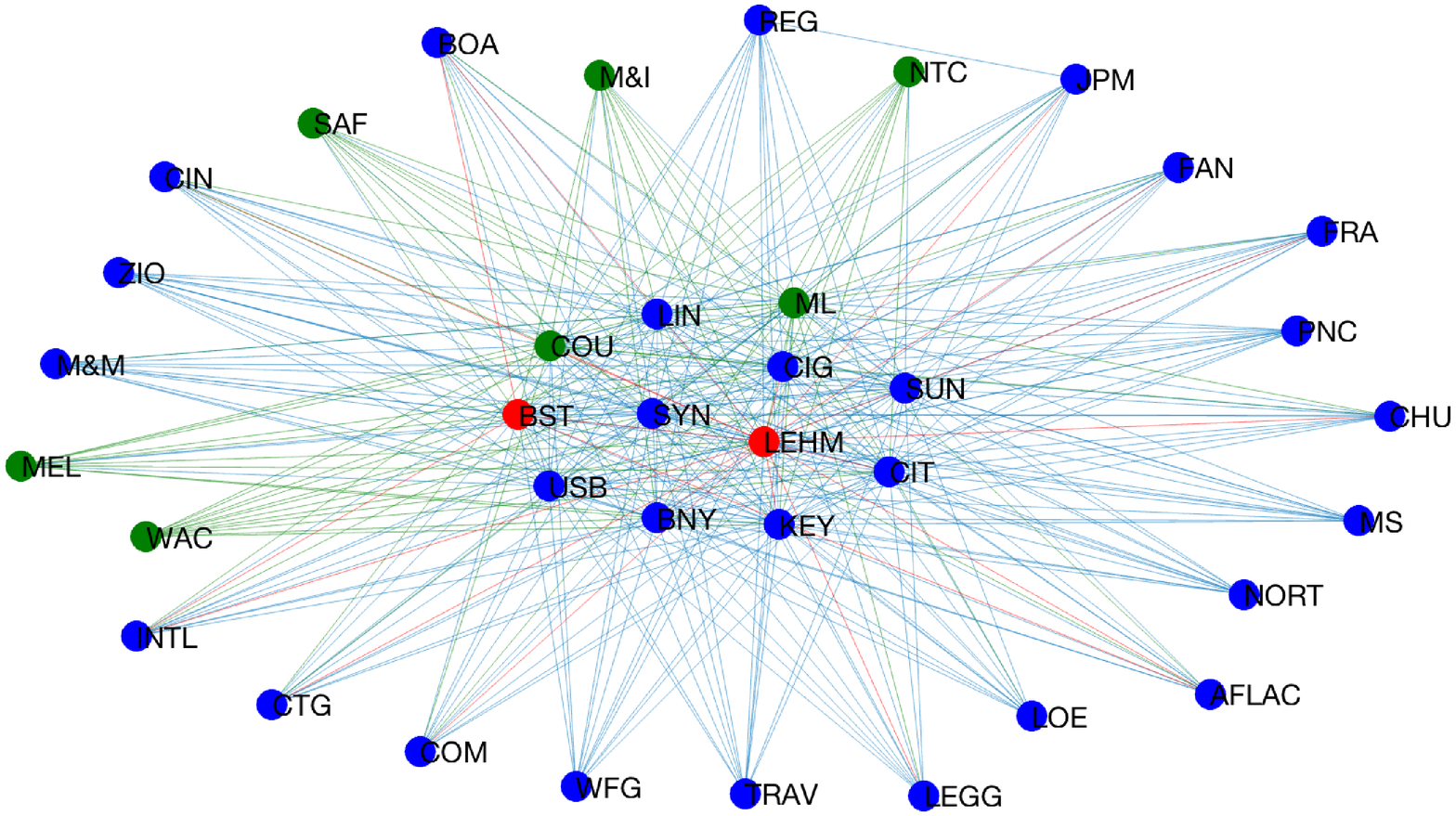}}
\subfloat[GAUSS, prior to GFC]{\label{fig:netcovar_newtork_GAUSS_prior_gfc}\includegraphics[width=0.5\textwidth]{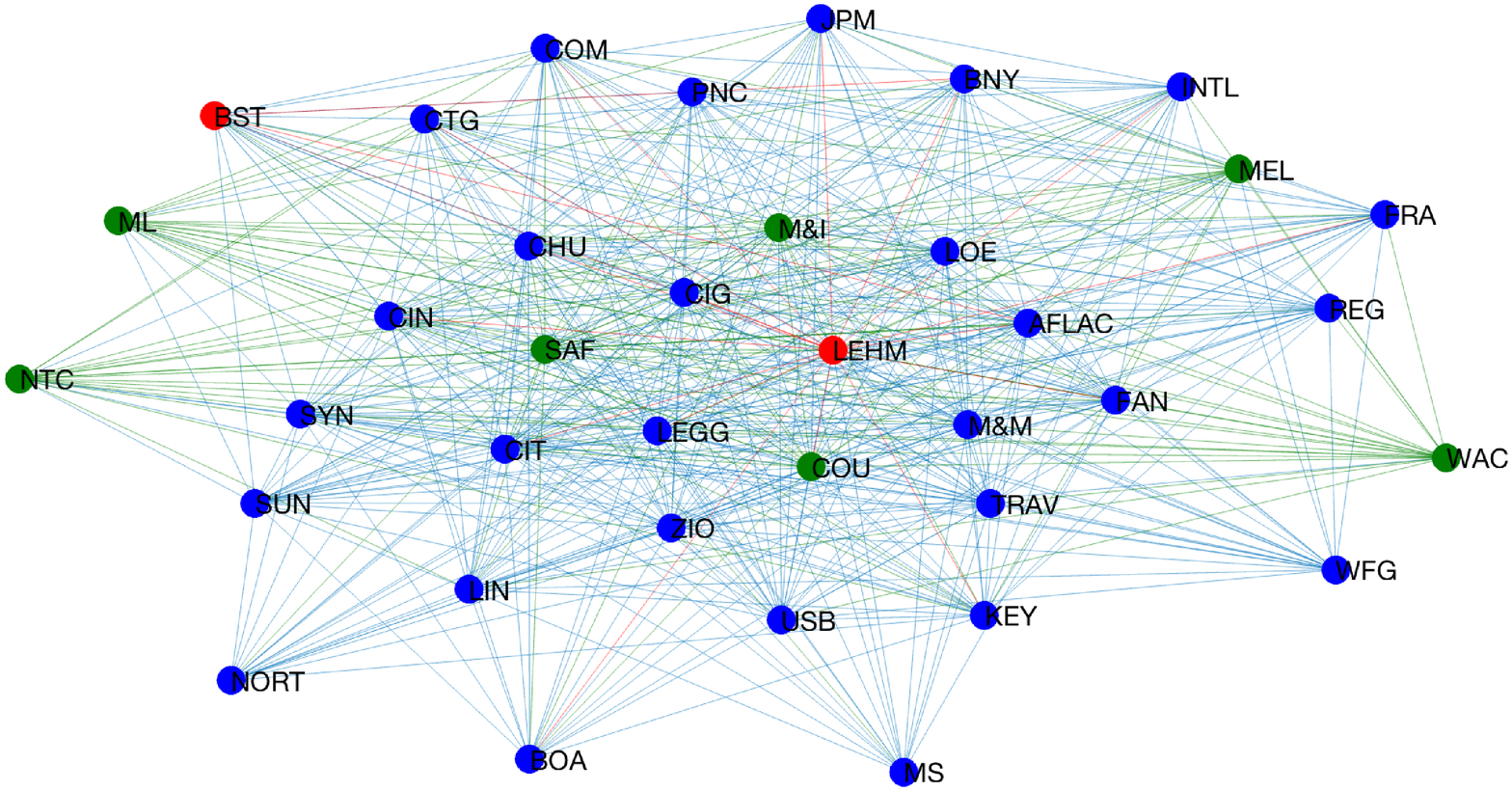}}\\
\subfloat[ESD, after the GFC]{\label{netcovar_newtork_ESD_after_gfc}\includegraphics[width=0.5\textwidth]{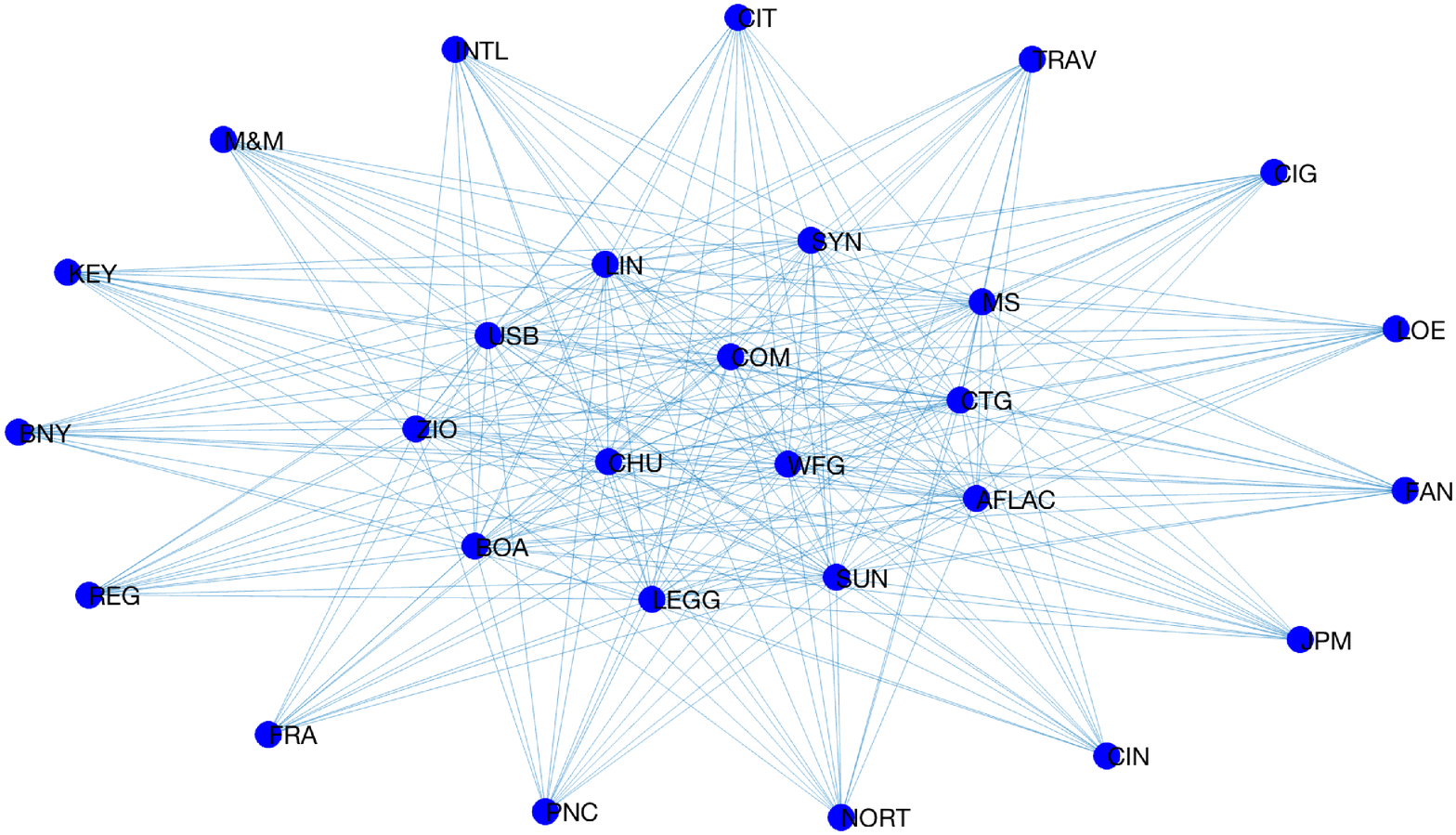}}
\subfloat[GAUSS, after the GFC]{\label{netcovar_newtork_GAUSS_after_gfc}\includegraphics[width=0.5\textwidth]{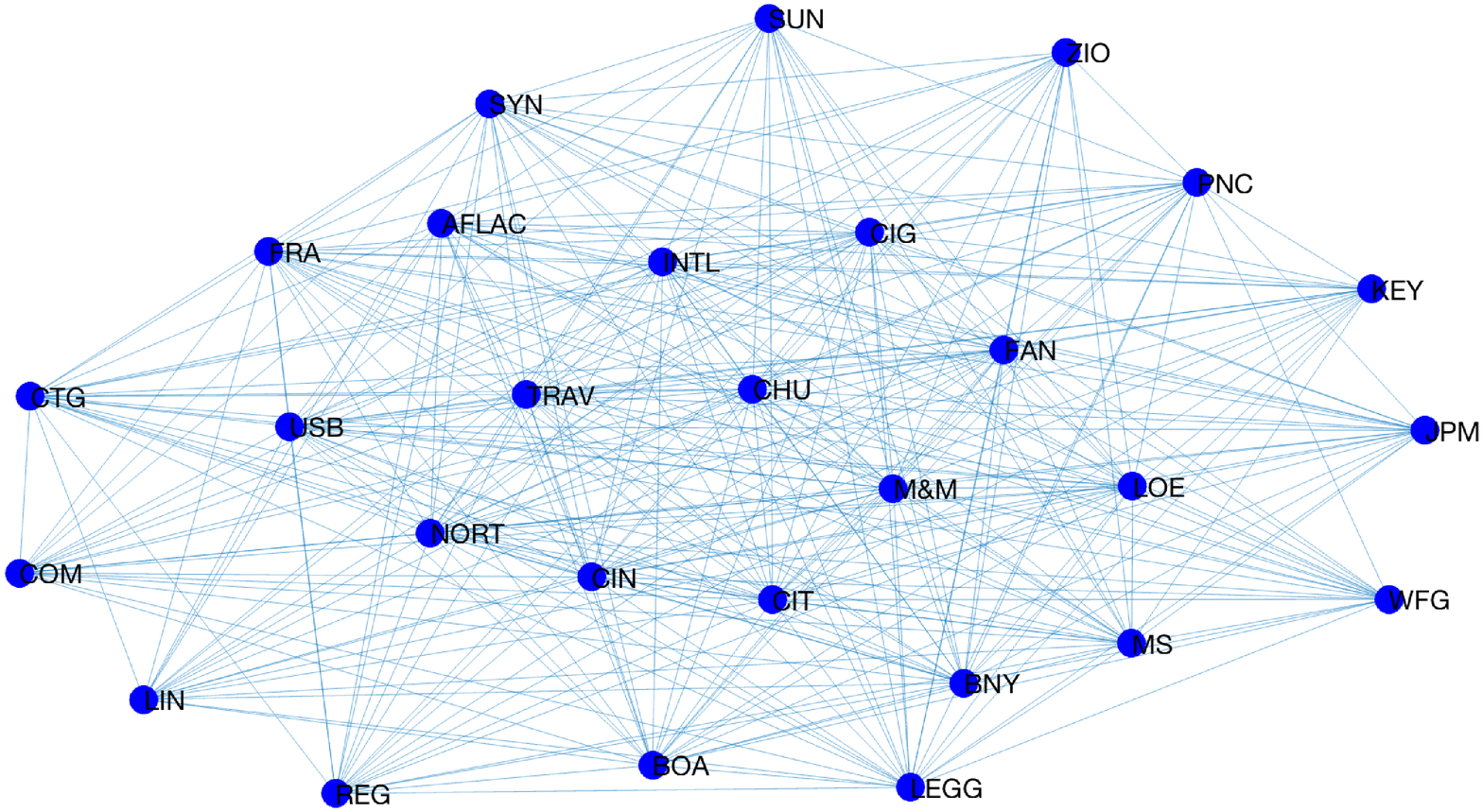}}
\caption{\footnotesize{NetCoVaR systemic network constructed using the systemic dominance test for two different periods before and after the global financial crisis of 2008 and two assumptions for the weekly log--returns (Gaussian and ESD). Red bullets denote the two institutions that defaulted (Lehman Bros. and Bear Stearns), green bullets denote the institutions that have been acquired, while blue bullets denotes institutions that did not experienced distress instances during the considered period.}}
\label{fig:netcovar_newtork}
\end{center}
\end{sidewaysfigure}
%
%
\begin{figure}[p]
\captionsetup{font={footnotesize}}
\begin{center}

\subfloat[ESD, prior to GFC]{\label{fig:netcovar_newtork_ESD_prior_gfc_RED}\includegraphics[width=1.0\textwidth]{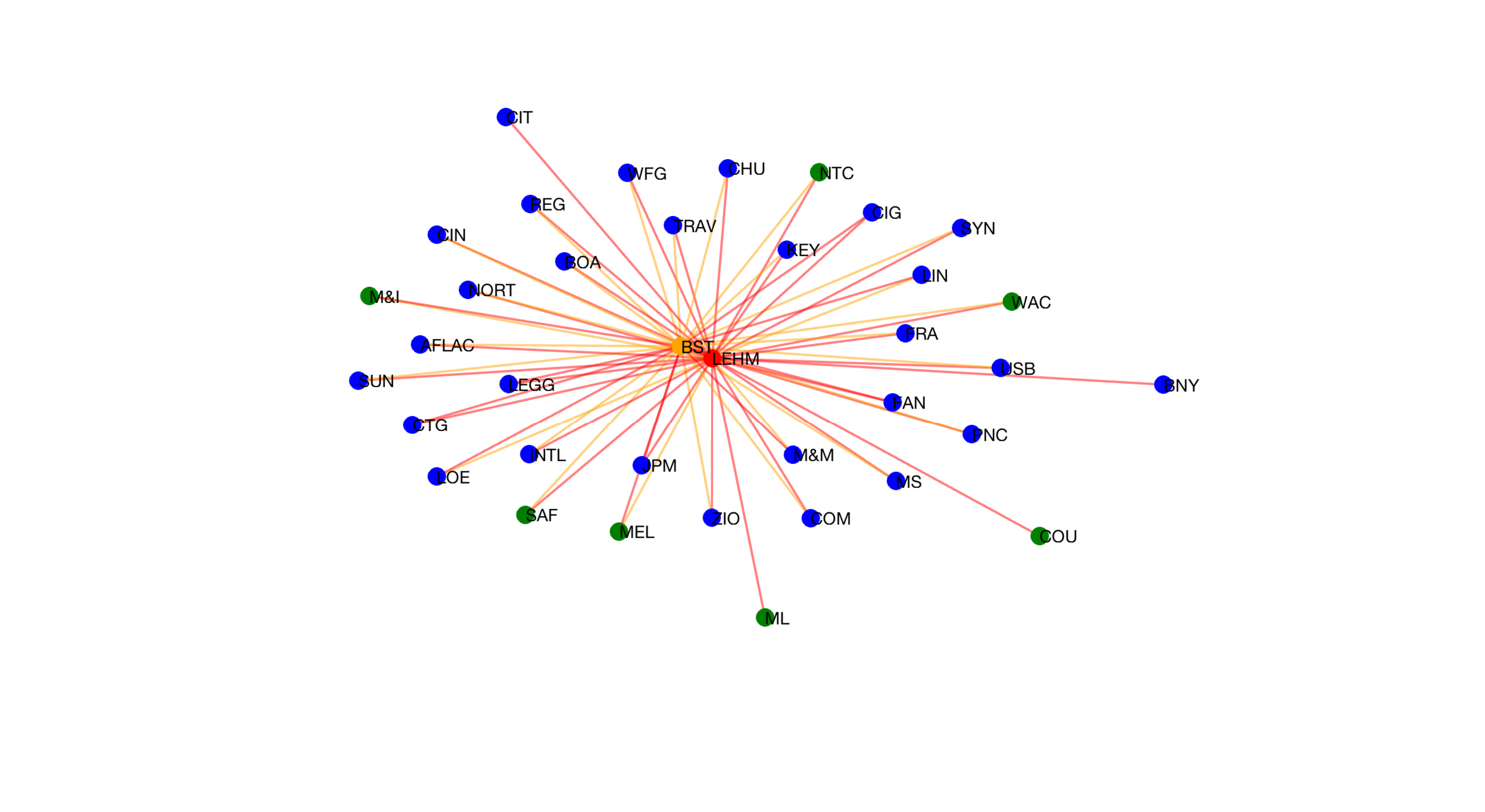}}\\
\subfloat[GAUSS, prior to GFC]{\label{fig:netcovar_newtork_GAUSS_prior_gfc_RED}\includegraphics[width=1.0\textwidth]{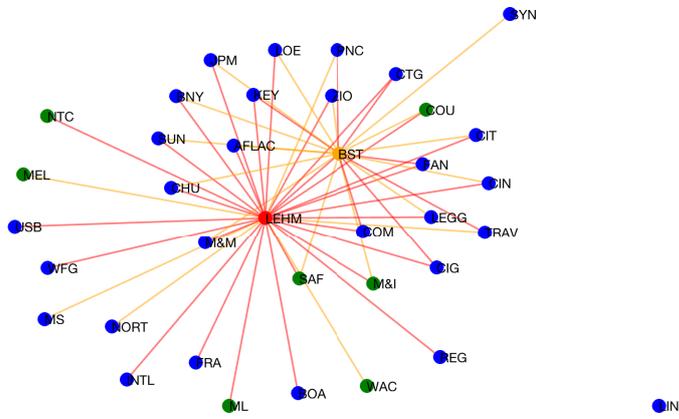}}
\caption{\footnotesize{NetCoVaR subgraphs involving Lehman Bros. Hdg. and Bear Stearns for the period before the global financial crisis of 2008 under the Gaussian and ESD assumptions for the weekly log--returns. Red bullets denote the two institutions that defaulted (Lehman Bros. and Bear Stearns), green bullets denote the institutions that have been acquired, while blue bullets denotes institutions that did not experienced distress instances during the considered period.}}
\label{fig:netcovar_newtork_RED}
\end{center}
\end{figure}
%
%
\begin{table}[!t]
\captionsetup{font={footnotesize}}
\begin{center}
\resizebox{0.6\columnwidth}{!}{%
  \begin{tabular}{lcccc}
\toprule
 & \multicolumn{2}{c}{Before the GFC}&\multicolumn{2}{c}{After the GFC}\\
 \cmidrule(lr){1-1}\cmidrule(lr){2-3}\cmidrule(lr){4-5}
          Institution & ESD & Gauss & ESD & Gauss \\ 
          \cmidrule(lr){1-1}\cmidrule(lr){2-3}\cmidrule(lr){4-5}
ZIONS 						&12  	&	33	&	23	&18\\
WELLS FARGO 				&12  &	22	&	25	&19\\
{\bf WACHOVIA}				&{\bf 12}  &	{\bf 22}	&	--	&--\\
US BANCORP					&31  &	25	&	25	&24\\
TRAVELERS 					&12  &	33	&	13	&26\\
SYNOVUS 					&31  &	31	&	23	&22\\
SUNTRUST 					&31  &	27	&	25	&18\\
{\bf SAFECO }					&{\bf 12}  &	{\bf 33}	&	--	&--\\
REGIONS 					&13  &	25	&	13	&19\\
PNC 						&12  &	25	&	13	&22\\
NORTHERN TRUST				&12  &	19	&	13	&27\\
{\bf NATIONAL CITY }			&{\bf 12}  &	{\bf 22}	&	--	&--\\
MORGAN STANLEY				&12  &	17	&	23	&23\\
{\bf MERRILL LYNCH }			&{\bf 32}  &	{\bf 23}	&	--	&--\\
{\bf MELLON 	}				&{\bf 12}  &	{\bf 24}	&	--	&--\\
{\bf MARSHALL \& ILSLEY}		&{\bf 12}  &	{\bf 35}	&	--	&--\\
MARSH \& MCLENNAN			&12  &	34	&	13	&26\\
LOEWS						&12  &	33	&	13	&25\\
LINCOLN 						&31  &	27	&	25	&19\\
{\bf\color{red}{LEHMAN} 	}		&{\bf\color{red}{36}}  &	{\bf\color{red}{34}}	&	--	&--\\
LEGG MASON					&12  &	33	&	23	&23\\
KEYCORP					&31  &	27	&	13	&19\\
JP MORGAN 					&13  &	24	&	13	&19\\
FRANKLIN 					&12  &	24	&	13	&24\\
FANNIE MAE					&12  &	33	&	13	&27\\
{\bf COUNTRYWIDE }			&{\bf 32}  &	{\bf 33}	&	--	&--\\
COMERICA					&12  &	25	&	27	&19\\
CITIZENS 					&32  &	36	&	13	&27\\
CITIGROUP					&12  &	25	&	25	&21\\
CINCINNATI 					&12  &	33	&	13	&26\\
CIGNA						&31  &	36	&	13	&26\\
{\bf CHUBB }					&{\bf 12}  &	{\bf 34}	&	25	&27\\
{\bf\color{red}{BEAR STEARNS}} 	&{\bf\color{red}{32}}  &	{\bf\color{red}{23}}	&	--	&--\\
BANK OF NEW YORK MELLON	&32  &	24	&	13	&23\\
BANK OF AMERICA				&12  &	19	&	23	&24\\
AMERICAN INTL				&12  &	24	&	13	&27\\
AFLAC						&12  &	33	&	27	&22\\
\cmidrule(lr){1-1}\cmidrule(lr){2-3}\cmidrule(lr){4-5}
{\bf Total number of edges}						&  {\bf 342}	&	{\bf515}	&	{\bf 257}	&{\bf321}	\\
{\bf \% of edges}						&  {\bf 51.14}	&	{\bf 77.33}	&	{\bf 67.99}	&{\bf84.92}	\\
\bottomrule
\end{tabular}}
\end{center}
\caption{\footnotesize{Number of edges of the NetCoVaR graphs.}}
\label{tab:graph_summary_stat}
\end{table}
%
Visual inspection of Figures \ref{fig:netcovar_newtork_ESD_prior_gfc}--\ref{fig:netcovar_newtork_GAUSS_prior_gfc} reveals a great difference between the two graphs. First, the NetCoVaR under the ESD places Lehman Bros. and Bear Stearns at  the centre of the network meaning that their are highly interconnected with all the remaining institutions. This is not the case for the NetCoVaR under the Gaussian assumption that places Bear Stearns at a corner. Furthermore, the number of edges is different on the two graphs indicating that the ESD assumption induces a very high level of sparsity as compared with the Gaussian counterpart. As regards the number of edges Table \ref{tab:graph_summary_stat} confirms previous results and reports and increase of the total number of edges after the GFC took place for both graphs. Indeed, the total number of edges can be interpreted as a proxy for the degree of interconnectedness of the financial market. Therefore, an increase on the number of edges implies that financial markets tends to become more interconnected during periods of financial crisis as a consequence of an increases of the correlation. This phenomenon is well documented in the systemic risk literature, see, e.g., \cite{billio_etal.2012}. More surprisingly, Table \ref{tab:graph_summary_stat} reveals another relevant difference between the two graphs concerning the institutions that have the large number of interconnections. Specifically, for the Gaussian NetCoVaR the institution with the highest number of interconnections are those that experienced distress instances during the GFC, while this is not the case for the ESD NetCoVaR. Among the institutions that have been acquired during the GFC only Merrill Lynch, Lehman Bros. Countrywide and Bear Stearns display a number of edges greater than $32$ under the ESD assumption.\newline
\indent To further illustrate the difference between the NetCoVaRs under the Gaussian and ESD assumptions, Figure \ref{fig:netcovar_newtork_RED} plots the subgraphs involving the two institutions that defaulted during the GFC: Lehman Bros. and Bear Stearns. Again, the difference between the ESD NetCoVaR and the Gaussian NetCoVaR is evident. Indeed, Bearn Stearns and Lehman Bros. are both directly connected with all the remaining institutions for the ESD NetCoVaR, while the Gaussian NetCoVaR display two clusters of institutions connected either with Lehman Bros. or with Bear Stearns.
%
\section{Conclusion}
\label{sec:conclusion}
%
\noindent In this paper we present the multivariate extension of the method of simulated quantiles proposed in \cite{dominicy_veredas.2013}. The method is useful when either the density function does not have an analytical expression or/and moments do not exits, provided that it can be easily simulated. Projectional quantiles along optimal directions are then introduced in order to carry the information over the parameters of interest in an efficient way. We establish the consistency and the asymptotic distribution of the proposed MMSQ estimator. We also introduce a sparse version of the MMSQ using the SCAD $\ell_1$--penalty of \cite{fan_li.2001} into the MMSQ objective function in order to achieve sparse estimation of the scaling matrix. We extend the asymptotic theory and we show that the sparse--MMSQ estimator enjoys the oracle properties under mild regularity conditions. The method is illustrated using several synthetic datasets from the Elliptical Stable distribution previously considered by \cite{lombardi_veredas.2009} for which alternative methods are recognised to perform poorly.  The methodology has been effectively applied in the context of systemic risk measurement. Our results confirm that the assumption of Elliptically Stable distributed returns as well as the introduced systemic risk measurement framework can effectively represent an improvement with respect to existing methodologies. 
%
%
\section*{Acknowledgements}
We would like to express our sincere thanks to all the participants to the International Workshop on Statistical Modelling 2017 for their constructive comments. The authors would like to thank Prof. Francesco Ravazzolo, Free University of Bozen, for his suggestions that greatly enhanced the quality and clarity of the manuscript.
%
\appendix
\section{Proofs of the main results}
\label{app:appendix_proofs}
%
\begin{proof}{Theorem \ref{th:proj_quant_asy}.}
\begin{enumerate}
\item[{\it (i)}] The proof of this result can be found in \cite{cramer.1946}.
\item[{\it (ii)}] Without loss of generality we can consider $\tau_{1}, \tau_{2}$ and $Z_{1}, Z_{2}$. Under the hypothesis of the theorem, the sample quantiles $\hat{q}_{\tau_{1},Z_{1}}$ and $\hat{q}_{\tau_{2},Z_{2}}$ admit the Bahadur representation
\begin{align}
\hat{q}_{\tau_{j},Z_{j}}-q_{\tau_{j},Z_{j}}=\frac{1}{n}\sum_{i=1}^n\frac{\tau_{j}-\bbone_{\left[z_{i,j}\leq q_{\tau_{j}}\right]}}{f_{Z_{j}}\left(q_{\tau_{j}}\right)}+R_{n,j},\nonumber
\end{align}
for $j=1,2$, where $R_{n,j}=o\left(\frac{1}{\sqrt{n}}\right)$. Let us start from the variance of $\hat{q}_{\tau_{1},Z_{1}}-q_{\tau_{1},Z_{1}}$. 
\begin{align}
\Var\left(\hat{q}_{\tau_{1},Z_{1}}-q_{\tau_{1},Z_{1}}\right)&=\Var\left(\frac{1}{n}\sum_{i=1}^n\frac{\tau_{1}-\bbone_{\left[z_{i,1}\leq q_{\tau_{1}}\right]}}{f_{Z_{1}}\left(q_{\tau_{1}}\right)}+R_{n,1}\right)\nonumber\\
& =\mathbb{E}\left[\left(\frac{1}{n}\sum_{i=1}^n\frac{\tau_{1}-\bbone_{\left[z_{i,1}\leq q_{\tau_{1}}\right]}}{f_{Z_{1}}\left(q_{\tau_{1}}\right)}+R_{n,1}\right)^2\right]\nonumber\\
& =\mathbb{E}\left[\left(\frac{1}{n}\sum_{i=1}^n\frac{\tau_{1}-\bbone_{\left[z_{i,1}\leq q_{\tau_{1}}\right]}}{f_{Z_{1}}\left(q_{\tau_{1}}\right)}\right)^2\right.\nonumber\\
&\qquad\left.+2R_{n,1}\frac{1}{n}\sum_{i=1}^n\frac{\tau_{1}-\bbone_{\left[z_{i,1}\leq q_{\tau_{1}}\right]}}{f_{Z_{1}}\left(q_{\tau_{1}}\right)}+R_{n,1}^2\right]\nonumber\\
& =\mathbb{E}\left[\left(\frac{1}{n}\sum_{i=1}^n\frac{\tau_{1}-\bbone_{\left[z_{i,1}\leq q_{\tau_{1}}\right]}}{f_{Z_{1}}\left(q_{\tau_{1}}\right)}\right)^2\right]\nonumber\\
&\qquad+2R_{n,1}\mathbb{E}\left[\frac{1}{n}\sum_{i=1}^n\frac{\tau_{1}-\bbone_{\left[z_{i,1}\leq q_{\tau_{1}}\right]}}{f_{Z_{1}}\left(q_{\tau_{1}}\right)}\right]+R_{n,1}^2\nonumber\\
&=\frac{1}{n^2f_{Z_{1}}\left(q_{\tau_{1}}\right)^2}\mathbb{E}\left[\left(\sum_{i=1}^n\tau_{1}-\bbone_{\left[z_{i,1}\leq q_{\tau_{1}}\right]}\right)^2\right]\nonumber\\
& \qquad+\frac{2R_{n,1}}{nf_{Z_{1}}\left(q_{\tau_{1}}\right)}\mathbb{E}\left[\sum_{i=1}^n\tau_{1}-\bbone_{\left[z_{i,1}\leq q_{\tau_{1}}\right]}\right]+R_{n,1}^2\nonumber\\
&=\frac{1}{n^2f_{Z_{1}}\left(q_{\tau_{1}}\right)^2}\Var\left(\sum_{i=1}^n\tau_{1}-\bbone_{\left[z_{i,1}\leq q_{\tau_{1}}\right]}\right)+R_{n,1}^2\nonumber\\
&=\frac{\tau_{1}\left(1-\tau_{1}\right)}{f_{Z_{1}}\left(q_{\tau_{1}}\right)^2}+R_{n,1}^2,\nonumber
\end{align}
where $R_{n,1}^2=o\left(\frac{1}{n}\right)$. The same holds for the variance of $\hat{q}_{\tau_{2},Z_{2}}-q_{\tau_{2},Z_{2}}$. Let us consider the covariance.
\begin{align}
&\Cov\left(\hat{q}_{\tau_{1},Z_{1}}-q_{\tau_{1},Z_{1}}, \hat{q}_{\tau_{2},Z_{2}}-q_{\tau_{2},Z_{2}}\right)\nonumber\\
&=\Cov\left(\frac{1}{n}\sum_{i=1}^n\frac{\tau_{1}-\bbone_{\left[z_{i,1}\leq q_{\tau_{1}}\right]}}{f\left(q_{\tau_{1}}\right)}+R_{n,1}, \frac{1}{n}\sum_{i=1}^n\frac{\tau_{2}-\bbone_{\left[z_{i,2}\leq q_{\tau_{2}}\right]}}{f\left(q_{\tau_{2}}\right)}R_{n,2}\right)\nonumber\\
& =\mathbb{E}\left[\left(\frac{1}{n}\sum_{i=1}^n\frac{\tau_{1}-\bbone_{\left[z_{i,1}\leq q_{\tau_{1}}\right]}}{f\left(q_{\tau_{1}}\right)}+R_{n,1}\right)\left(\frac{1}{n}\sum_{i=1}^n\frac{\tau_{2}-\bbone_{\left[z_{i,2}\leq q_{\tau_{2}}\right]}}{f\left(q_{\tau_{2}}\right)}+R_{n,2}\right)\right]\nonumber\\
&=\frac{1}{n^2}\mathbb{E}\left[\sum_{i=1}^n\frac{\tau_{1}-\bbone_{\left[z_{i,1}\leq q_{\tau_{1}}\right]}}{f\left(q_{\tau_{1}}\right)}\sum_{i=1}^n\frac{\tau_{2}-\bbone_{\left[z_{i,2}\leq q_{\tau_{2}}\right]}}{f\left(q_{\tau_{2}}\right)}\right]\nonumber\\
&\quad+\mathbb{E}\left[R_{n,1}\frac{1}{n}\sum_{i=1}^n\frac{\tau_{2}-\bbone_{\left[z_{i,2}\leq q_{\tau_{2}}\right]}}{f\left(q_{\tau_{2}}\right)}\right]+\mathbb{E}\left[R_{n,2}\frac{1}{n}\sum_{i=1}^n\frac{\tau_{1}-\bbone_{\left[z_{i,1}\leq q_{\tau_{1}}\right]}}{f\left(q_{\tau_{1}}\right)}\right]+R_{n,1}R_{n,2}\nonumber\\
&=\frac{1}{n^2}\mathbb{E}\left[\left(\frac{n\tau_{1}}{f\left(q_{\tau_{1}}\right)}-\sum_{i=1}^n\frac{\bbone_{\left[z_{i,1}\leq q_{\tau_{1}}\right]}}{f\left(q_{\tau_{1}}\right)}\right)\left(\frac{n\tau_{2}}{f\left(q_{\tau_{2}}\right)}-\sum_{i=1}^n\frac{\bbone_{\left[z_{i,2}\leq q_{\tau_{2}}\right]}}{f\left(q_{\tau_{2}}\right)}\right)\right]\nonumber\\
&\quad+\frac{R_{n, 1}}{nf\left(q_{\tau_{2}}\right)}\mathbb{E}\left[\sum_{i=1}^n\tau_{2}-\bbone_{\left[z_{i,2}\leq q_{\tau_{2}}\right]}\right]+\frac{R_{n,2}}{nf\left(q_{\tau_{1}}\right)}\mathbb{E}\left[\sum_{i=1}^n\tau_{1}-\bbone_{\left[z_{i,1}\leq q_{\tau_{1}}\right]}\right]+R_{n,1}R_{n,2}\nonumber\\
&= \frac{1}{n^2}\mathbb{E}\left[\frac{n\tau_{1}}{f\left(q_{\tau_{1}}\right)}\frac{n\tau_{2}}{f\left(q_{\tau_{2}}\right)}\right]-\frac{1}{n^2}\mathbb{E}\left[\frac{n\tau_{1}}{f\left(q_{\tau_{1}}\right)}\sum_{i=1}^n\frac{\bbone_{\left[z_{i,2}\leq q_{\tau_{2}}\right]}}{f\left(q_{\tau_{2}}\right)}\right]\nonumber\\
&\quad-\frac{1}{n^2}\mathbb{E}\left[\frac{n\tau_{2}}{f\left(q_{\tau_{2}}\right)}\sum_{i=1}^n\frac{\bbone_{\left[z_{i,1}\leq q_{\tau_{1}}\right]}}{f\left(q_{\tau_{1}}\right)}\right]+\frac{1}{n^2}\mathbb{E}\left[\sum_{i=1}^n\frac{\bbone_{\left[z_{i,1}\leq q_{\tau_{1}}\right]}}{f\left(q_{\tau_{1}}\right)}\sum_{i=1}^n\frac{\bbone_{\left[z_{i,2}\leq q_{\tau_{2}}\right]}}{f\left(q_{\tau_{2}}\right)}\right]+R_{n,1}R_{n,2}\nonumber\\
&=\frac{\tau_{1}\tau_{2}}{f\left(q_{\tau_{1}}\right)f\left(q_{\tau_{2}}\right)}- \frac{\tau_{1}}{nf\left(q_{\tau_{1}}\right)f\left(q_{\tau_{2}}\right)}\mathbb{E}\left[\sum_{i=1}^n\bbone_{\left[z_{i,2}\leq q_{\tau_{2}}\right]}\right]\nonumber\\
&\quad-\frac{\tau_{2}}{nf\left(q_{\tau_{1}}\right)f\left(q_{\tau_{2}}\right)}\mathbb{E}\left[\sum_{i=1}^n\bbone_{\left[z_{i,1}\leq q_{\tau_{1}}\right]}\right]\nonumber\\
&\quad+\frac{1}{n^2f\left(q_{\tau_{1}}\right)f\left(q_{\tau_{2}}\right)}\mathbb{E}\left[\sum_{i=1}^n\bbone_{\left[z_{i,1}\leq q_{\tau_{1}}\right]}\sum_{i=1}^n\bbone_{\left[z_{i,2}\leq q_{\tau_{2}}\right]}\right]+R_{n,1}R_{n,2}\nonumber\\
& =\frac{\tau_{1}\tau_{2}}{f\left(q_{\tau_{1}}\right)f\left(q_{\tau_{2}}\right)}- \frac{\tau_{1}}{f\left(q_{\tau_{1}}\right)f\left(q_{\tau_{2}}\right)}\mathbb{E}\left[\bbone_{\left[z_{2}\leq q_{\tau_{2}}\right]}\right]\nonumber\\
&\quad-\frac{\tau_{2}}{f\left(q_{\tau_{1}}\right)f\left(q_{\tau_{2}}\right)}\mathbb{E}\left[\bbone_{\left[z_{1}\leq q_{\tau_{1}}\right]}\right]\nonumber\\
&\quad+\frac{1}{f\left(q_{\tau_{1}}\right)f\left(q_{\tau_{2}}\right)}\mathbb{E}\left[\bbone_{\left[z_{1}\leq q_{\tau_{1}}\right]}\bbone_{\left[z_{2}\leq q_{\tau_{2}}\right]}\right]+R_{n,1}R_{n,2}\nonumber\\
&=\frac{\tau_{1}\tau_{2}}{f\left(q_{\tau_{1}}\right)f\left(q_{\tau_{2}}\right)}-2\frac{\tau_{1}\tau_{2}}{f\left(q_{\tau_{1}}\right)f\left(q_{\tau_{2}}\right)}+\frac{F_{Z_{1},Z_{2}}\left(\bq_{\tau},\bSigma_{Z_{1},Z_{2}}\right)}{f\left(q_{\tau_{1}}\right)f\left(q_{\tau_{2}}\right)}+R_{n,1}R_{n,2}\nonumber\\
& =-\frac{\tau_{1}\tau_{2}}{f\left(q_{\tau_{1}}\right)f\left(q_{\tau_{2}}\right)}+\frac{F_{Z_{1},Z_{2}}\left(\bq_{\tau},\bSigma_{Z_{1},Z_{2}}\right)}{f\left(q_{\tau_{1}}\right)f\left(q_{\tau_{2}}\right)}+R_{n,1}R_{n,2},\nonumber
\end{align}
where $\bq_{\tau}=\left(q_{\tau_{1}}, q_{\tau_{2}}\right)^\prime$ and $R_{n,1}R_{n,2}=o\left(\frac{1}{n}\right)$.
\item[{\it (iii)}] Using the Bahadur representation
\begin{align}
&\Cov\left(\hat{q}_{\tau_{k},\bu}-q_{\tau_{k},\bu}, \hat{q}_{\tau_{j},\bu}-q_{\tau_{j},\bu}\right)\nonumber\\
&=\Cov\left(\frac{1}{n}\sum_{i=1}^n\frac{\tau_{i}-\bbone_{\left[z_{i}\leq q_{\tau_{k},\bu}\right]}}{f\left(q_{\tau_{k},\bu}\right)}+R_{n,1}, \frac{1}{n}\sum_{i=1}^n\frac{\tau_{j}-\bbone_{\left[z_{i}\leq q_{\tau_{j},\bu}\right]}}{f\left(q_{\tau_{j},\bu}\right)}+R_{n,2}\right)\nonumber\\
&=\mathbb{E}\left[\left(\frac{1}{n}\sum_{i=1}^n\frac{\tau_{k}-\bbone_{\left[z_{i}\leq q_{\tau_{k},\bu}\right]}}{f\left(q_{\tau_{i},\bu}\right)}+R_{n,1}\right)\left(\frac{1}{n}\sum_{i=1}^n\frac{\tau_{j}-\bbone_{\left[z_{i}\leq q_{\tau_{j},\bu}\right]}}{f\left(q_{\tau_{j},\bu}\right)}+R_{n,2}\right)\right]\nonumber\\
&=\frac{1}{n^2}\mathbb{E}\left[\sum_{i=1}^n\frac{\tau_{k}-\bbone_{\left[z_{i}\leq q_{\tau_{k},\bu}\right]}}{f\left(q_{\tau_{k},\bu}\right)}\sum_{i=1}^n\frac{\tau_{j}-\bbone_{\left[z_{i}\leq q_{\tau_{j},\bu}\right]}}{f\left(q_{\tau_{j},\bu}\right)}\right]\nonumber\\
&\quad+R_{n,1}\mathbb{E}\left[\frac{1}{n}\sum_{i=1}^n\frac{\tau_{j}-\bbone_{\left[z_{i}\leq q_{\tau_{j},\bu}\right]}}{f\left(q_{\tau_{j},\bu}\right)}\right]+R_{n,2}\mathbb{E}\left[\sum_{i=1}^n\frac{\tau_{k}-\bbone_{\left[z_{i}\leq q_{\tau_{k},\bu}\right]}}{f\left(q_{\tau_{k},\bu}\right)}\right]+R_{n,1}R_{n,2}\nonumber\\
&=\frac{1}{f\left(q_{\tau_{k},\bu}\right)f\left(q_{\tau_{j},\bu}\right)}\mathbb{E}\left[\left(\tau_{k}-\bbone_{\left[z_{i}\leq q_{\tau_{k},\bu}\right]}\right)\left(\tau_{j}-\bbone_{\left[z_{i}\leq q_{\tau_{j},\bu}\right]}\right)\right]+R_{n,1}R_{n,2}\nonumber\\
&=\frac{1}{f\left(q_{\tau_{k},\bu}\right)f\left(q_{\tau_{j},\bu}\right)}\left(\tau_{k}\tau_{j}- \tau_{k}\mathbb{E}\left[\bbone_{\left[z_{i}\leq q_{\tau_{j},\bu}\right]}\right]-\tau_{j}\mathbb{E}\left[\bbone_{\left[z_{i}\leq q_{\tau_{k},\bu}\right]}\right]\right)\nonumber\\
&\quad+\frac{1}{f\left(q_{\tau_{k},\bu}\right)f\left(q_{\tau_{j},\bu}\right)}\left(\mathbb{E}\left[\bbone_{\left[z_{i}\leq q_{\tau_{k},\bu}\right]}\bbone_{\left[z_{i}\leq q_{\tau_{j},\bu}\right]}\right]\right)+R_{n,1}R_{n,2}\nonumber\\
&=\frac{\tau_{k}\wedge\tau{j}-\tau_{k}\tau_{j}}{f\left(q_{\tau_{k},\bu}\right)f\left(q_{\tau_{j},\bu}\right)}+R_{n,1}R_{n,2}.\nonumber
\end{align}
\end{enumerate}
\end{proof}

\begin{proof}{Theorem \ref{th:functions_quant_asy}.}
The function $\hat{\boldsymbol{\Phi}}$ is assumed to be continuously differentiable, so Delta method applies
\begin{align}
\hat{\bPhi}\approx \bPhi_{\vartheta}+\frac{\partial\bPhi_{\vartheta}}{\partial\bq}\left(\hat{\bq}-\bq\right),
\end{align} 
then 
\begin{align}
\Var\left(\hat{\bPhi}\right)&\approx \Var\left(\frac{\partial\bPhi_{\vartheta}}{\partial\bq}\hat{\bq}\right)\nonumber\\
&=\frac{\partial\bPhi_{\vartheta}}{\partial\bq}^{\prime}\Cov\left(\hat{\bq}\right)\frac{\partial\bPhi_{\vartheta}}{\partial\bq},
\end{align}
where $\hat{\bq}=\left(\hat{\bq}_{\boldsymbol{\tau}_{1}\bu_{1}}, \dots, \hat{\bq}_{\boldsymbol{\tau}_{K}\bu_{K}}\right)$.
\end{proof}

\begin{proof}{Theorem \ref{th:mmsq_asy}.}
The first order condition of \eqref{eq:mmsq_min_problem} is 
\begin{align}
\frac{1}{R}\sum_{r=1}^{R}\frac{\partial \tilde{\bPhi}_{\vartheta}^{r}}{\partial \vartheta}\mathbf{W}_{\bar{\vartheta}}\left(\hat{\bPhi}_{\vartheta}-\frac{1}{R}\sum_{r=1}^{R}\tilde{\bPhi}_{\vartheta}^{r}\right)=0,
\end{align}
where $\bar{\vartheta}$ is a consistent estimate of $\vartheta$. Let us consider the first order Taylor expansion around the true parameter $\vartheta_{0}$
\begin{align}
& \frac{1}{R}\sum_{r=1}^{R}\frac{\partial \tilde{\bPhi}_{\vartheta_{0}}^{r}}{\partial \vartheta}\mathbf{W}_{\bar{\vartheta}}\left(\hat{\bPhi}_{\vartheta}-\frac{1}{R}\sum_{r=1}^{R}\tilde{\bPhi}_{\vartheta_{0}}^{r}\right)\nonumber\\
& \qquad -\frac{1}{R}\sum_{r=1}^{R}\frac{\partial \tilde{\bPhi}_{\vartheta_{0}}^{r}}{\partial \vartheta}\mathbf{W}_{\bar{\vartheta}}\frac{1}{R}\sum_{r=1}^{R}\frac{\partial \tilde{\bPhi}_{\vartheta_{0}}^{r}}{\partial \vartheta}\left(\hat{\vartheta}-\vartheta_{0}\right)=o_{p}\left(1\right).
\end{align}
From this equation we get 
\begin{align}
\sqrt{n}\left(\hat{\vartheta}-\vartheta_{0}\right)\approx \left(\frac{\partial\tilde{\bPhi}_{\vartheta}^{\prime}}{\partial\vartheta}\mathbf{W}_{\bar{\vartheta}}\frac{\partial\tilde{\bPhi}_{\vartheta}}{\partial\vartheta}\right)^{-1}\frac{\partial\tilde{\bPhi}_{\vartheta}}{\partial\vartheta}\mathbf{W}_{\bar{\vartheta}}\sqrt{n}\left(\hat{\bPhi}-\frac{1}{R}\sum_{r=1}^{R}\tilde{\bPhi}_{\vartheta_{0}}^{r}\right),
\end{align}
as $n\rightarrow\infty$. From Theorem \ref{th:functions_quant_asy} 
\begin{align}
\sqrt{n}\left(\hat{\bPhi}-\frac{1}{R}\sum_{r=1}^{R}\tilde{\bPhi}_{\vartheta_{0}}^{r}\right)\rightarrow^{d}\mathcal{N}\left(\bO, \left(1+\frac{1}{R}\right)\bOmega_{\vartheta}\right),
\end{align}
as $n\rightarrow\infty$, and $\tilde{\boldsymbol{\Phi}}_{\vartheta_{0}}^{r}$ converges to $\boldsymbol{\Phi}_{\vartheta}$. Moreover since $\bar{\vartheta}$ is consistent the matrix $\mathbf{W}_{\bar{\vartheta}}$ converges to $\mathbf{W}_{\vartheta}$. 
From these results we get  
\begin{align}
&\Var\left(\sqrt{n}\left(\hat{\vartheta}-\vartheta\right)\right)\rightarrow
\left(1+\frac{1}{R}\right) \left[\mathbf{H}_\vartheta^{-1}\frac{\partial\boldsymbol{\Phi}_{\vartheta}}{\partial\vartheta}\right]\mathbf{W}_{\vartheta}\bOmega_{\vartheta}\mathbf{W}_\vartheta^{\prime}\left[\mathbf{H}_\vartheta^{-1}\frac{\partial\boldsymbol{\Phi}_{\vartheta}}{\partial\vartheta}\right]^{\prime},
\end{align}
as $n\rightarrow\infty$, where $\mathbf{H}_\vartheta=\frac{\partial\boldsymbol{\Phi}_{\vartheta}}{\partial\vartheta^\prime}\mathbf{W}_{\vartheta}\frac{\partial\boldsymbol{\Phi}_{\vartheta}}{\partial\vartheta}$.
\end{proof}
\begin{proof}{Theorem \ref{th:4}.}
We prove this theorem following \cite{fan_li.2001} and \cite{gao_massam.2015}. 
In the following we denote by $\sigma_{ij}^{0}$ and $\sigma_{ij}$ respectively the zero and non zero off--diagonal elements of the variance covariance matrix.\\
Let us consider a ball $\|\vartheta-\vartheta_{0}\|\leq Mn^{-\frac{1}{2}}$ for some finite constant $M$. In order to prove the result in equation \eqref{eq:s3}, let us consider the first order condition of equation \eqref{eq:s1} and its first order taylor expansion
\begin{align}\label{eq:s2}
 \frac{\partial\mathcal{Q}\left(\vartheta\right)}{\partial \vartheta}&=-2\frac{\partial \tilde{\boldsymbol{\Phi}}_{\vartheta}^{R}}{\partial \vartheta}\mathbf{W}_{\vartheta}\left(\hat{\boldsymbol{\Phi}}-\tilde{\boldsymbol{\Phi}}_{\vartheta}^{R}\right)+n\bv\nonumber\\
&\approx-2\frac{\partial \tilde{\boldsymbol{\Phi}}_{\vartheta_{0}}^{R}}{\partial \vartheta}\mathbf{W}_{\vartheta}\left(\hat{\bPhi}-\tilde{\bPhi}_{\vartheta_{0}}^{R}\right)+2\frac{\partial \tilde{\bPhi}_{\vartheta_{0}}^{R}}{\partial {\vartheta}^{\prime}}\mathbf{W}_{\vartheta}\frac{\partial \tilde{\bPhi}_{\vartheta_{0}}^{R}}{\partial \vartheta}\left(\vartheta-\vartheta_{0}\right)+n\bv,
\end{align}
where $\bv=\left(\bO;p_{\lambda_{n}}'\left(\vert\sigma_{ij}\vert\right)\sgn\left(\sigma_{ij}\right),  i<j\right)$. The first two terms are $\mathcal{O}_{p}\left(n^{-\frac{1}{2}}\right)$. Regarding the penalisation term, let us first consider the zero off--diagonal element $\sigma_{ij}^{0}$. For a given $\lambda_{n}$, the first derivative $p_{\lambda_{n}}^{\prime}\left(\vert\sigma_{ij}\vert\right)$ with respect to $\vert\sigma_{ij}\vert$ is given by
\begin{align}
p_{\lambda_{n}}^{\prime}\left(\vert\sigma_{ij}\vert\right)=\begin{cases}\lambda_{n} & \mbox{if } \vert\sigma_{ij}\vert\leq\lambda_{n} \\
\frac{\left(a\lambda_{n}-\vert\sigma_{ij}\vert\right)}{a-1} & \mbox{if }\lambda_{n}<\vert\sigma_{ij}\vert\leq a\lambda_{n}\\
0 & \mbox{if } a\lambda_{n}<\vert\sigma_{ij}\vert,
\end{cases}
\end{align}
and it holds
\begin{align}
\lim_{\vert\sigma_{ij}\vert\rightarrow0}\frac{p_{\lambda_{n}}^{\prime}\left(\vert\sigma_{ij}\vert\right)}{\lambda_{n}}=1.
\end{align} 
Then, for a generic $\sigma_{ij}^{0}$, the corresponding element in $n\bv$ can be written as
\begin{align}
n\lambda_{n}\sgn\left(\sigma_{ij}\right)\frac{p_{\lambda_{n}}^{\prime}\left(\vert\sigma_{ij}\vert\right)}{\lambda_{n}}=n\lambda_{n}\sgn\left(\sigma_{ij}\right).
\end{align}
We rewrite \eqref{eq:s2} as follows
\begin{align}\label{eq:s3}
 \frac{\partial\mathcal{Q}\left(\vartheta\right)}{\partial \vartheta}=n\lambda_{n}\left\{\lambda_{n}^{-1}\bv-\mathcal{O}_{p}\left(n^{-\frac{n}{2}}\lambda_{n}^{-1}\right)\right\},
\end{align}
Since $\lim\inf_{n\rightarrow\infty}\lim\inf_{\vert\sigma_{ij}\vert\rightarrow0}\frac{p_{\lambda_{n}}^{\prime}\left(\vert\sigma_{ij}\vert\right)}{\lambda_{n}}>0$ and $\sqrt{n}\lambda_{n}\rightarrow\infty$, the term $n\bv$ has asymptotic order higher that $\mathcal{O}_{p}\left(n^{-\frac{1}{2}}\right)$ and dominates the equation \eqref{eq:s3}. This means that the sign of $\frac{\partial\mathcal{Q}\left(\vartheta\right)}{\partial \sigma_{ij}}$  is determined by the sign of $\sigma_{ij}$, i.e. for any local minimiser it holds $\hat{\sigma}_{i,j}=0$ with probability $1$. Now consider the case in which $\sigma_{ij}$ is not a zero element, then using the Taylor approximation we can calculate the following 
\begin{align}\label{eq:s4}
\mathcal{Q}\left(\vartheta_{0}\right)-\mathcal{Q}\left(\vartheta\right)&=\left(\hat{\bPhi}-\tilde{\bPhi}_{\vartheta_{0}}^{R}\right)^{\prime}\mathbf{W}_{\vartheta_{0}}\left(\hat{\bPhi}-\tilde{\bPhi}_{\vartheta_{0}}^{R}\right)-\left(\hat{\bPhi}-\tilde{\bPhi}_{\vartheta}^{R}\right)^{\prime}\mathbf{W}_{\vartheta}\left(\hat{\bPhi}-\tilde{\bPhi}_{\vartheta}^{R}\right)\nonumber\\
&\quad +n\sum_{i<j}\left[p_{\lambda}\left(\vert\sigma_{ij}^{0}\vert\right)-p_{\lambda}\left(\vert\sigma_{ij}\vert\right)\right]\nonumber\\
&\approx2\frac{\partial \tilde{\bPhi}_{\vartheta_{0}}^{R}}{\partial \vartheta}\mathbf{W}_{\vartheta_{0}}\left(\hat{\bPhi}-\tilde{\bPhi}_{\vartheta_{0}}^{R}\right)\left(\vartheta-\vartheta_{0}\right)\nonumber\\
&\quad+\left(\vartheta-\vartheta_{0}\right)^{\prime}\left[-2\frac{\partial \tilde{\bPhi}_{\vartheta_{0}}^{R}}{\partial {\vartheta}^{\prime}}\mathbf{W}_{\vartheta_{0}}\frac{\partial \tilde{\bPhi}_{\vartheta_{0}}^{R}}{\partial \vartheta}\right]\left(\vartheta-\vartheta_{0}\right)\nonumber\\
&\quad -n\sum_{i<j}\left(p_{\lambda_{n}}^{\prime}\left(\vert\sigma_{ij}\vert\right)\sgn\left(\sigma_{ij}\right)\left(\sigma_{ij}-\sigma_{ij}^{0}\right)+p_{\lambda_{n}}^{\prime\prime}\left(\vert\sigma_{ij}\vert\right)\left(\sigma_{ij}-\sigma_{ij}^{0}\right)^2\right),\nonumber
\end{align}
where $p_{\lambda_{n}}^{\prime\prime}\left(|\sigma_{ij}|\right)$ stands for the second derivative. For $n$ large enough the summation term in  equation \eqref{eq:s4} is negligible since $\sigma_{ij}\neq 0$ and
\begin{align}
\lim_{n\rightarrow\infty}p_{\lambda_{n}}^{\prime}\left(|\sigma_{ij}|\right)&=0\nonumber\\
\lim_{n\rightarrow\infty}p_{\lambda_{n}}^{\prime\prime}\left(|\sigma_{ij}|\right)&=0.
\end{align}
The same holds for the fist term. The matrix 
\begin{align}
-2\frac{\partial \tilde{\bPhi}_{\vartheta_{0}}^{R}}{\partial \vartheta^\prime}\mathbf{W}_{\vartheta_{0}}\frac{\partial \tilde{\bPhi}_{\vartheta_{0}}^{R}}{\partial \vartheta},
\end{align} 
is negative definite and for $n$ large it dominates the other terms, therefore $\mathcal{Q}\left(\vartheta_{0}\right)-\mathcal{Q}\left(\vartheta\right)\leq 0$. This implies that there exist a local minimizer $\hat{\vartheta}$ such that $\|\hat{\vartheta}-\vartheta_{0}\|=\mathcal{O}_{p}\left(n^{-\frac{1}{2}}\right)$.
\end{proof}

\begin{proof}{Theorem \ref{th:5}.}
Let us consider the first order Taylor expansion with respect to $\vartheta_{0}^{1}$ of the first order condition computed in equation 
\eqref{eq:s4} 
\begin{align}
\frac{\partial Q\left(\vartheta\right)}{\partial \vartheta^1}&=-2\frac{\partial \tilde{\bPhi}_{\vartheta}^{R}}{\partial \vartheta^1}\mathbf{W}_{\vartheta^1}\left(\hat{\bPhi}-\tilde{\bPhi}_{\vartheta}^{R}\right)+n\bv\nonumber\\
&=-2\frac{\partial \tilde{\bPhi}_{\vartheta_{0}}^{R}}{\partial \vartheta^1}\mathbf{W}_{\vartheta_{0}^1}\left(\hat{\bPhi}-\tilde{\bPhi}_{\vartheta_{0}}^{R}\right)+2\left(\frac{\partial \tilde{\bPhi}_{\vartheta_{0}}^{R}}{\partial {\vartheta^1}^\prime}\mathbf{W}_{\vartheta_{0}^1}\frac{\partial \tilde{\bPhi}_{\vartheta_{0}}^{R}}{\partial \vartheta^1}\right)\left(\vartheta^1-\vartheta_{0}^1\right)\nonumber\\
&\qquad+n\bv_{0}+n\bP_{0}\left(\vartheta^1-\vartheta_{0}^1\right)=0,
\end{align}
where $\bv=\left(\bO;p_{\lambda_{n}}^{\prime}\left(|\sigma_{ij}|\right)\sgn\left(\sigma_{ij}\right), i<j\right)$ and $\bv_{0}$ is $\bv$ computed at the true value of the variance covariance matrix; $\bP=\diag\left\{\bO,p_{\lambda_{n}}^{\prime\prime}\left(|\sigma_{ij}|\right), i<j\right\}$ and $\bP_{0}$ is $\bP$ computed at the true parameter of the variance covariance matrix.
\begin{align}
&2\left(\frac{\partial \tilde{\bPhi}_{\vartheta_{0}}^{R}}{\partial {\vartheta^1}^\prime}\mathbf{W}_{\vartheta_{0}^1}\frac{\partial \tilde{\bPhi}_{\vartheta_{0}}^{R}}{\partial \vartheta^1}\right)\left(\vartheta^1-\vartheta_{0}^1\right)+n\bv+n\bP\left(\vartheta^1-\vartheta_{0}^1\right)\nonumber\\
&\quad=2\frac{\partial \tilde{\bPhi}_{\vartheta_{0}}^{R}}{\partial \vartheta^1}\mathbf{W}_{\vartheta_{0}^1}\left(\hat{\bPhi}-\tilde{\bPhi}_{\vartheta_{0}}^{R}\right)\sqrt{n}\left[2\left(\frac{\partial \tilde{\bPhi}_{\vartheta_{0}}^{R}}{\partial {\vartheta^1}^\prime}\mathbf{W}_{\vartheta_{0}^1}\frac{\partial \tilde{\bPhi}_{\vartheta_{0}}^{R}}{\partial \vartheta^1}\right)+n\bP_{0}\right]\nonumber\\
&\qquad\times\left\{\vartheta^1-\vartheta_{0}^1+\left[2\left(\frac{\partial \tilde{\bPhi}_{\vartheta_{0}}^{R}}{\partial {\vartheta^1}^\prime}\mathbf{W}_{\vartheta_{0}^1}\frac{\partial \tilde{\bPhi}_{\vartheta_{0}}^{R}}{\partial \vartheta^1}\right)+n\bP_{0}\right]^{-1}n\bv_{0}\right\}\nonumber\\
&\quad=2\frac{\partial \tilde{\bPhi}_{\vartheta_{0}}^{R}}{\partial \vartheta^1}\mathbf{W}_{\vartheta_{0}^1}\sqrt{n}\left(\hat{\bPhi}-\tilde{\bPhi}_{\vartheta_{0}}^{R}\right)\xrightarrow[]{d}\mathcal{N}\left(\bO, \frac{\partial \bPhi_{\vartheta_{0}}}{\partial \vartheta^1}\mathbf{W}_{\vartheta_{0}^1}\bOmega_{\vartheta_{0}}\mathbf{W}_{\vartheta^1}^\prime\frac{\partial \bPhi_{\vartheta_{0}}}{\partial {\vartheta^1}^\prime}\right).
\end{align}
Since $\bv_{0}$ and $\bP_{0}$ vanish asymptotically, we apply the same argument of Theorem \ref{th:mmsq_asy} to complete the proof. 
\end{proof}
%
\section*{Appendix B: MMSQ initialisation}
\label{app:mmsq_init}
%
\noindent Without loss of generality we consider the case $m\geq 2$. Since each variables $Y_i$ have univariate Elliptical Stable distribution, then marginals' parameters can be estimated using the approach of \cite{mcculloch.1986}. The off--diagonal parameter of the scale matrix is estimated using the following procedure. For each couple of variables $\bY_{ij}=\left(Y_{i}, Y_{j}\right)^\prime$ it holds 
\begin{align}
\bY_{ij}\sim\mathcal{ESD}_2\left(\alpha, \bxi_{ij}, \bOmega_{ij}\right),
\end{align}
where $\bxi_{ij} =\left(\xi_{i},\xi_{j}\right)^\prime$ and $\bOmega_{ij} = \left[\begin{matrix}\omega^{2}_{i}&\omega_{ij}\\ \omega_{ij}&\omega^{2}_{j}\end{matrix}\right]$. Let us consider the standardised variables $\bX_{ij}=\left(X_{i}, X_{j}\right)^\prime$ where
\begin{align}
\left(X_{i}, X_{j}\right)=\left(\frac{Y_{i}-\xi_{i}}{\omega_{ii}}, \frac{Y_{j}-\xi_{j}}{\omega_{jj}}\right),
\end{align} 
then  $\bX_{ij}\thicksim\mathcal{ESD}_2\left(\alpha, \bO, \bar{\bOmega}_{ij}\right)$ where $\bar{\bOmega}_{ij}=\left[\begin{matrix}1 & \rho_{ij}\\ \rho_{ij} & 1  \end{matrix}\right]$. Using the Definition \ref{def:1}, it turns out that the optimal direction for $\rho_{ij}$ is $\bu=\left(\frac{1}{\sqrt{2}}, \frac{1}{\sqrt{2}}\right)^{\prime}$. Therefore, we project $\bX_{ij}$ along $\bu$ and we obtain the variable $X_{\bu}=\bu^{\prime}\bX_{ij}$ such that $X_{\bu}\sim\mathcal{ESD}_1\left(\alpha, 0, 1+\rho_{ij}\right)$. Now, since $X_{\bu}$ is a univariate random variable we can apply the method of \cite{mcculloch.1986} to initialise the scale of a univariate ESD. 
%
\clearpage
\newpage
\section*{Appendix C: synthetic data results}
\label{app:appendix_tables}
%
%
\begin{table}[!ht]
\captionsetup{font={footnotesize}}
\setlength{\tabcolsep}{5pt}
\begin{center}
\resizebox{0.6\columnwidth}{!}{
\begin{tabular}{rccccccc}
\\
\toprule
& &\multicolumn{3}{c}{$n=500$}&\multicolumn{3}{c}{$n=2000$}\\
 \cmidrule(lr){1-1}\cmidrule(lr){2-2}\cmidrule(lr){3-5}\cmidrule(lr){6-8}
          Par. & True  & BIAS & SSD & ECP & BIAS & SSD & ECP \\ 
          \cmidrule(lr){1-1}\cmidrule(lr){2-2}\cmidrule(lr){3-5}\cmidrule(lr){6-8}
          $\alpha$  	&{\bf1.70} &-0.0075  &0.0996  &0.7970   &-0.0041  &0.0535  &0.7650  \\
    $\xi_{1}$		&0.00  &0.0016  &0.0443  &0.9380     &0.0013  &0.0201  &0.9500  \\
     $\xi_{2}$ 		&0.00 &0.0088  &0.0841  &0.9440     &0.0021  &0.0385  &0.9590  \\
     $\omega_{11}$ 	&0.50 &0.0112  &0.2904  &0.6030   &-0.0046  &0.0605  &0.5330  \\
     $\omega_{22}$ 	&2.00 &-0.0409  &0.3599  &0.6870   &-0.0059  &0.1439  &0.6910  \\
     $\omega_{12}$ 	&0.90 &-0.1044  &0.2841  &0.8090   &-0.0369  &0.1680  &0.8280  \\
      
     \hline
Par. & True  & BIAS & SSD & ECP & BIAS & SSD & ECP \\ 
          \cmidrule(lr){1-1}\cmidrule(lr){2-2}\cmidrule(lr){3-5}\cmidrule(lr){6-8}
          $\alpha$        &{\bf1.90} &-0.0315  &0.0876  &0.8750  &-0.0141  &0.0626  &0.8760  \\
     $\xi_{1}$        &0.00  &-0.0003  &0.0444  &0.9390 &0.0010  &0.0209  &0.9440  \\
     $\xi_{2}$        &0.00  &0.0029  &0.0891  &0.9240 &0.0005  &0.0401  &0.9510  \\
     $\omega_{11}$  &0.50 &-0.0069  &0.2040  &0.6480  &0.0045  &0.4682  &0.6120  \\
     $\omega_{22}$  &2.00 &-0.0412  &0.3563  &0.7700 &0.0002  &0.4357  &0.7380  \\
     $\omega_{12}$  &0.90 &-0.1862  &0.3717  &0.7530  &-0.1373  &0.3110  &0.7730  \\  
      
     \hline
Par. & True  & BIAS & SSD & ECP & BIAS & SSD & ECP \\ 
          \cmidrule(lr){1-1}\cmidrule(lr){2-2}\cmidrule(lr){3-5}\cmidrule(lr){6-8}
          $\alpha$        &{\bf1.95} &-0.0628  &0.0974  &0.8580  &-0.0310  &0.0586  &0.8580  \\
     $\xi_{1}$        &0.00  &0.0006  &0.0436  &0.9360  &0.0047  &0.1060  &0.9490  \\
     $\xi_{2}$        &0.00  &0.0038  &0.0862  &0.9220  &-0.0008  &0.0645  &0.9520  \\
     $\omega_{11}$  &0.50 &0.0111  &0.5107  &0.6270  &-0.0014  &0.2008  &0.6310  \\
     $\omega_{22}$  &2.00 &-0.0688  &0.4903  &0.7650  &-0.0272  &0.3358  &0.7580  \\
     $\omega_{12}$  &0.90 &-0.2227  &0.4907  &0.6980    &-0.2181  &0.4444  &0.7190  \\ 
\bottomrule
\end{tabular}}
\caption{\footnotesize{Bias (BIAS), sample standard deviation (SSD), and empirical coverage probability (ECP) at the 95\% confidence level for the locations $\boldsymbol{\xi}=\left(\xi_{1},\xi_2,\dots,\xi_d\right)$, scale matrix $\boldsymbol{\Omega}=\left\{\omega_{ij}\right\}$, with $i,j=1,2,\dots,d$ and $i\leq j$ and characteristic exponent $\alpha$ of the bivariate Elliptical Stable distribution of dimension $m=2$. The results reported above are obtained using $1,000$ replications for three different values of $\alpha=\left(1.70,1.90,1.95\right)$.}}
\label{tab:table_coverage_esd_dim2}
\end{center}
\end{table}

\begin{sidewaystable}[p]
\captionsetup{font={footnotesize}}
\setlength{\tabcolsep}{5pt}
\begin{center}
\resizebox{1.0\columnwidth}{!}{
\begin{tabular}{rrccccccrrccccccrrcccccc}
\\
\toprule
& &\multicolumn{3}{c}{$n=500$}&\multicolumn{3}{c}{$n=2000$}& &&\multicolumn{3}{c}{$n=500$}&\multicolumn{3}{c}{$n=2000$}& &&\multicolumn{3}{c}{$n=500$}&\multicolumn{3}{c}{$n=2000$}\\
 \cmidrule(lr){1-1}\cmidrule(lr){2-2}\cmidrule(lr){3-5}\cmidrule(lr){6-8}\cmidrule(lr){9-9}\cmidrule(lr){10-10}\cmidrule(lr){11-13}\cmidrule(lr){14-16}\cmidrule(lr){17-17}\cmidrule(lr){18-18}\cmidrule(lr){19-21}\cmidrule(lr){22-24}
Par. & True  & BIAS & SSD & ECP & BIAS & SSD & ECP &Par. & True  & BIAS & SSD & ECP & BIAS & SSD & ECP&Par. & True  & BIAS & SSD & ECP & BIAS & SSD & ECP\\
 \cmidrule(lr){1-1}\cmidrule(lr){2-2}\cmidrule(lr){3-5}\cmidrule(lr){6-8}\cmidrule(lr){9-9}\cmidrule(lr){10-10}\cmidrule(lr){11-13}\cmidrule(lr){14-16}\cmidrule(lr){17-17}\cmidrule(lr){18-18}\cmidrule(lr){19-21}\cmidrule(lr){22-24}
 $\alpha$  	&{\bf1.70}   &-0.0055  &0.0613  &0.7958  &-0.0001  &0.0352  &0.8006 &$\alpha$ &{\bf1.90}      &-0.0387  &0.1539  &0.9891  &-0.0092  &0.0571  &0.7480& $\alpha$	&{\bf1.95}      &-0.0662  &0.2190  &0.9854  &-0.0278  &0.1222  &0.9910  \\
\multicolumn{24}{c}{\bf Locations}\\
     $\mu_{1}$		&0.00  &-0.0008  &0.0281  &0.9409  &0.0010  &0.0470  &0.9376  &$\mu_{1}$&0.00  &0.0030  &0.0920  &0.9659  &-0.0007  &0.0142  &0.9540& $\mu_{1}$ &0.00  &-0.0011  &0.0284  &0.9610  &-0.0010  &0.0140  &0.9550  \\
     $\mu_{2}$ 		&0.00  &0.0011  &0.0406  &0.9479  &-0.0001  &0.0625  &0.9315  &$\mu_{2}$&0.00  &0.0052  &0.0768  &0.9628  &-0.0028  &0.0192  &0.9560  &$\mu_{2}$&0.00  &0.0012  &0.0415  &0.9463  &-0.0029  &0.0193  &0.9565  \\
     $\mu_{3}$ 		&0.00  &-0.0024  &0.0533  &0.9550  &0.0021  &0.0984  &0.9406  &$\mu_{3}$&0.00  &0.0123  &0.3309  &0.9566  &-0.0012  &0.0275  &0.9580  &$\mu_{3}$&0.00  &-0.0051  &0.0572  &0.9382  &-0.0019  &0.0274  &0.9475  \\
     $\mu_{4}$ 		&0.00  &-0.0055  &0.0785  &0.9409  &0.0177  &0.5072  &0.9527  &$\mu_{4}$&0.00  &0.0130  &0.3594  &0.9395  &0.0029  &0.1348  &0.9570  &$\mu_{4}$&0.00  &-0.0162  &0.2342  &0.9431  &-0.0025  &0.0400  &0.9520  \\
     $\mu_{5}$ 		&0.00  &0.0023  &0.1149  &0.9389  &0.0278  &1.0030  &0.9436  &$\mu_{5}$&0.00  &0.0306  &0.5551  &0.9333  &0.0000  &0.2171  &0.9410  &$\mu_{5}$&0.00  &-0.0221  &0.4252  &0.9496  &-0.0094  &0.0596  &0.9385  \\
\multicolumn{24}{c}{\bf Scale, diagonals}\\
     $\omega_{11}$ 	&0.5000&-0.0047  &0.0312  &0.7688  &-0.0015  &0.0160  &0.8187  &$\omega_{11}$&0.5000&-0.0062  &0.0297  &0.7628  &-0.0024  &0.0160  &0.8150  &$\omega_{11}$&0.5000&-0.0040  &0.0293  &0.7724  &-0.0031  &0.0136  &0.8186  \\
     $\omega_{22}$ 	&0.7071&0.0040  &0.0393  &0.7678  &0.0019  &0.0214  &0.7795  &$\omega_{22}$&0.7071&-0.0021  &0.0375  &0.7736  &0.0000  &0.0180  &0.7790  &$\omega_{22}$&0.7071&0.0004  &0.0416  &0.7333  &-0.0012  &0.0184  &0.7676  \\
     $\omega_{33}$ 	&1.0000&-0.0058  &0.0547  &0.7247  &-0.0033  &0.0316  &0.7402  &$\omega_{33}$&1.0000&-0.0063  &0.0499  &0.7271  &-0.0034  &0.0249  &0.7420  &$\omega_{33}$&1.0000&-0.0005  &0.0516  &0.7561  &-0.0033  &0.0254  &0.7481  \\
     $\omega_{44}$ 	&1.4142&0.0022  &0.0801  &0.7337  &0.0040  &0.0479  &0.7422  &$\omega_{44}$&1.4142&-0.0021  &0.0779  &0.7876  &0.0017  &0.0367  &0.7520  &$\omega_{44}$&1.4142&-0.0010  &0.0728  &0.7577  &0.0012  &0.0409  &0.7271  \\
     $\omega_{55}$ 	&2.0000&-0.0091  &0.1144  &0.7047  &0.0011  &0.0681  &0.7382  &$\omega_{55}$&2.0000&-0.0123  &0.1070  &0.7426  &0.0005  &0.0523  &0.7690  &$\omega_{55}$&2.0000&-0.0042  &0.1085  &0.7236  &0.0021  &0.0493  &0.7661  \\
     \multicolumn{24}{c}{\bf Scale, off--diagonals}\\
     $\omega_{12}$   &0.7071&-0.0171  &0.1312  &0.9650  &-0.0080  &0.0706  &0.9778  &$\omega_{12}$&0.7071&-0.0260  &0.1174  &0.9643  &-0.0043  &0.0577  &0.9890  &$\omega_{12}$ &0.7071&-0.0276  &0.1349  &0.9496  &-0.0043  &0.0602  &0.9880  \\
     $\omega_{13}$   &0.8000&-0.0490  &0.1764  &0.9469  &-0.0219  &0.0983  &0.9748  &$\omega_{13}$&0.8000&-0.0751  &0.1486  &0.9271  &-0.0155  &0.0703  &0.9870  &$\omega_{13}$&0.8000&-0.0783  &0.1569  &0.9154  &-0.0231  &0.0670  &0.9835  \\
     $\omega_{14}$   &0.00&0.0124  &0.1292  &0.9269  &0.0071  &0.0657  &0.9275  &$\omega_{14}$&0.00&0.0123  &0.1158  &0.9519  &0.0056  &0.0615  &0.9600  &$\omega_{14}$&0.00&0.0183  &0.1326  &0.9463  &0.0044  &0.0574  &0.9550  \\
     $\omega_{15}$   &0.00&0.0178  &0.1456  &0.8859  &0.0085  &0.0724  &0.8751  &$\omega_{15}$&0.00&0.0266  &0.1506  &0.8992  &0.0060  &0.0669  &0.8900  &$\omega_{15}$&0.00&0.0226  &0.1373  &0.8862  &-0.0034  &0.0548  &0.9100  \\
     $\omega_{23}$   &0.5657&-0.0167  &0.1558  &0.9289  &0.0010  &0.0841  &0.9527  & $\omega_{23}$&0.5657&-0.0318  &0.1193  &0.9287  &0.0018  &0.0631  &0.9750  &$\omega_{23}$&0.5657&-0.0339  &0.1248  &0.9171 &-0.0064  &0.0614  &0.9805  \\
     $\omega_{24}$   &0.00&0.0103  &0.1168  &0.9109  &0.0050  &0.0708  &0.8207  &$\omega_{24}$&0.00&0.0116  &0.1151  &0.9442  &0.0010  &0.0613  &0.8920  &$\omega_{24}$&0.00&0.0146  &0.1123  &0.9496  &0.0019  &0.0632  &0.9220  \\
     $\omega_{25}$   &0.00&0.0252  &0.1336  &0.8749  &0.0100  &0.0723  &0.8258  &$\omega_{25}$&0.00&0.0182  &0.1184  &0.9240  &0.0039  &0.0634  &0.8780  &$\omega_{25}$&0.00&0.0168  &0.1233  &0.9268  &-0.0036  &0.0684  &0.9280  \\
     $\omega_{34}$   &0.00&0.0101  &0.1194  &0.9600  &0.0031  &0.0648  &0.9587  &$\omega_{34}$&0.00&-0.0013  &0.1107  &0.9674  &0.0014  &0.0538  &0.9750  &$\omega_{34}$ &0.00&0.0068  &0.1100  &0.9659  &-0.0006  &0.0507  &0.9805  \\
     $\omega_{35}$   &0.00&0.0119  &0.1194  &0.9660  &0.0046  &0.0619  &0.9527  &$\omega_{35}$&0.00&0.0029  &0.1103  &0.9767 &0.0014  &0.0579  &0.9740  &$\omega_{35}$ &0.00&0.0055  &0.1053  &0.9707  &-0.0078  &0.0524  &0.9835  \\
     $\omega_{45}$   &0.9016&-0.1466  &0.2975  &0.9489      &-0.0253  &0.0981  &0.9778  &$\omega_{45}$&0.9016&-0.1425  &0.2389  &0.9659  &-0.0208  &0.0903  &0.9950  &$\omega_{45}$&0.9016&-0.1245  &0.1787  &0.9447  &-0.0416  &0.0724  &0.9910  \\
\bottomrule
\end{tabular}}
\caption{\footnotesize{Bias (BIAS), sample standard deviation (SSD), and empirical coverage probability (ECP) at the 95\% confidence level for the locations $\boldsymbol{\mu}=\left(\mu_{1},\mu_2,\dots,\mu_d\right)$, scale matrix $\boldsymbol{\Omega}=\left\{\omega_{ij}\right\}$, with $i,j=1,2,\dots,d$ and $i\leq j$ and characteristic exponent $\alpha$ of the Elliptical Stable distribution of dimension $m=5$. The results reported above are obtained using $1,000$ replications  for three different values of $\alpha=\left(1.70,1.90,1.95\right)$.}}
\label{tab:table_coverage_esd_dim5}
\end{center}
\end{sidewaystable}

%
\begin{figure}[!ht]
\captionsetup{font={footnotesize}}
\begin{center}

\subfloat[True]{\label{fig:smmsq_12_alpha_170_TRUE}\includegraphics[width=0.32\textwidth]{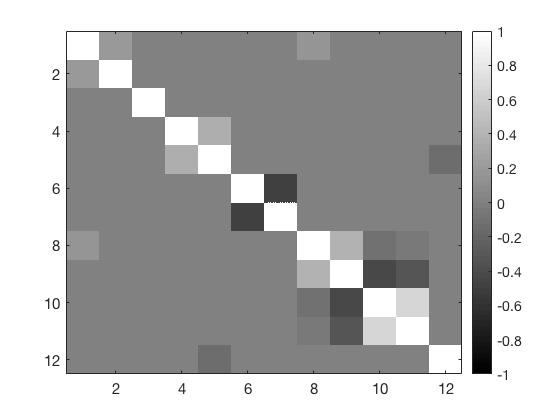}}
\subfloat[S--MMSQ, $\alpha=1.70$]{\label{fig:smmsq_12_alpha_170_SMMSQ}\includegraphics[width=0.32\textwidth]{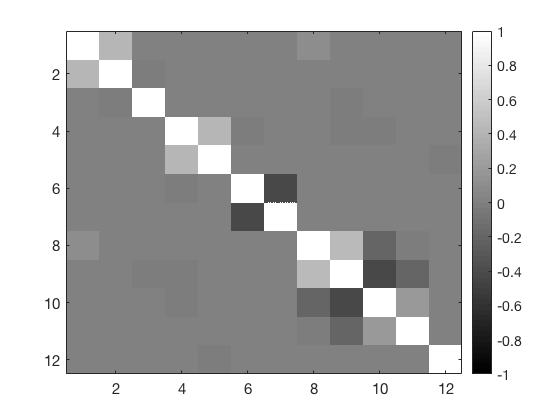}}
\subfloat[GLasso, $\alpha=1.70$]{\label{fig:smmsq_12_alpha_170_LASSO}\includegraphics[width=0.32\textwidth]{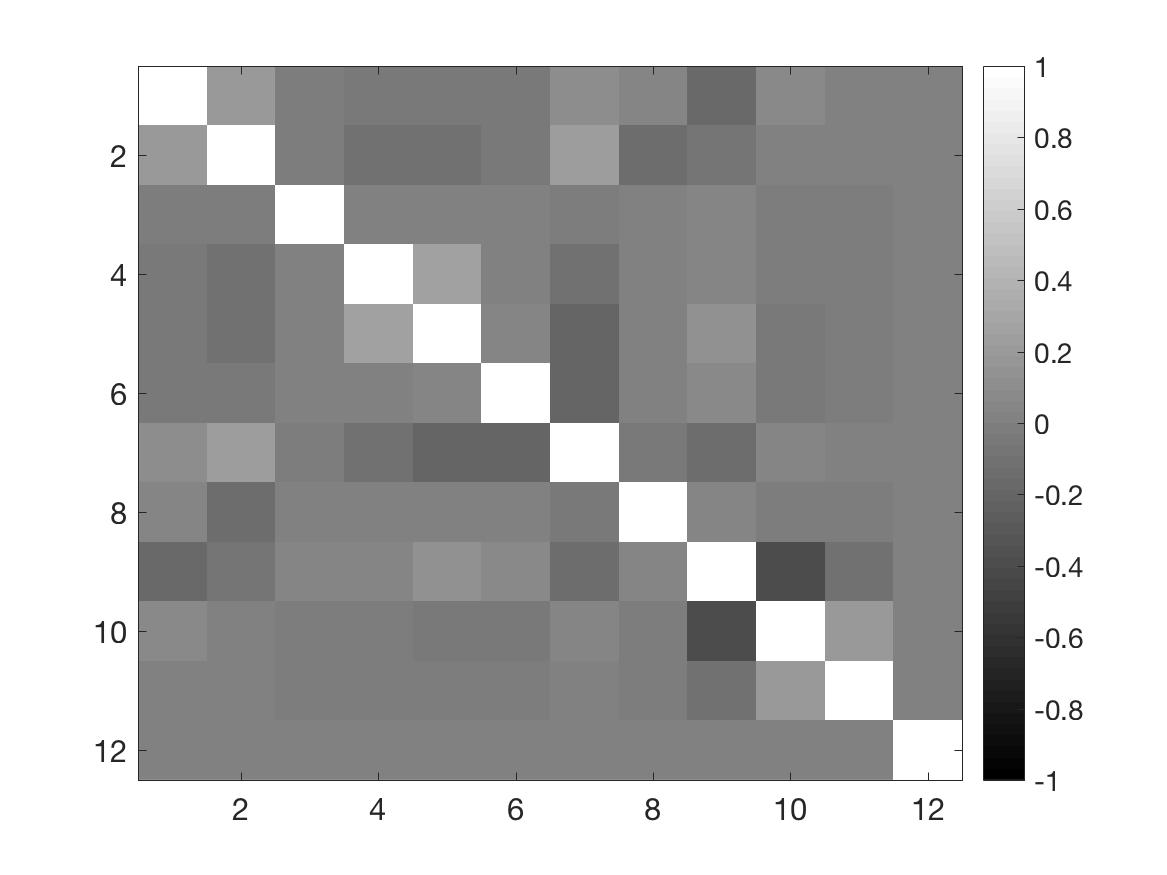}}\\
\subfloat[SCAD, $\alpha=1.70$]{\label{fig:smmsq_12_alpha_170_SCAD}\includegraphics[width=0.32\textwidth]{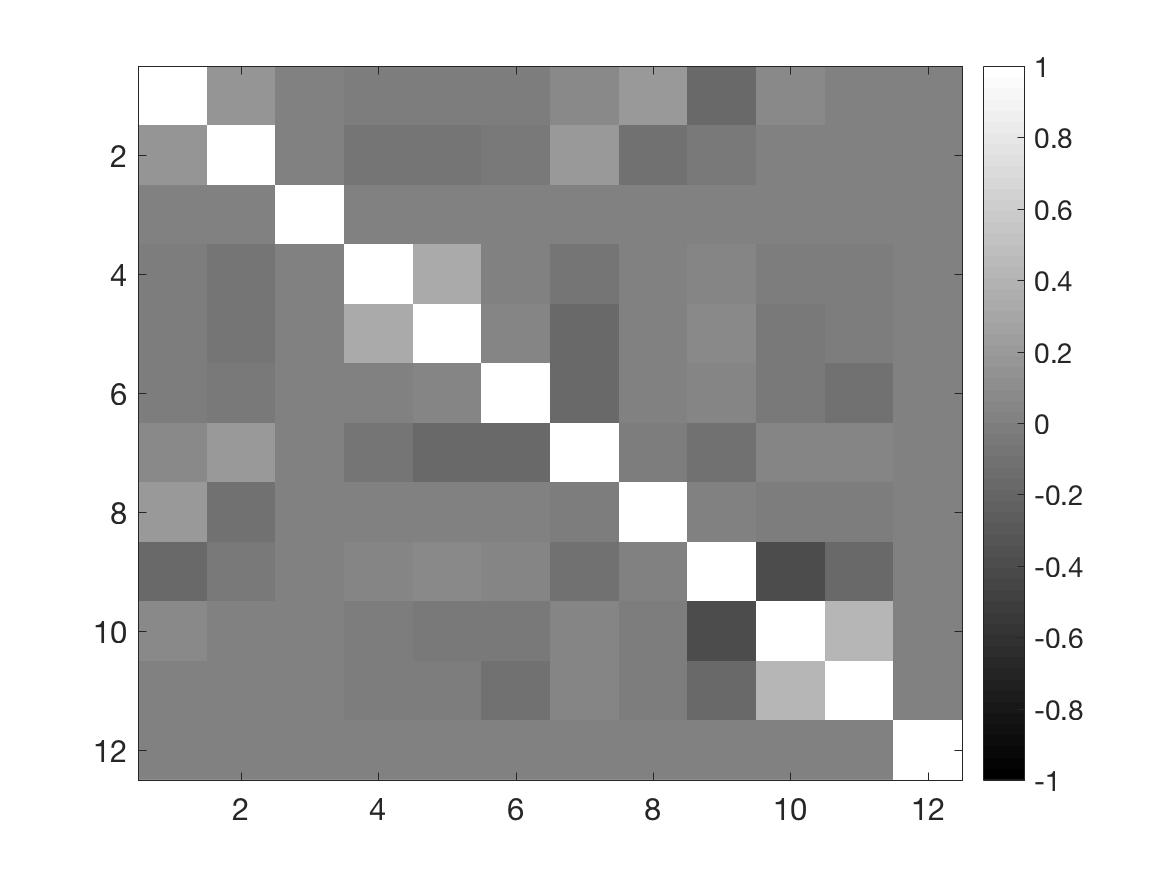}}
\subfloat[Adaptive Lasso, $\alpha=1.70$]{\label{fig:smmsq_12_alpha_170_ADAPT}\includegraphics[width=0.32\textwidth]{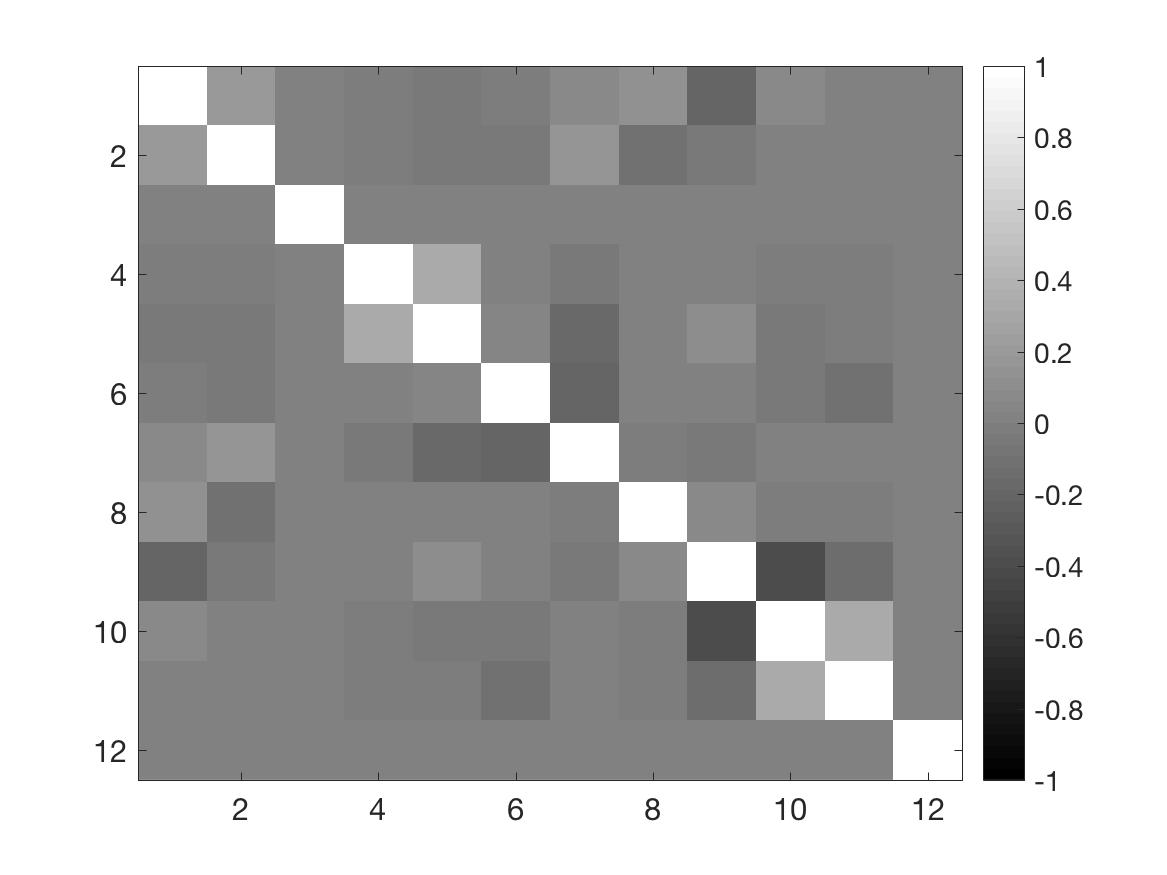}}
\subfloat[S--MMSQ, $\alpha=1.90$]{\label{fig:smmsq_12_alpha_190_SMMSQ}\includegraphics[width=0.32\textwidth]{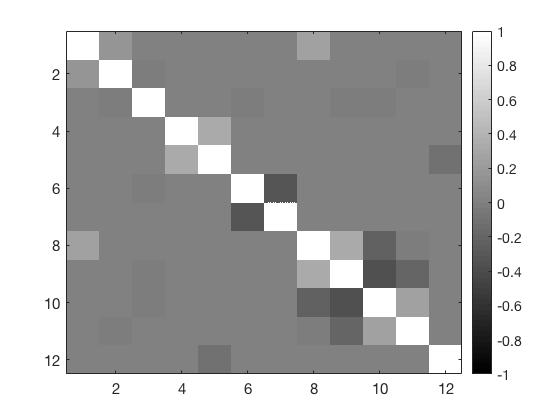}}\\
\subfloat[GLasso, $\alpha=1.90$]{\label{fig:smmsq_12_alpha_190_LASSO}\includegraphics[width=0.32\textwidth]{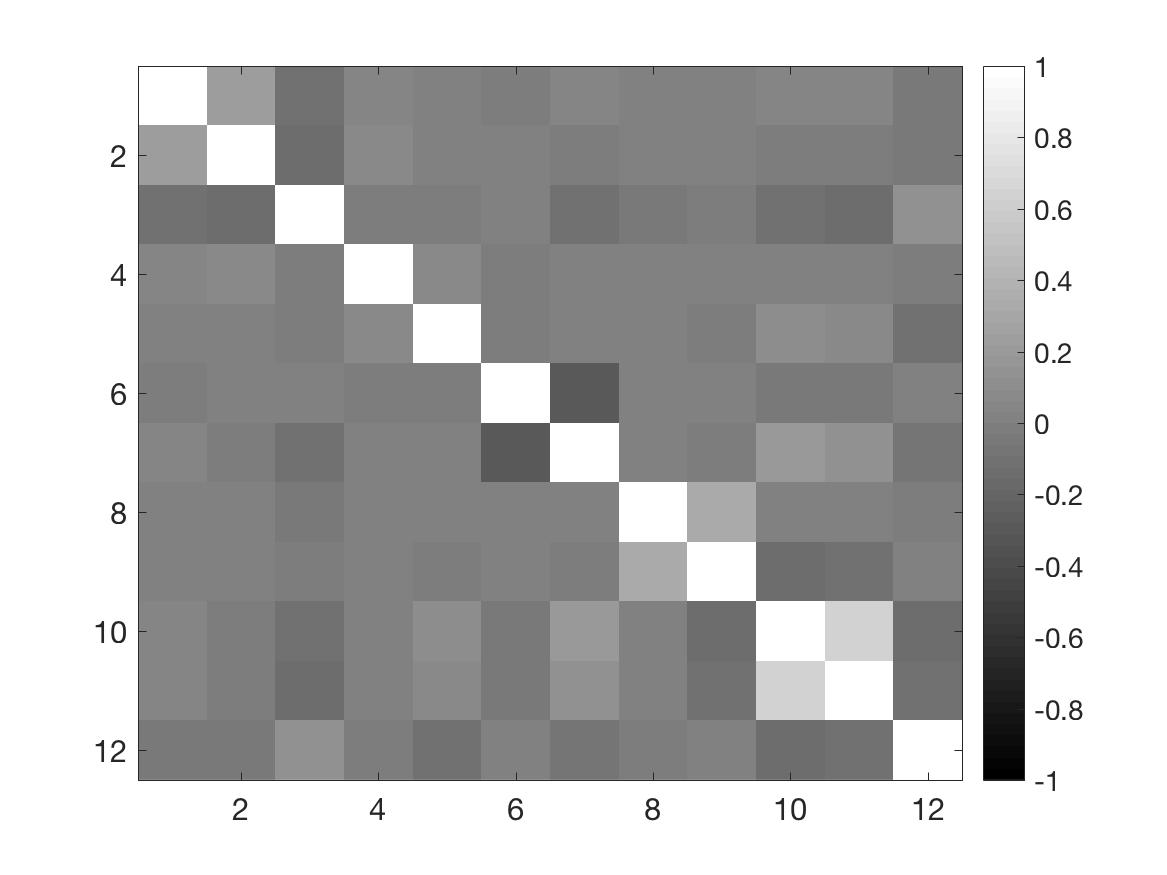}}
\subfloat[SCAD, $\alpha=1.90$]{\label{fig:smmsq_12_alpha_190_SCAD}\includegraphics[width=0.32\textwidth]{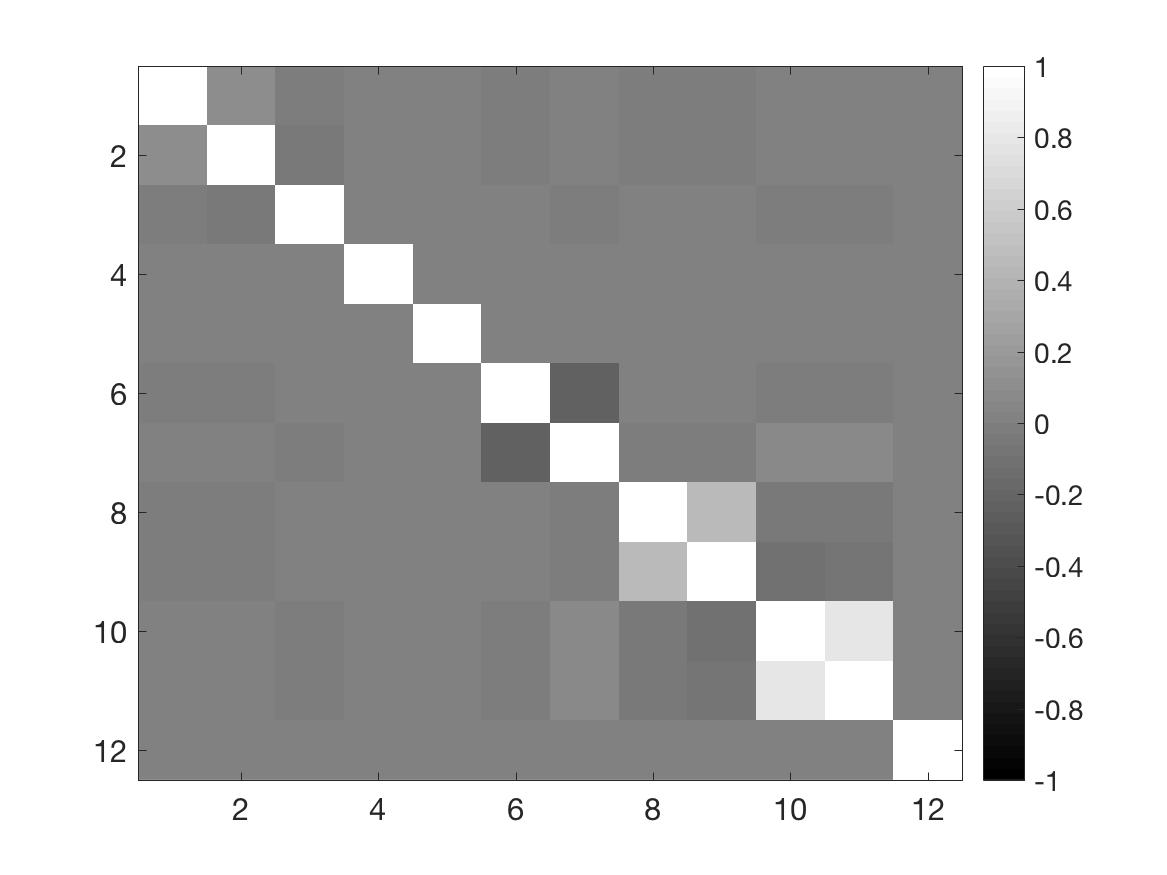}}
\subfloat[Adaptive Lasso, $\alpha=1.90$]{\label{fig:smmsq_12_alpha_190_ADAPT}\includegraphics[width=0.32\textwidth]{ESD_Ex1_Dim12_190_SCAD.jpg}}
\caption{\footnotesize{Band structure of the true and estimated scale matrices through S--MMSQ of the $12$--dimensional Elliptical Stable simulated experiment discussed in Section  \ref{sec:msq_synthetic_ex}, for $\alpha =\left(1.70,1.90\right)$ and sample size $n=200$.}}
\label{fig:smmsq_12_alpha_170_190}
\end{center}
\end{figure}
%
\begin{figure}[!ht]
\captionsetup{font={footnotesize}}
\begin{center}

\subfloat[True]{\label{fig:smmsq_12_alpha_195_TRUE}\includegraphics[width=0.32\textwidth]{MMSQ_True_dim12.jpg}}
\subfloat[S--MMSQ, $\alpha=1.95$]{\label{fig:smmsq_12_alpha_195_SMMSQ}\includegraphics[width=0.32\textwidth]{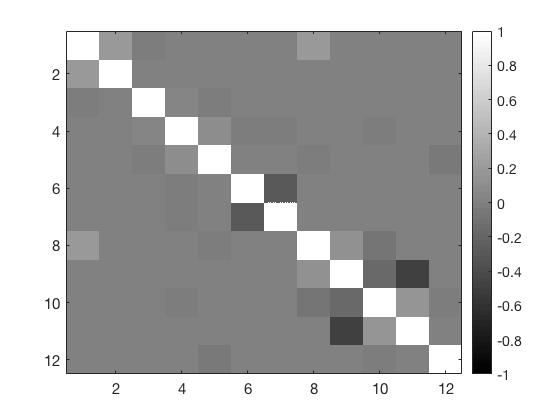}}
\subfloat[GLasso, $\alpha=1.95$]{\label{fig:smmsq_12_alpha_195_LASSO}\includegraphics[width=0.32\textwidth]{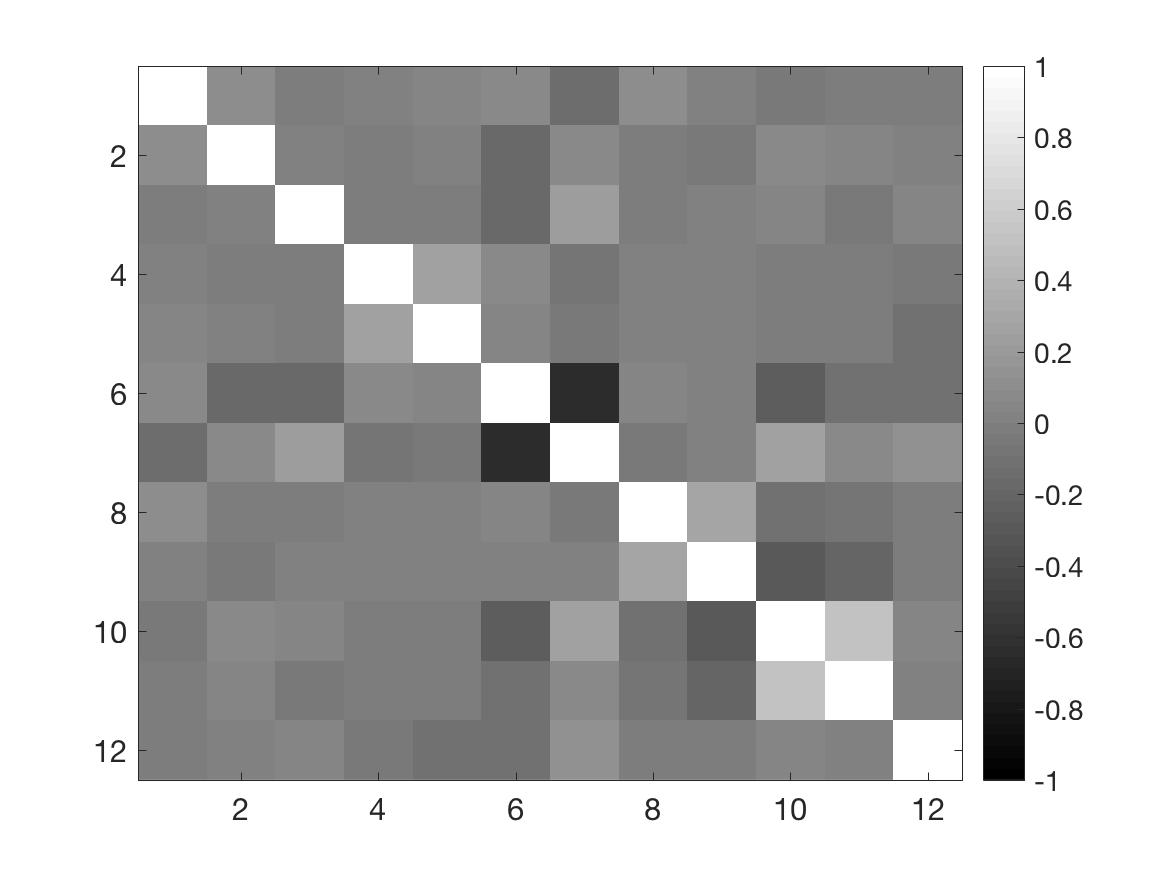}}\\
\subfloat[SCAD, $\alpha=1.95$]{\label{fig:smmsq_12_alpha_195_SCAD}\includegraphics[width=0.32\textwidth]{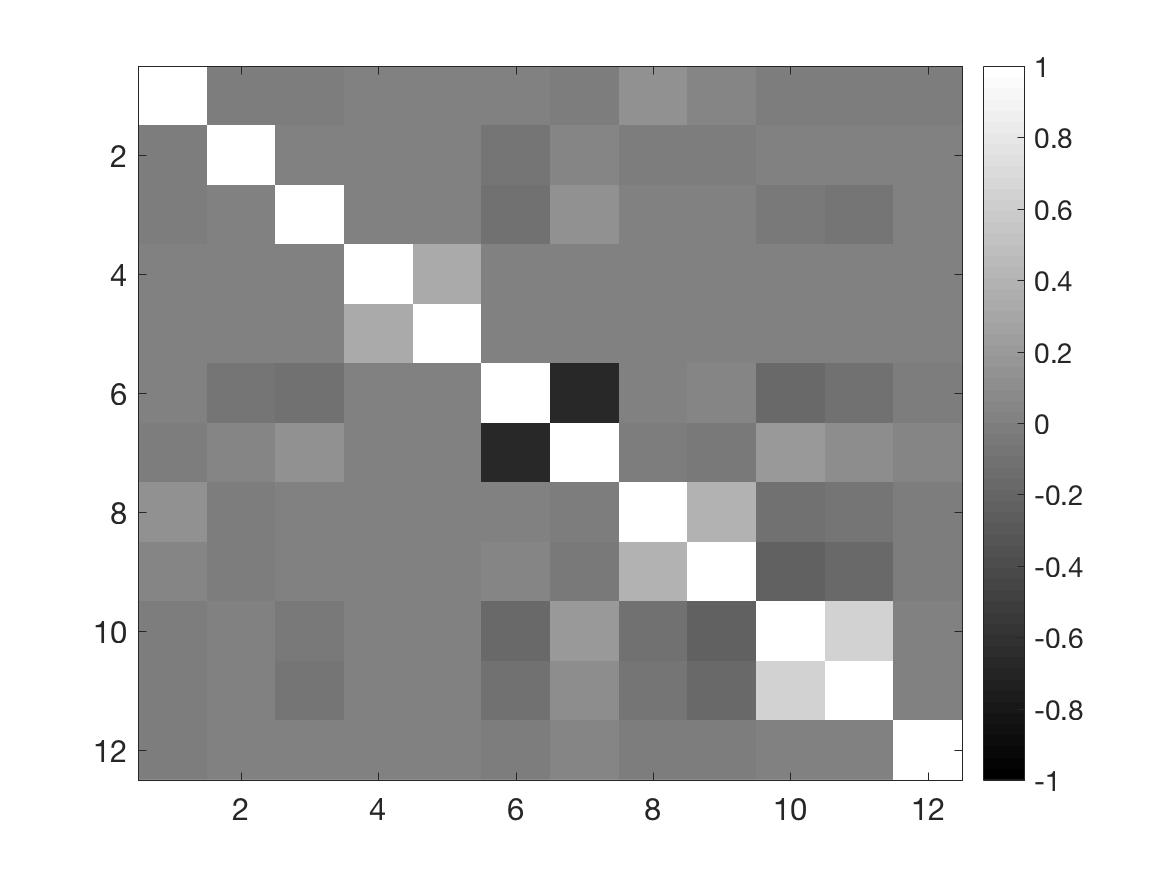}}
\subfloat[Adaptive Lasso, $\alpha=1.95$]{\label{fig:smmsq_12_alpha_195_ADAPT}\includegraphics[width=0.32\textwidth]{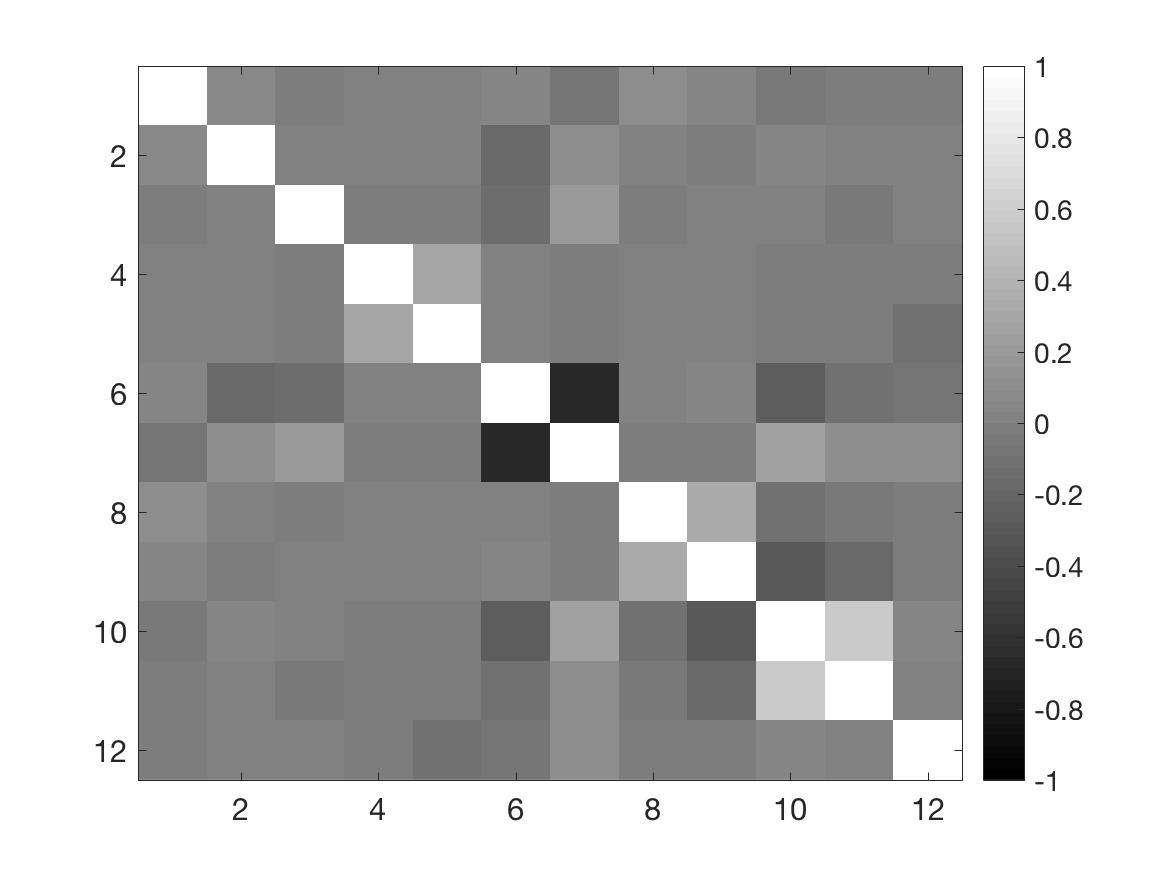}}
\subfloat[S--MMSQ, $\alpha=2.00$]{\label{fig:smmsq_12_alpha_200_SMMSQ}\includegraphics[width=0.32\textwidth]{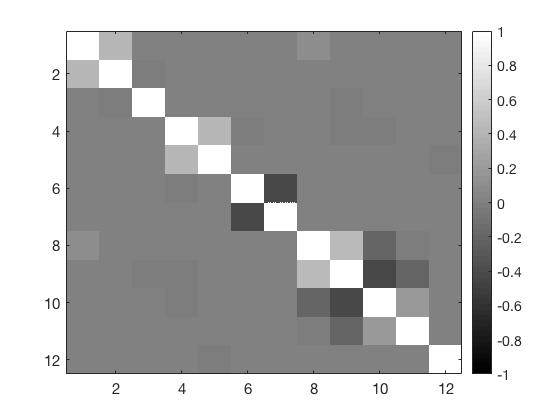}}\\
\subfloat[GLasso, $\alpha=2.00$]{\label{fig:smmsq_12_alpha_200_LASSO}\includegraphics[width=0.32\textwidth]{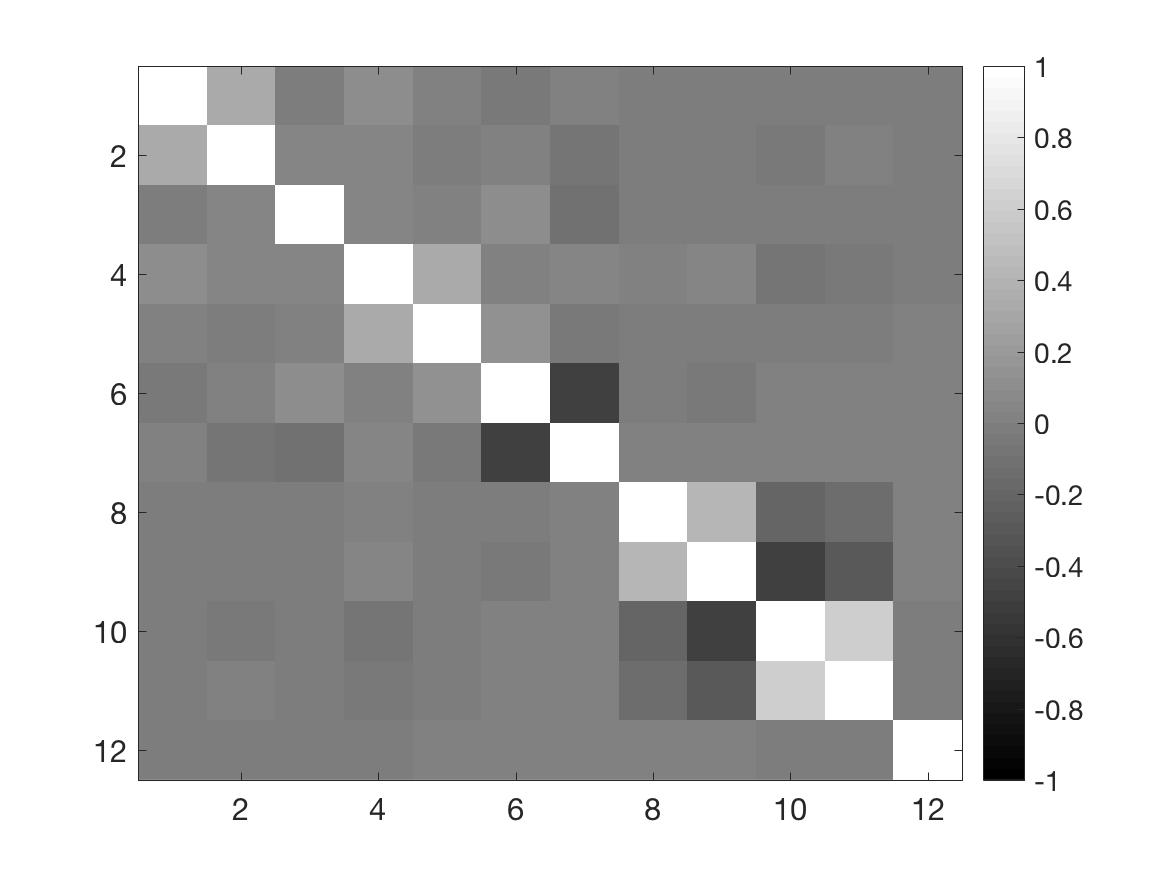}}
\subfloat[SCAD, $\alpha=2.00$]{\label{fig:smmsq_12_alpha_200_SCAD}\includegraphics[width=0.32\textwidth]{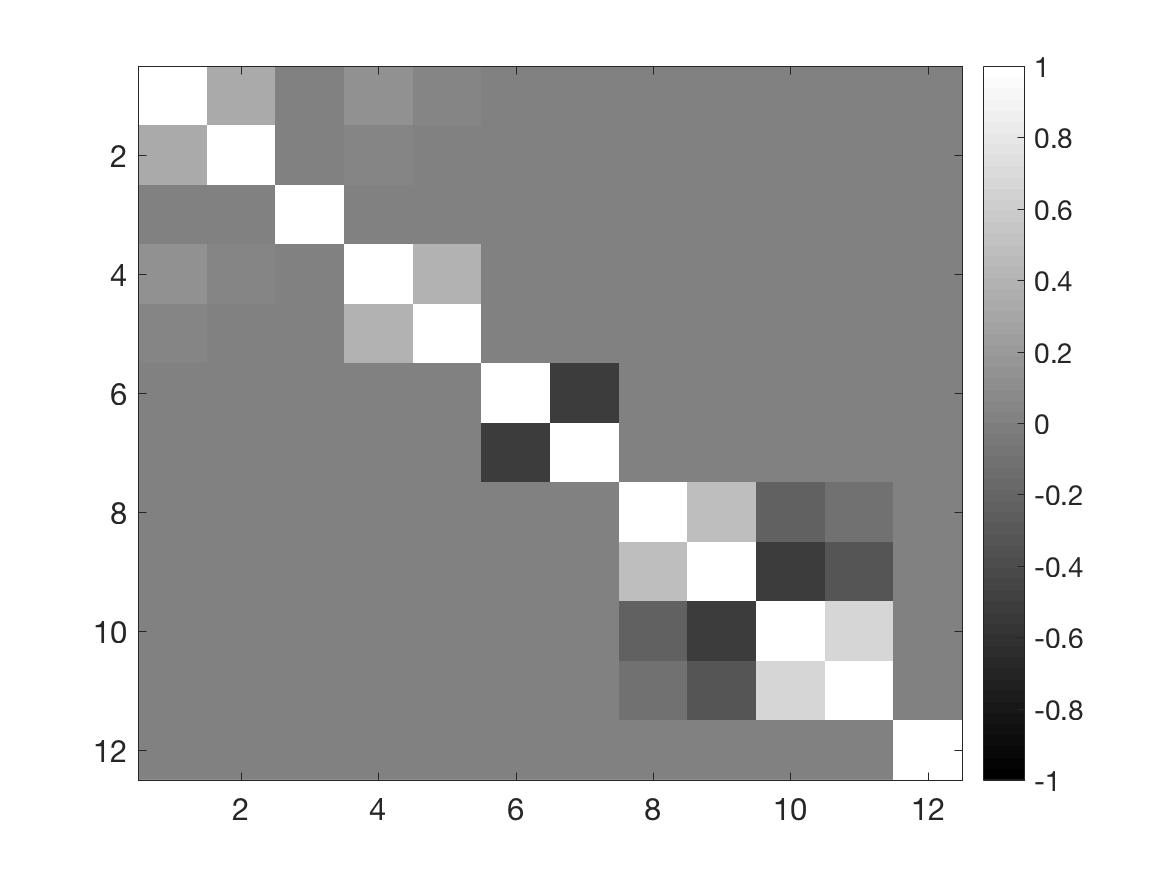}}
\subfloat[Adaptive Lasso, $\alpha=2.00$]{\label{fig:smmsq_12_alpha_200_ADAPT}\includegraphics[width=0.32\textwidth]{ESD_Ex1_Dim12_200_SCAD.jpg}}
\caption{\footnotesize{Band structure of the true and estimated scale matrices of the Elliptical Stable distribution experiment discussed in Section  \ref{sec:msq_synthetic_ex}, for $\alpha =\left(1.95,2.00\right)$ and sample size $n=200$.}}
\label{fig:smmsq_12_alpha_195_200}
\end{center}
\end{figure}
%
%
\begin{figure}[!ht]
\captionsetup{font={footnotesize}}
\begin{center}
\subfloat[GLASSO]{\label{fig:dim27_true}\includegraphics[width=0.4\textwidth]{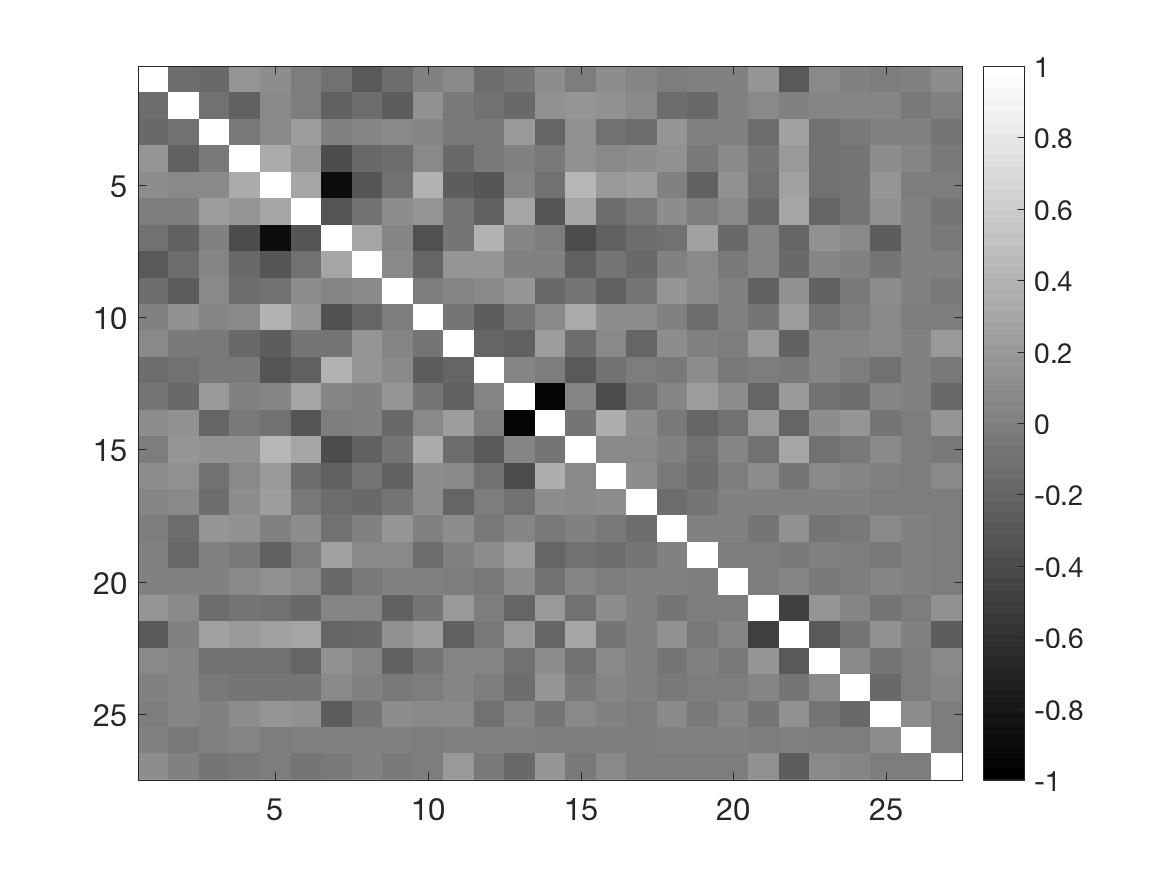}}\quad
\subfloat[S--MMSQ]{\label{fig:dim27}\includegraphics[width=0.4\textwidth]{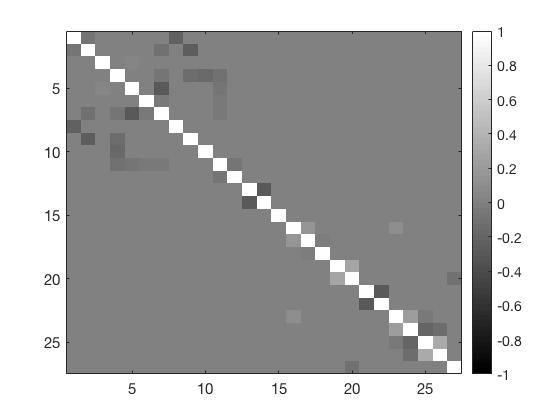}}\\
\subfloat[TRUE]{\label{fig:dim27}\includegraphics[width=0.4\textwidth]{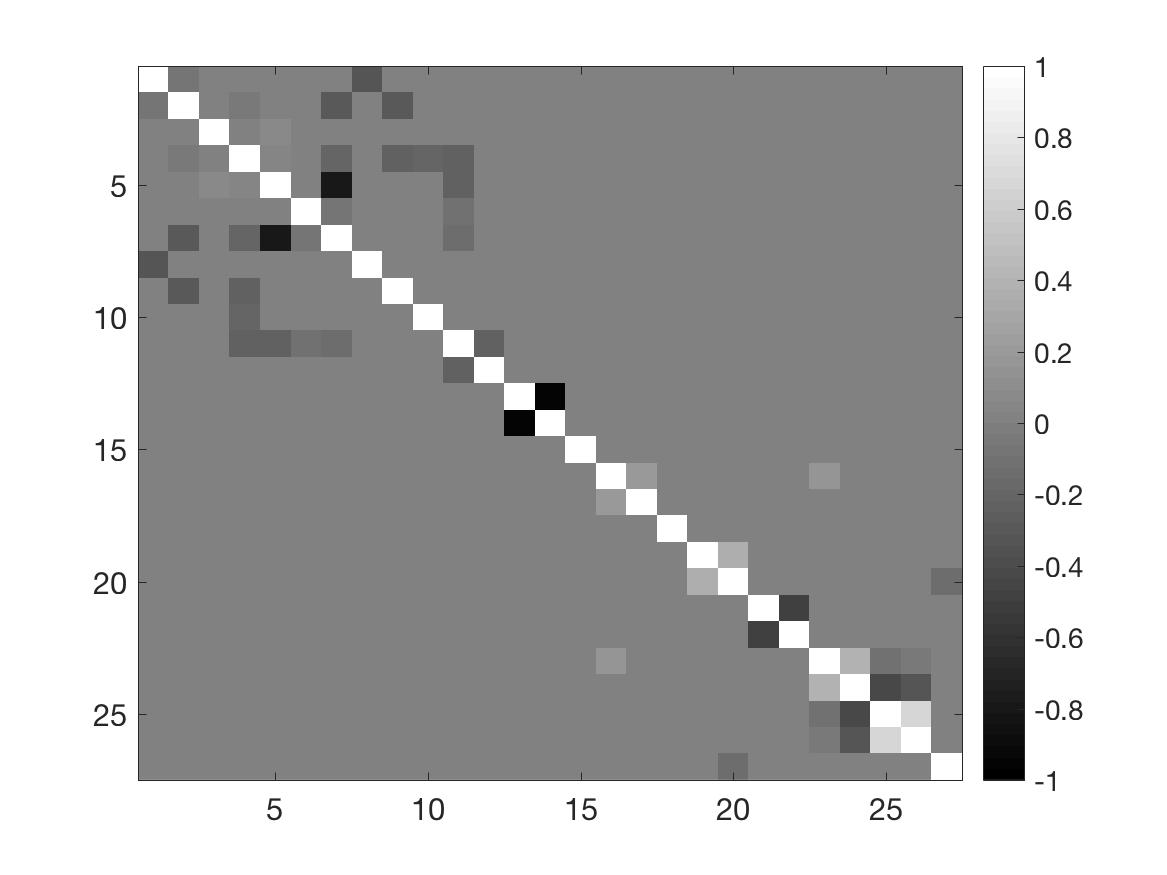}}\\
\subfloat[SCAD]{\label{fig:Austria_CoVaR}\includegraphics[width=0.4\textwidth]{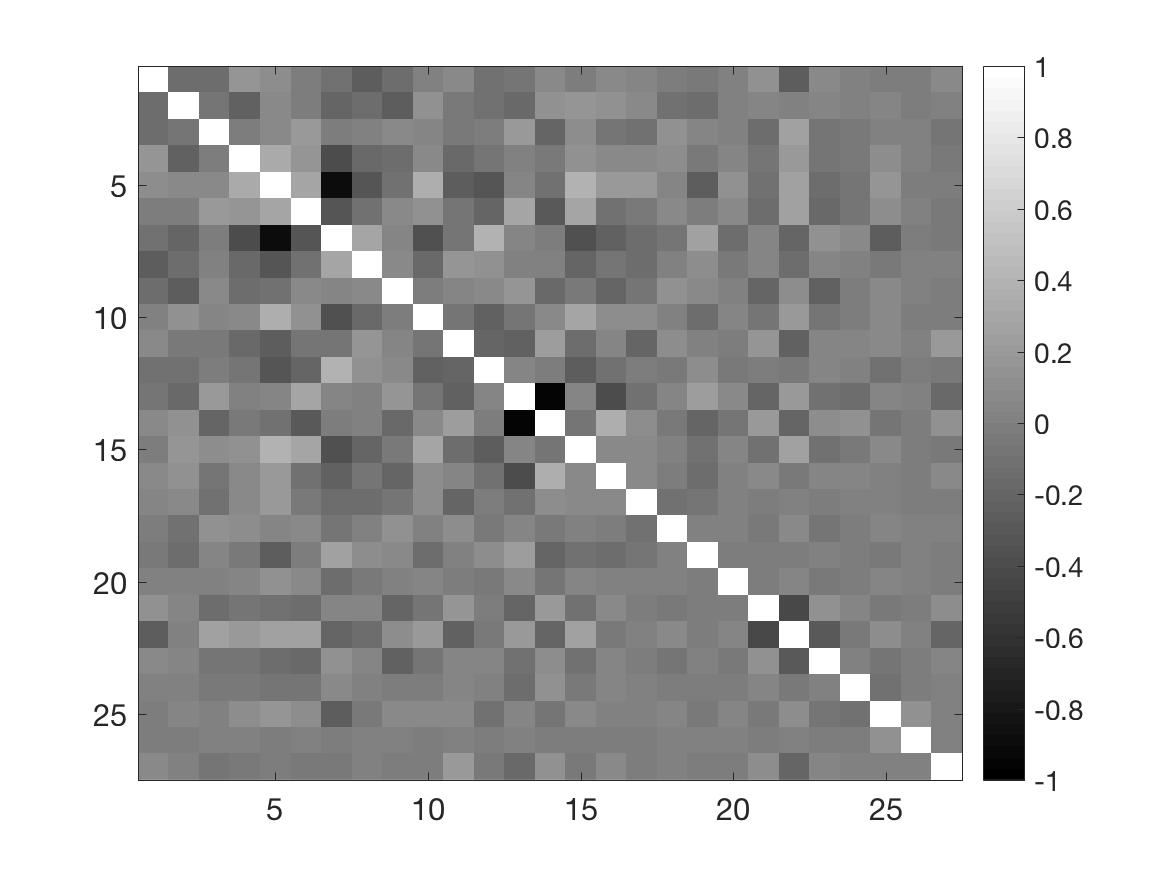}}\quad
\subfloat[Adaptive Lasso]{\label{fig:Austria_CoVaR}\includegraphics[width=0.4\textwidth]{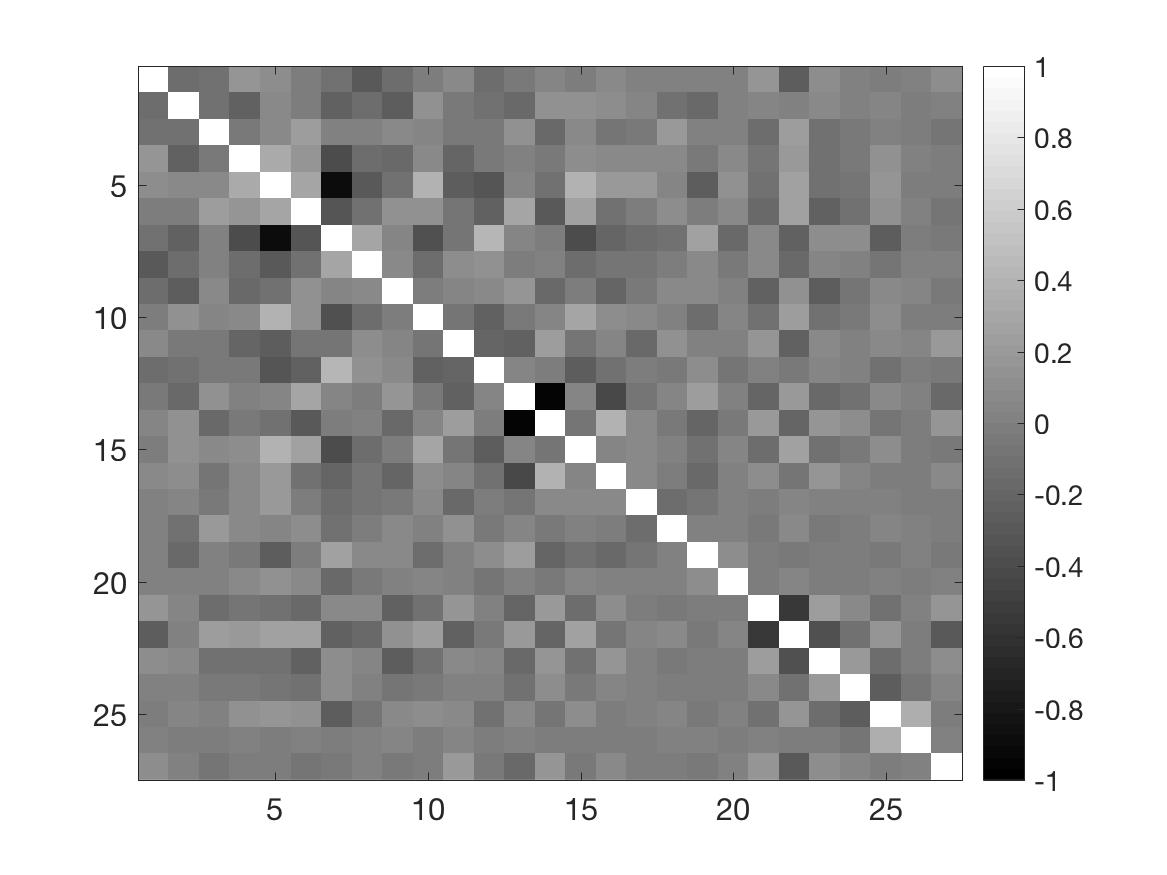}}\\
\caption{\footnotesize{Band structure of the true {\it (left)} and estimated {\it (right)} scale matrices through S--MMSQ of the $12$--dimensional Elliptical Stable simulated experiment discussed in Section  \ref{sec:msq_synthetic_ex}, for $\alpha =1.70$ and sample size $n=200$.}}
\label{fig:smmsq_27_alpha_170}
\end{center}
\end{figure}
%
%
\begin{figure}[!ht]
\captionsetup{font={footnotesize}}
\begin{center}
\subfloat[GLASSO]{\label{fig:dim27_true}\includegraphics[width=0.4\textwidth]{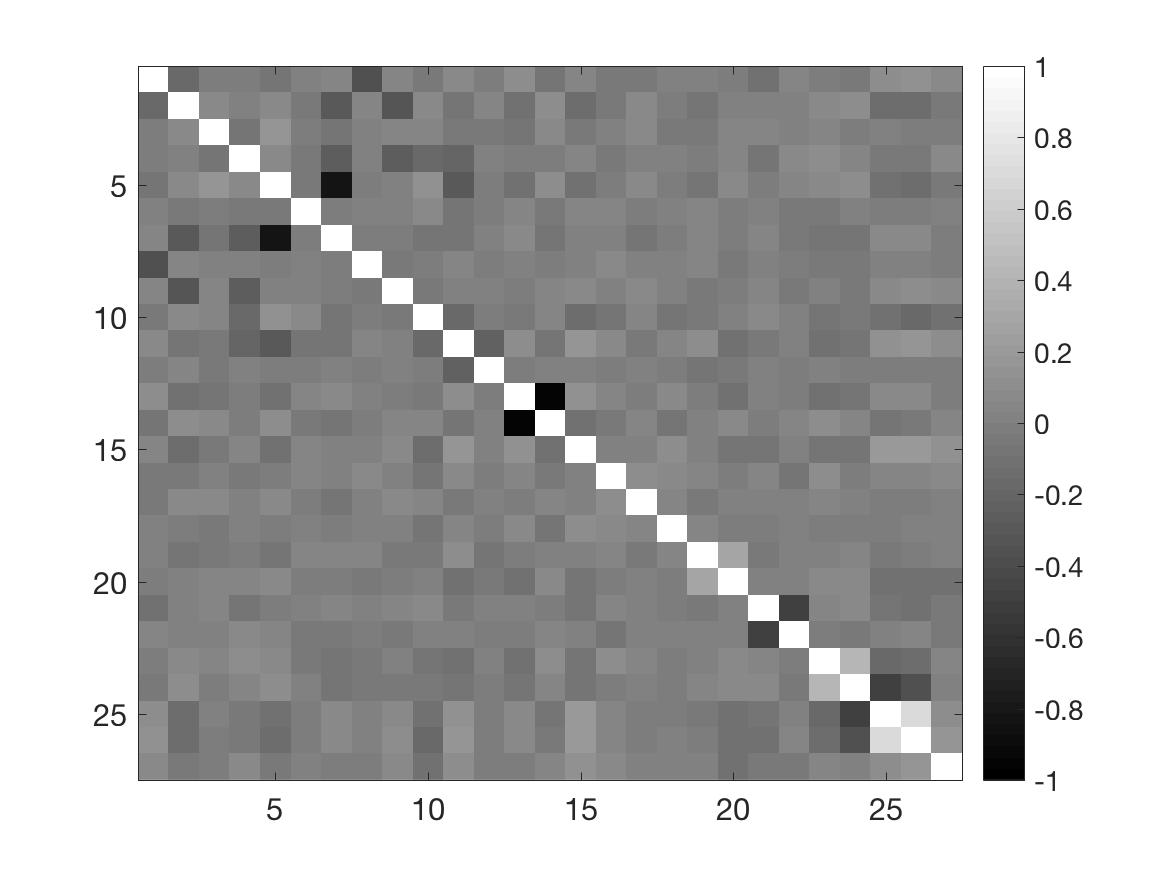}}\quad
\subfloat[S--MMSQ]{\label{fig:dim27}\includegraphics[width=0.4\textwidth]{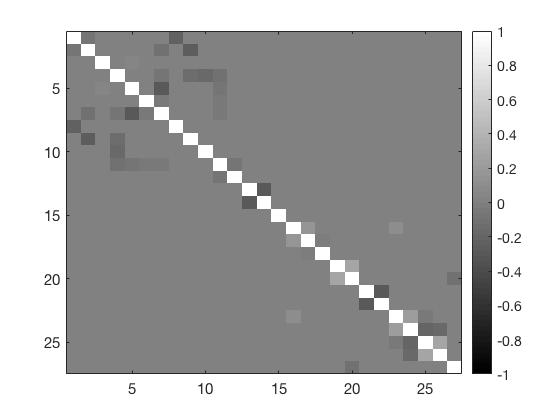}}\\
\subfloat[TRUE]{\label{fig:dim27}\includegraphics[width=0.4\textwidth]{ESD_Ex1_Dim27_TRUE.jpg}}\\
\subfloat[SCAD]{\label{fig:Austria_CoVaR}\includegraphics[width=0.4\textwidth]{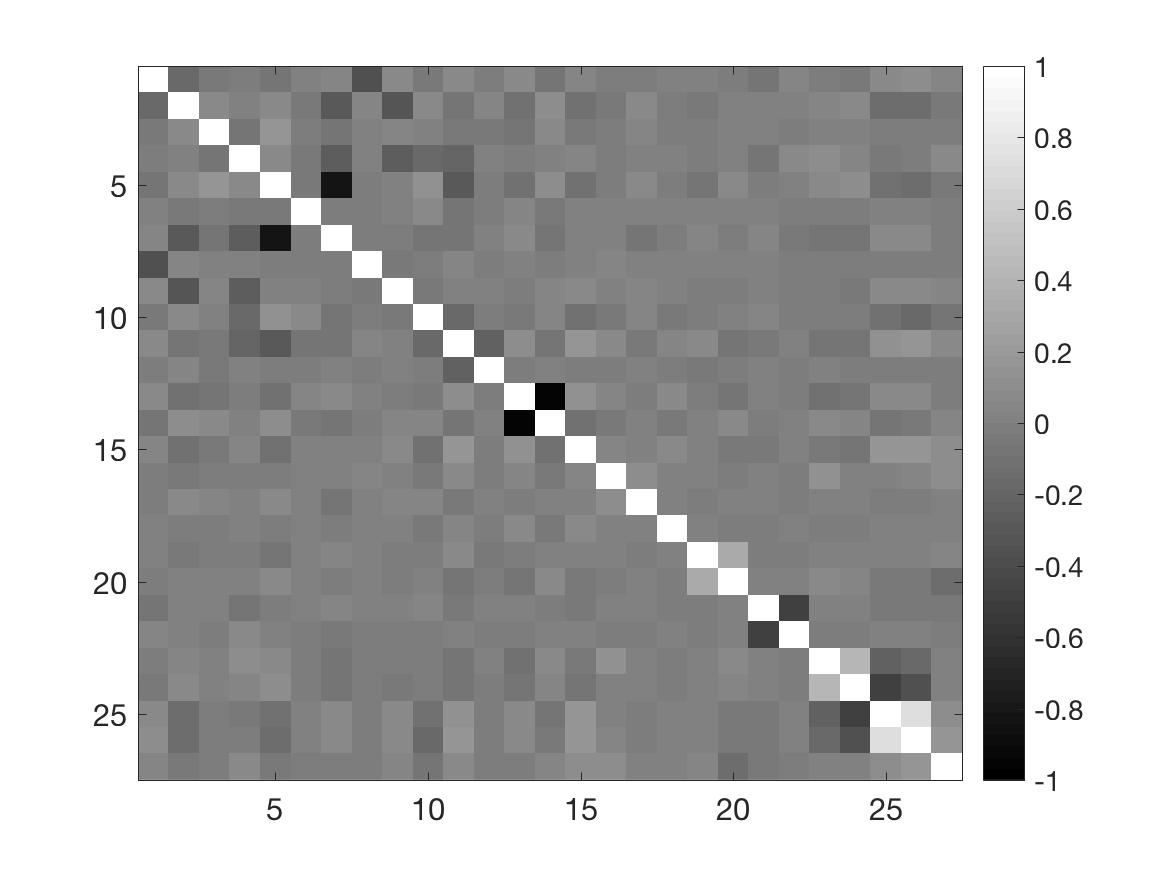}}\quad
\subfloat[Adaptive Lasso]{\label{fig:Austria_CoVaR}\includegraphics[width=0.4\textwidth]{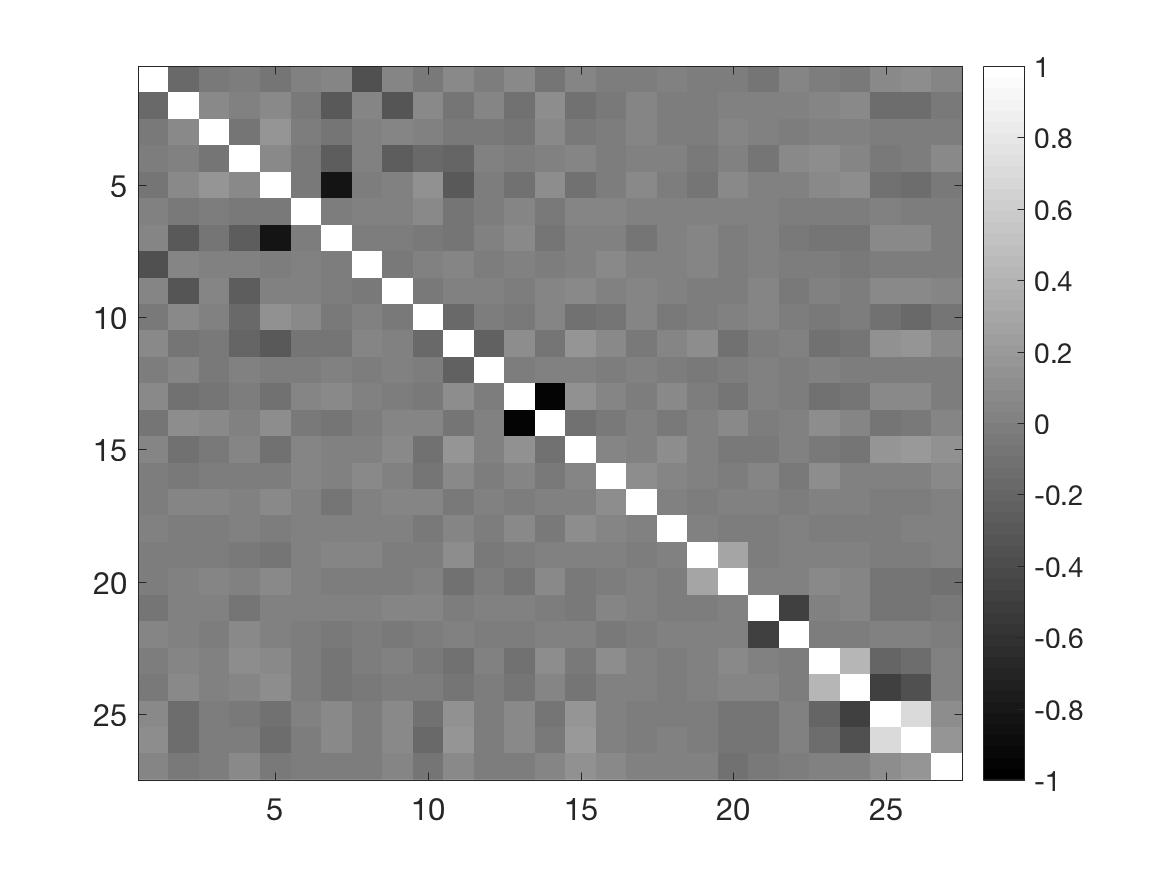}}\\
\caption{\footnotesize{Band structure of the true {\it (left)} and estimated {\it (right)} scale matrices through S--MMSQ of the $12$--dimensional Elliptical Stable simulated experiment discussed in Section  \ref{sec:msq_synthetic_ex}, for $\alpha =1.90$ and sample size $n=200$.}}
\label{fig:smmsq_27_alpha_190}
\end{center}
\end{figure}
%
%
\begin{figure}[!ht]
\captionsetup{font={footnotesize}}
\begin{center}
\subfloat[GLASSO]{\label{fig:dim27_true}\includegraphics[width=0.4\textwidth]{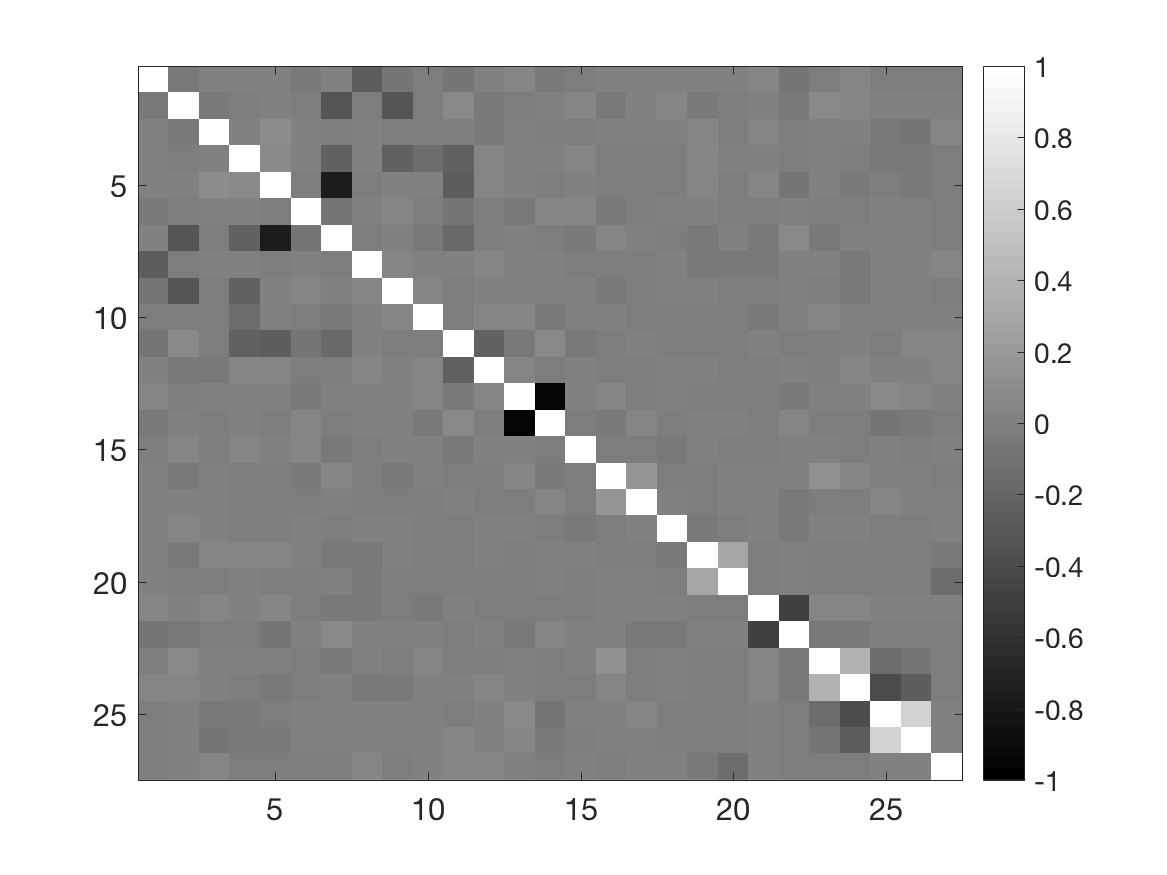}}\quad
\subfloat[S--MMSQ]{\label{fig:dim27}\includegraphics[width=0.4\textwidth]{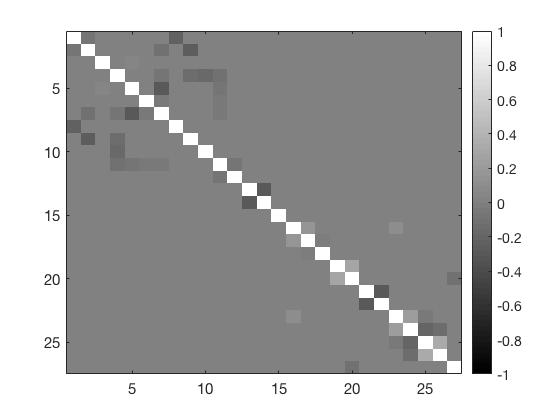}}\\
\subfloat[TRUE]{\label{fig:dim27}\includegraphics[width=0.4\textwidth]{ESD_Ex1_Dim27_TRUE.jpg}}\\
\subfloat[SCAD]{\label{fig:Austria_CoVaR}\includegraphics[width=0.4\textwidth]{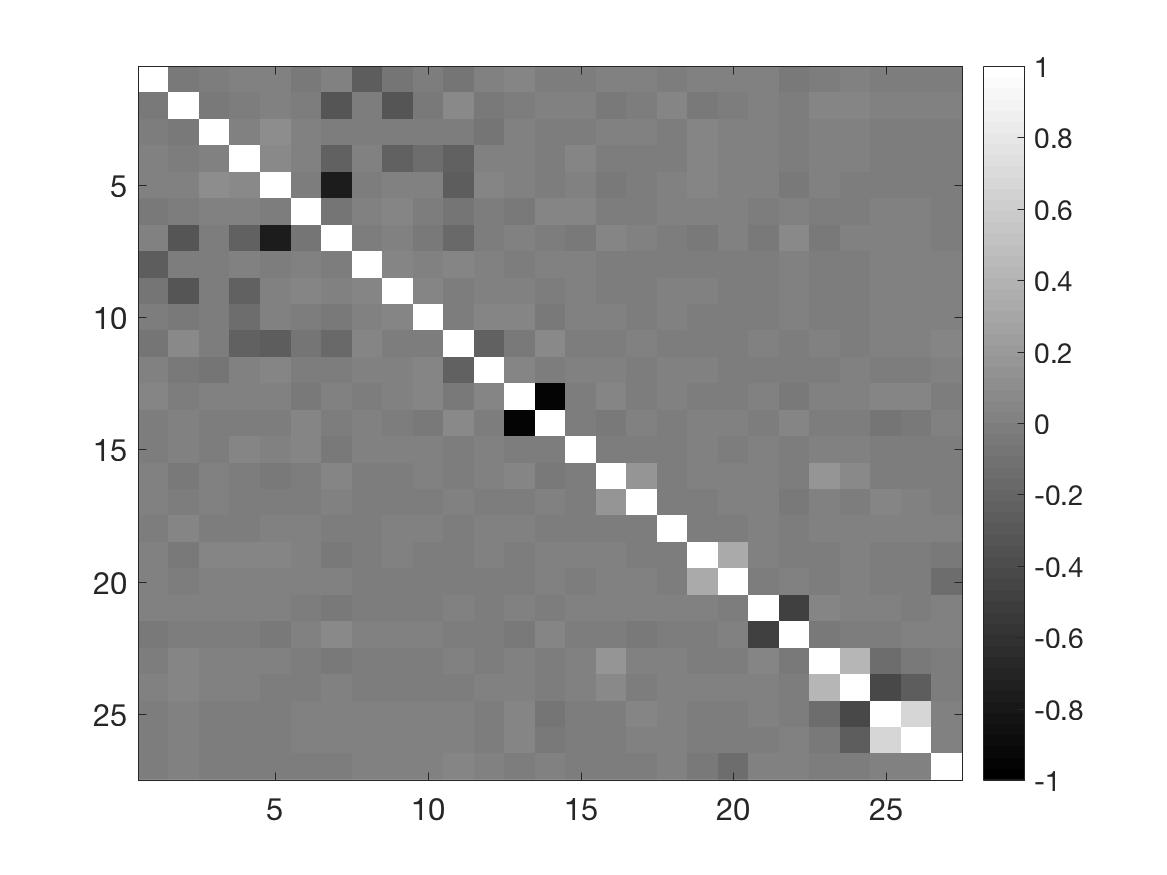}}\quad
\subfloat[Adaptive Lasso]{\label{fig:Austria_CoVaR}\includegraphics[width=0.4\textwidth]{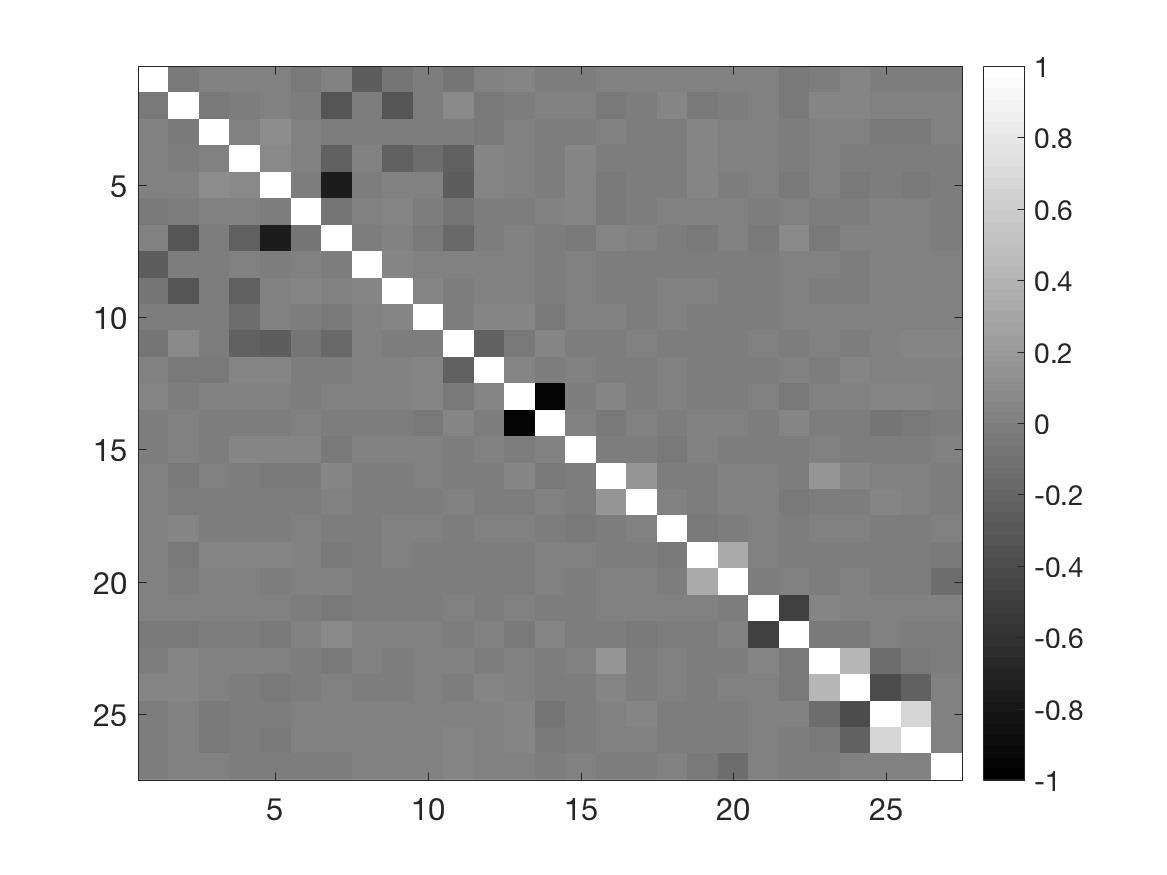}}\\
\caption{\footnotesize{Band structure of the true {\it (left)} and estimated {\it (right)} scale matrices through S--MMSQ of the $12$--dimensional Elliptical Stable simulated experiment discussed in Section  \ref{sec:msq_synthetic_ex}, for $\alpha =1.95$ and sample size $n=200$.}}
\label{fig:smmsq_27_alpha_195}
\end{center}
\end{figure}
%
%
\begin{figure}[!ht]
\captionsetup{font={footnotesize}}
\begin{center}
\subfloat[GLASSO]{\label{fig:dim27_true}\includegraphics[width=0.4\textwidth]{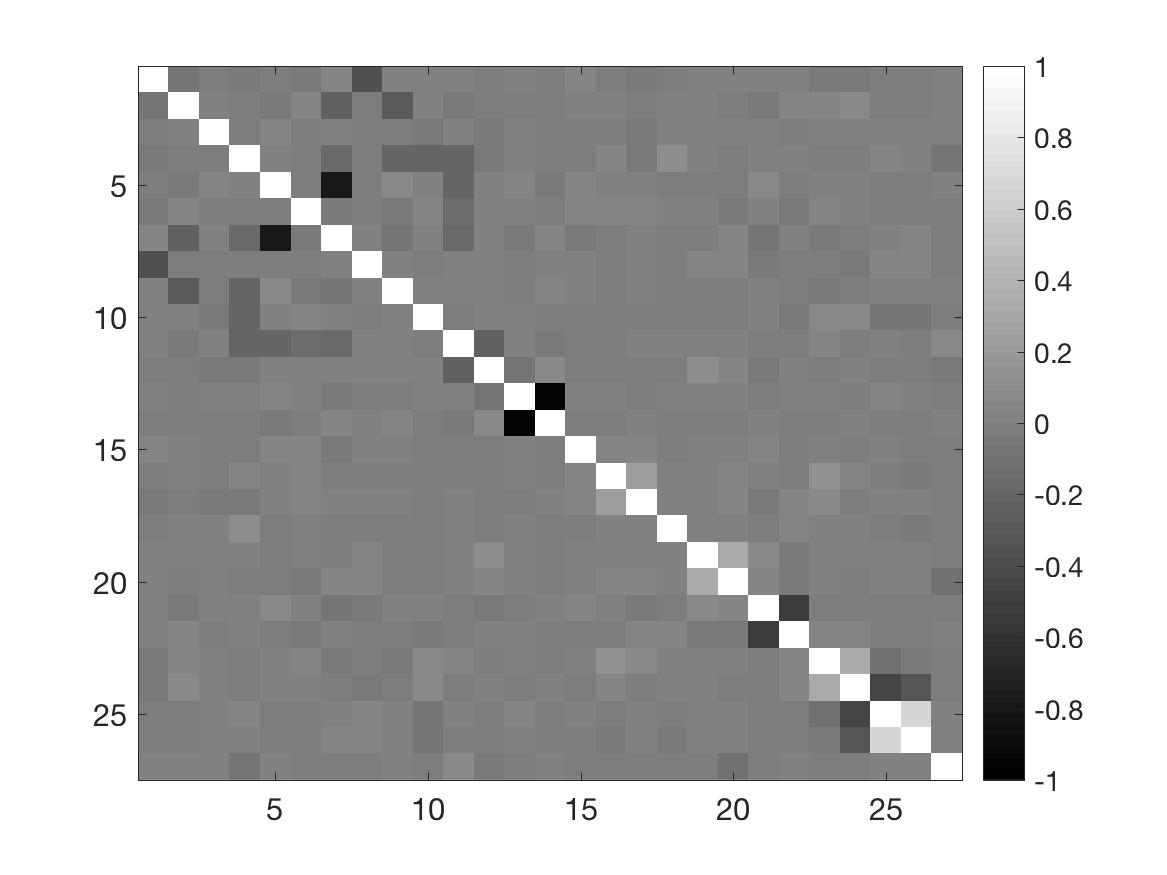}}\quad
\subfloat[S--MMSQ]{\label{fig:dim27}\includegraphics[width=0.4\textwidth]{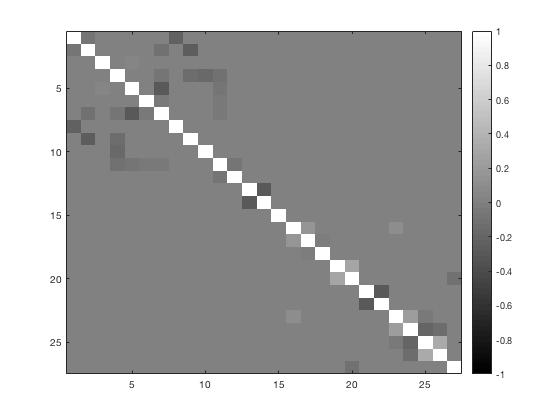}}\\
\subfloat[TRUE]{\label{fig:dim27}\includegraphics[width=0.4\textwidth]{ESD_Ex1_Dim27_TRUE.jpg}}\\
\subfloat[SCAD]{\label{fig:Austria_CoVaR}\includegraphics[width=0.4\textwidth]{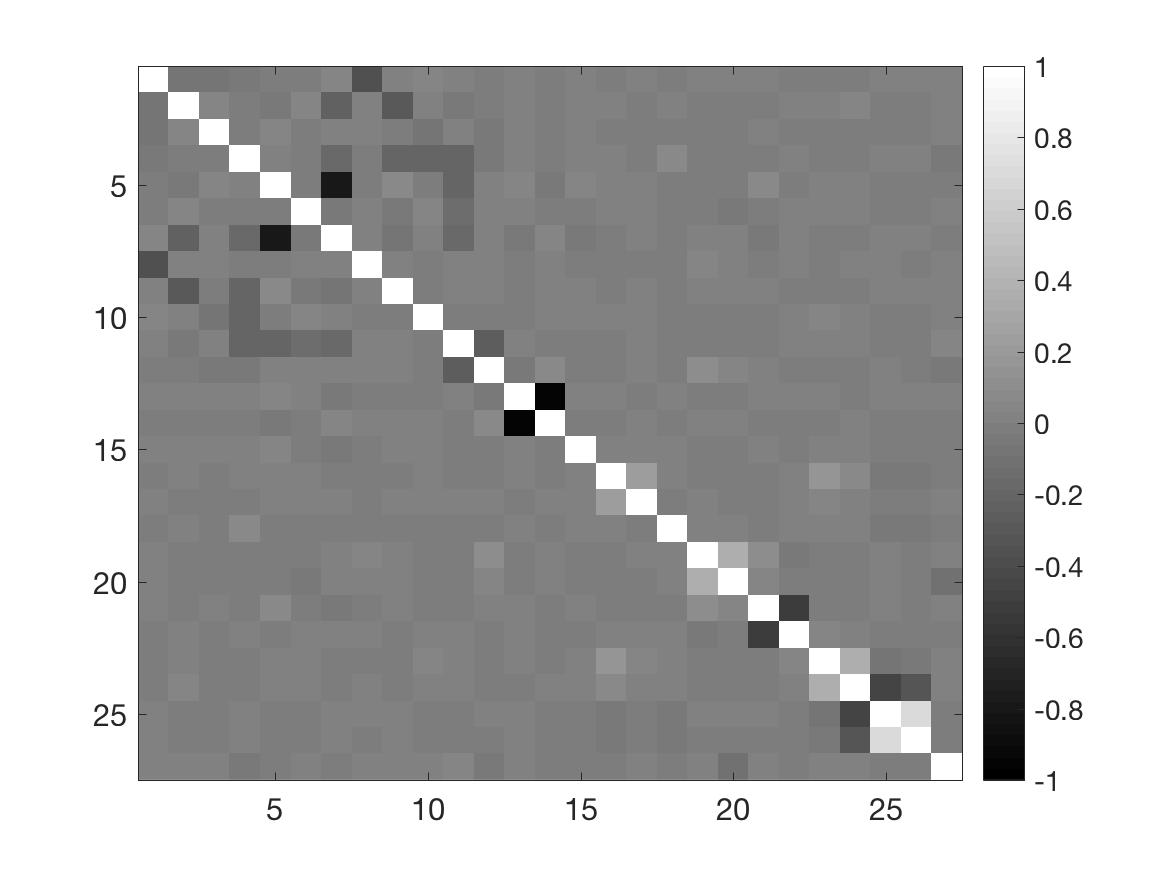}}\quad
\subfloat[Adaptive Lasso]{\label{fig:Austria_CoVaR}\includegraphics[width=0.4\textwidth]{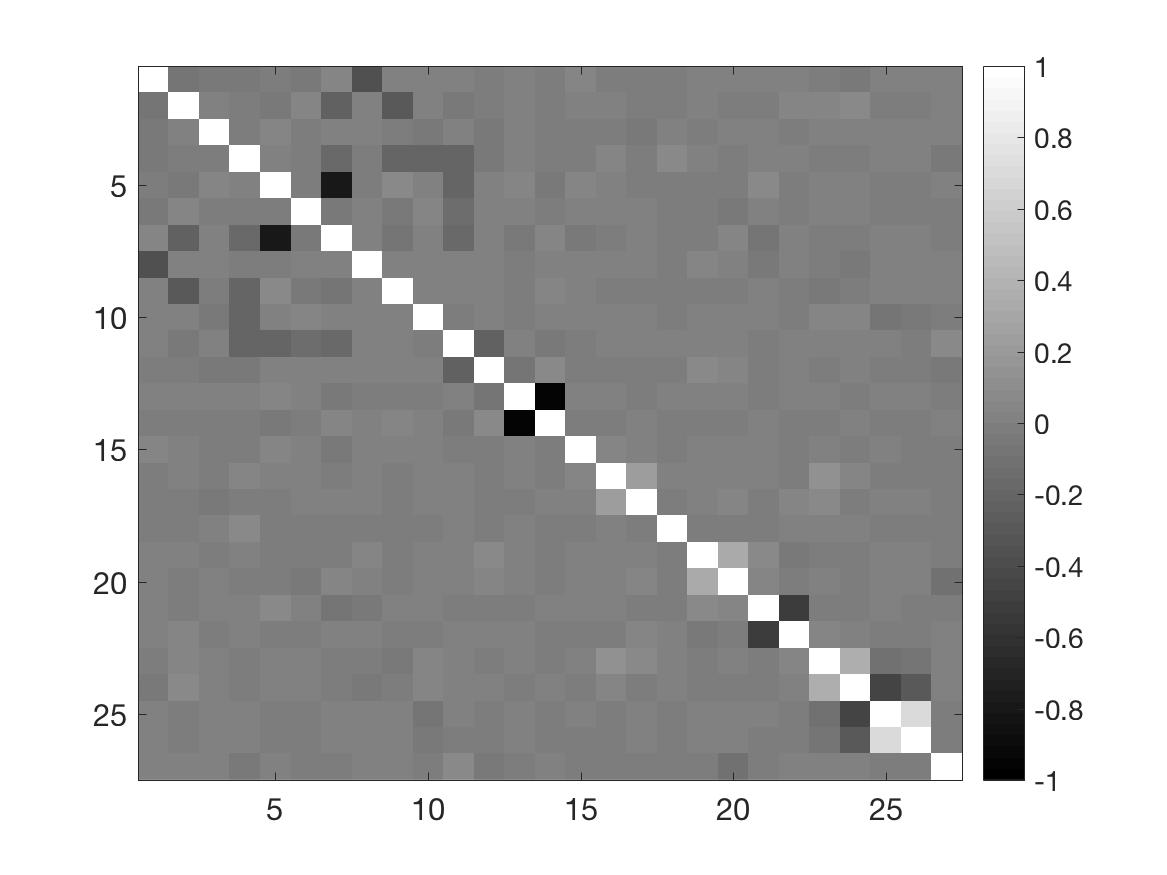}}\\
\caption{\footnotesize{Band structure of the true {\it (left)} and estimated {\it (right)} scale matrices through S--MMSQ of the $12$--dimensional Elliptical Stable simulated experiment discussed in Section  \ref{sec:msq_synthetic_ex}, for $\alpha =2.00$ and sample size $n=200$.}}
\label{fig:smmsq_27_alpha_200}
\end{center}
\end{figure}
%

%
\newpage
\clearpage

\begin{thebibliography}{}

\bibitem[Abdul-Hamid and Nolan, 1998]{abdul-hamid_nolan.1998}
Abdul-Hamid, H. and Nolan, J.~P. (1998).
\newblock Multivariate stable densities as functions of one-dimensional
  projections.
\newblock {\em J. Multivariate Anal.}, 67(1):80--89.

\bibitem[Acharya et~al., 2012]{acharya_etal.2012}
Acharya, V., Engle, R., and Richardson, M. (2012).
\newblock Capital shortfall: a new approach to ranking and regulating systemic
  risks.
\newblock {\em {American Economic Review}}, 102:59--64.

\bibitem[Adrian and Brunnermeier, 2011]{adrian_brunnermeier.2011}
Adrian, T. and Brunnermeier, M. (2011).
\newblock Covar.
\newblock {\em Working paper}.

\bibitem[Adrian and Brunnermeier, 2016]{adrian_brunnermeier.2016}
Adrian, T. and Brunnermeier, M.~K. (2016).
\newblock Covar.
\newblock {\em American Economic Review}, 106(7):1705--41.

\bibitem[Baldi et~al., 2000]{baldi_etal.2000}
Baldi, P., Brunak, S., Chauvin, Y., Andersen, C. A.~F., and Nielsen, H. (2000).
\newblock Assessing the accuracy of prediction algorithms for classification:
  an overview.
\newblock {\em Bioinformatics}, 16(5):412--424.

\bibitem[Benoit et~al., 2016]{benoit_etal.2016}
Benoit, S., Colliard, J.-E., Hurlin, C., and P{\'e}rignon, C. (2016).
\newblock Where the risks lie: A survey on systemic risk.
\newblock {\em Review of Finance}.

\bibitem[Bernardi and Catania, 2015]{bernardi_catania.2015}
Bernardi, M. and Catania, L. (2015).
\newblock Switching-gas copula models with application to systemic risk.
\newblock {\em arXiv preprint arXiv:1504.03733}.

\bibitem[Bernardi et~al., 2016a]{bernardi_etal.2016c}
Bernardi, M., Cerqueti, R., and Palestini, A. (2016a).
\newblock Allocation of risk capital in a cost cooperative game induced by a
  modified expected shortfall.
\newblock {\em arXiv preprint arXiv:1608.02365}.

\bibitem[Bernardi et~al., 2016b]{bernardi_etal.2016}
Bernardi, M., Durante, F., and Jaworski, P. (2016b).
\newblock Covar of families of copula.
\newblock {\em forthcoming Statistics \& Probability Letters}.

\bibitem[Bernardi et~al., 2017a]{bernardi_etal.2017}
Bernardi, M., Durante, F., and Jaworski, P. (2017a).
\newblock Co{V}a{R} of families of copulas.
\newblock {\em Statist. Probab. Lett.}, 120:8--17.

\bibitem[Bernardi et~al., 2017b]{bernardi_etal.2017b}
Bernardi, M., Durante, F., Jaworski, P., Petrella, L., and Salvadori, G.
  (2017b).
\newblock Conditional risk based on multivariate hazard scenarios.
\newblock {\em Stochastic Environmental Research and Risk Assessment}, pages
  1--9.

\bibitem[Bernardi et~al., 2016c]{bernardi_etal.2016b}
Bernardi, M., Durante, F., Jaworski, P., and Salvadori, G. (2016c).
\newblock Conditional risk based on multivariate hazard scenarios.
\newblock {\em submitted SERRA}.

\bibitem[Bernardi et~al., 2015]{bernardi_etal.2015}
Bernardi, M., Gayraud, G., and Petrella, L. (2015).
\newblock Bayesian tail risk interdependence using quantile regression.
\newblock {\em Bayesian Anal.}, 10(3):553--603.

\bibitem[Bernardi et~al., 2017c]{bernardi_etal.2017d}
Bernardi, M., Maruotti, A., and Petrella, L. (2017c).
\newblock Multiple risk measures for multivariate dynamic heavy--tailed models.
\newblock {\em Journal of Empirical Finance}, 43(Supplement C):1 -- 32.

\bibitem[Bien and Tibshirani, 2011]{bien_tibshirani.2011}
Bien, J. and Tibshirani, R.~J. (2011).
\newblock Sparse estimation of a covariance matrix.
\newblock {\em Biometrika}, 98(4):807--820.

\bibitem[Billio et~al., 2012]{billio_etal.2012}
Billio, M., Getmansky, M., Lo, A., and Pellizon, L. (2012).
\newblock Econometric measures of connectedness and systemic risk in the
  finance and insurance sectors.
\newblock {\em {Journal of Financial Econometrics}}, 101:535--559.

\bibitem[Bisias et~al., 2012]{bisias_etal.2012}
Bisias, D., Flood, M., Lo, A.~W., and Valavanis, S. (2012).
\newblock A survey of systemic risk analytics.
\newblock {\em Annual Review of Financial Economics}, 4(1):255--296.

\bibitem[Branco and Dey, 2001]{branco_dey.2001}
Branco, M.~D. and Dey, D.~K. (2001).
\newblock A general class of multivariate skew-elliptical distributions.
\newblock {\em J. Multivariate Anal.}, 79(1):99--113.

\bibitem[Brownlees and Engle, 2016]{brownlees_engle.2016}
Brownlees, C. and Engle, R.~F. (2016).
\newblock {SRISK}: {A} {C}onditional {C}apital {S}hortfall {M}easure of
  {S}ystemic {R}isk.
\newblock {\em Review of Financial Studies}.

\bibitem[Byczkowski et~al., 1993]{byczkowski_etal.1993}
Byczkowski, T., Nolan, J.~P., and Rajput, B. (1993).
\newblock Approximation of multidimensional stable densities.
\newblock {\em J. Multivariate Anal.}, 46(1):13--31.

\bibitem[Caporin et~al., 2013]{caporin_etal.2013}
Caporin, M., Pelizzon, L., Ravazzolo, F., and Rigobon, R. (2013).
\newblock Measuring sovereign contagion in europe.
\newblock Technical report, National Bureau of Economic Research.

\bibitem[Castro and Ferrari, 2014]{castro_ferrari.2014}
Castro, C. and Ferrari, S. (2014).
\newblock Measuring and testing for the systemically important financial
  institutions.
\newblock {\em Journal of Empirical Finance}, 25(0):1 -- 14.

\bibitem[Cram{\'e}r, 1946]{cramer.1946}
Cram{\'e}r, H. (1946).
\newblock {\em Mathematical {M}ethods of {S}tatistics}.
\newblock Princeton Mathematical Series, vol. 9. Princeton University Press,
  Princeton, N. J.

\bibitem[Dominicy et~al., 2013]{dominicy_etal.2013}
Dominicy, Y., Ogata, H., and Veredas, D. (2013).
\newblock Inference for vast dimensional elliptical distributions.
\newblock {\em Comput. Statist.}, 28(4):1853--1880.

\bibitem[Dominicy and Veredas, 2013]{dominicy_veredas.2013}
Dominicy, Y. and Veredas, D. (2013).
\newblock The method of simulated quantiles.
\newblock {\em J. Econometrics}, 172(2):235--247.

\bibitem[Embrechts et~al., 2005]{embrechts_etal.2005}
Embrechts, P., Frey, R., and McNeil, A. (2005).
\newblock Quantitative risk management.
\newblock {\em Princeton Series in Finance, Princeton}, 10.

\bibitem[Engle et~al., 2014]{engle_etal.2014}
Engle, R., Jondeau, E., and Rockinger, M. (2014).
\newblock Systemic risk in europe.
\newblock {\em Review of Finance}.

\bibitem[Fan et~al., 2009]{fan_etal.2009}
Fan, J., Feng, Y., and Wu, Y. (2009).
\newblock Network exploration via the adaptive lasso and {SCAD} penalties.
\newblock {\em Ann. Appl. Stat.}, 3(2):521--541.

\bibitem[Fan and Li, 2001]{fan_li.2001}
Fan, J. and Li, R. (2001).
\newblock Variable selection via nonconcave penalized likelihood and its oracle
  properties.
\newblock {\em J. Amer. Statist. Assoc.}, 96(456):1348--1360.

\bibitem[Fang et~al., 1990]{fang_etal1990}
Fang, K.~T., Kotz, S., and Ng, K.~W. (1990).
\newblock {\em Symmetric multivariate and related distributions}, volume~36 of
  {\em Monographs on Statistics and Applied Probability}.
\newblock Chapman and Hall, Ltd., London.

\bibitem[Friedman et~al., 2008]{friedman_etal.2008}
Friedman, J., Hastie, T., and Tibshirani, R. (2008).
\newblock Sparse inverse covariance estimation with the graphical lasso.
\newblock {\em Biostatistics (Oxford, England)}, 9(3):432--441.

\bibitem[Gallant and Tauchen, 1996]{gallant_tauchen.1996}
Gallant, A.~R. and Tauchen, G. (1996).
\newblock Which moments to match?
\newblock {\em Econometric Theory}, 12(4):657--681.

\bibitem[Gao and Massam, 2015]{gao_massam.2015}
Gao, X. and Massam, H. (2015).
\newblock Estimation of symmetry-constrained {G}aussian graphical models:
  application to clustered dense networks.
\newblock {\em J. Comput. Graph. Statist.}, 24(4):909--929.

\bibitem[Girardi and Erg\"{u}n, 2013]{girardi_ergun.2013}
Girardi, G. and Erg\"{u}n, A. (2013).
\newblock Systemic risk measurement: Multivariate garch estimation of covar.
\newblock {\em {Journal of Banking \& Finance}}, 37:3169--3180.

\bibitem[Gouri{\'e}roux and Monfort, 1996]{gourieroux_monfort.1996}
Gouri{\'e}roux, C. and Monfort, A. (1996).
\newblock {\em Simulation-based econometric methods}.
\newblock CORE lectures. Oxford University Press.

\bibitem[Gouri{\'e}roux et~al., 1993]{gourieroux_etal.1993}
Gouri{\'e}roux, C., Monfort, A., and Renault, E. (1993).
\newblock Indirect inference.
\newblock {\em Journal of Applied Econometrics}, 8(S1):S85--S118.

\bibitem[Hallin et~al., 2010a]{hallin_etal.2010b}
Hallin, M., Paindaveine, D., and {\v{S}}iman, M. (2010a).
\newblock Multivariate quantiles and multiple-output regression quantiles: from
  {$L_1$} optimization to halfspace depth.
\newblock {\em Ann. Statist.}, 38(2):635--669.

\bibitem[Hallin et~al., 2010b]{hallin_etal.2010}
Hallin, M., Paindaveine, D., and {\v{S}}iman, M. (2010b).
\newblock Rejoinder [mr2604670; mr2604671; mr2604672; mr2604673].
\newblock {\em Ann. Statist.}, 38(2):694--703.

\bibitem[Hansen, 1982]{hansen.1982}
Hansen, L.~P. (1982).
\newblock Large sample properties of generalized method of moments estimators.
\newblock {\em Econometrica}, 50(4):1029--1054.

\bibitem[Hautsch et~al., 2014]{hautsch_etal.2014}
Hautsch, N., Schaumburg, J., and Schienle, M. (2014).
\newblock Financial network systemic risk contributions.
\newblock {\em Review of Finance}.

\bibitem[Hunter and Li, 2005]{hunter_li.2005}
Hunter, D.~R. and Li, R. (2005).
\newblock Variable selection using mm algorithms.
\newblock {\em Ann. Statist.}, 33(4):1617--1642.

\bibitem[Jiang and Turnbull, 2004]{jiang_turnbull.2004}
Jiang, W. and Turnbull, B. (2004).
\newblock The indirect method: inference based on intermediate statistics---a
  synthesis and examples.
\newblock {\em Statist. Sci.}, 19(2):239--263.

\bibitem[Kim and White, 2004]{Kim_white.2004}
Kim, T.-H. and White, H. (2004).
\newblock On more robust estimation of skewness and kurtosis.
\newblock {\em Finance Research Letters}, 1(1):56--73.

\bibitem[Koenker, 2005]{koenker.2005}
Koenker, R. (2005).
\newblock {\em Quantile regression}, volume~38 of {\em Econometric Society
  Monographs}.
\newblock Cambridge University Press, Cambridge.

\bibitem[Kong and Mizera, 2012]{kong_mizera.2012}
Kong, L. and Mizera, I. (2012).
\newblock Quantile tomography: using quantiles with multivariate data.
\newblock {\em Statist. Sinica}, 22(4):1589--1610.

\bibitem[Koponen, 1995]{ismo.1995}
Koponen, I. (1995).
\newblock Analytic approach to the problem of convergence of truncated l\'evy
  flights towards the gaussian stochastic process.
\newblock {\em Phys. Rev. E}, 52:1197--1199.

\bibitem[Kristensen and Shin, 2012]{kristensen_shin.2012}
Kristensen, D. and Shin, Y. (2012).
\newblock Estimation of dynamic models with nonparametric simulated maximum
  likelihood.
\newblock {\em J. Econometrics}, 167(1):76--94.

\bibitem[Lombardi and Veredas, 2009]{lombardi_veredas.2009}
Lombardi, M.~J. and Veredas, D. (2009).
\newblock Indirect estimation of elliptical stable distributions.
\newblock {\em Comput. Statist. Data Anal.}, 53(6):2309--2324.

\bibitem[Lucas et~al., 2014]{lucas_etal.2014}
Lucas, A., Schwaab, B., and Zhang, X. (2014).
\newblock Conditional euro area sovereign default risk.
\newblock {\em Journal of Business \& Economic Statistics}, 32(2):271--284.

\bibitem[Mainik and Schaanning, 2014]{mainik_schaanning.2014}
Mainik, G. and Schaanning, E. (2014).
\newblock {On dependence consistency of CoVaR and some other systemic risk
  measures}.
\newblock {\em Statistics \& Risk Modeling}, 31(1):47--77.

\bibitem[Mathai and Moschopoulos, 1992]{mathai_moschopoulos.2015}
Mathai, A.~M. and Moschopoulos, P.~G. (1992).
\newblock A form of multivariate gamma distribution.
\newblock {\em Ann. Inst. Statist. Math.}, 44(1):97--106.

\bibitem[Matsui and Takemura, 2009]{matsui_takemura.2009}
Matsui, M. and Takemura, A. (2009).
\newblock Integral representations of one-dimensional projections for
  multivariate stable densities.
\newblock {\em Journal of Multivariate Analysis}, 100(3):334 -- 344.

\bibitem[McCulloch, 1986]{mcculloch.1986}
McCulloch, J.~H. (1986).
\newblock Simple consistent estimators of stable distribution parameters.
\newblock {\em Communications in Statistics-Simulation and Computation},
  15(4):1109--1136.

\bibitem[McFadden, 1989]{mcfadden.1989}
McFadden, D. (1989).
\newblock A method of simulated moments for estimation of discrete response
  models without numerical integration.
\newblock {\em Econometrica}, 57(5):995--1026.

\bibitem[McNeil et~al., 2015]{mcneil_etal.2015}
McNeil, A.~J., Frey, R., and Embrechts, P. (2015).
\newblock {\em Quantitative risk management}.
\newblock Princeton Series in Finance. Princeton University Press, Princeton,
  NJ, revised edition.
\newblock Concepts, techniques and tools.

\bibitem[Meinshausen and B\"{u}hlmann, 2006]{meinshausen_buhlmann.2006}
Meinshausen, N. and B\"{u}hlmann, P. (2006).
\newblock High-dimensional graphs and variable selection with the lasso.
\newblock {\em Ann. Statist.}, 34(3):1436--1462.

\bibitem[Nolan, 2008]{nolan.2008}
Nolan, J.~P. (2008).
\newblock An overview of multivariate stable distributions.
\newblock {\em Online: http://academic2. american. edu/~
  jpnolan/stable/overview. pdf}, 2(008).

\bibitem[Nolan, 2013]{nolan.2013}
Nolan, J.~P. (2013).
\newblock Multivariate elliptically contoured stable distributions: theory and
  estimation.
\newblock {\em Comput. Statist.}, 28(5):2067--2089.

\bibitem[Oh and Patton, 2013]{oh_patton.2013}
Oh, D.~H. and Patton, A.~J. (2013).
\newblock Simulated method of moments estimation for copula-based multivariate
  models.
\newblock {\em Journal of the American Statistical Association},
  108(502):689--700.

\bibitem[Paindaveine and {\v{S}}iman, 2011]{paindaveine_siman.2011}
Paindaveine, D. and {\v{S}}iman, M. (2011).
\newblock On directional multiple-output quantile regression.
\newblock {\em J. Multivariate Anal.}, 102(2):193--212.

\bibitem[Salvadori et~al., 2016]{salvadori_etal.2016}
Salvadori, G., Durante, F., De~Michele, C., Bernardi, M., and Petrella, L.
  (2016).
\newblock A multivariate copula-based framework for dealing with hazard
  scenarios and failure probabilities.
\newblock {\em Water Resources Research}, 52(5):3701--3721.

\bibitem[Samorodnitsky and Taqqu, 1994]{samorodnitsky_etal.1994}
Samorodnitsky, G. and Taqqu, M.~S. (1994).
\newblock {\em Stable non-{G}aussian random processes}.
\newblock Stochastic Modeling. Chapman \& Hall, New York.
\newblock Stochastic models with infinite variance.

\bibitem[Serfling, 2002]{serfling.2002}
Serfling, R. (2002).
\newblock Quantile functions for multivariate analysis: approaches and
  applications.
\newblock {\em Statist. Neerlandica}, 56(2):214--232.
\newblock Special issue: Frontier research in theoretical statistics, 2000
  (Eindhoven).

\bibitem[Sordo et~al., 2015]{sordo_etal.2015}
Sordo, M.~A., Su\'arez-Llorens, A., and Bello, A.~J. (2015).
\newblock Comparison of conditional distributions in portfolios of dependent
  risks.
\newblock {\em Insurance: Mathematics and Economics}, 61:62 -- 69.

\bibitem[Tibshirani, 1996]{tibshirani.1996}
Tibshirani, R. (1996).
\newblock Regression shrinkage and selection via the lasso.
\newblock {\em J. Roy. Statist. Soc. Ser. B}, 58(1):267--288.

\bibitem[Wang, 2010]{wang.2010}
Wang, H. (2010).
\newblock Sparse seemingly unrelated regression modelling: applications in
  finance and econometrics.
\newblock {\em Comput. Statist. Data Anal.}, 54(11):2866--2877.

\bibitem[Wang, 2015]{wang.2015}
Wang, H. (2015).
\newblock Scaling it up: Stochastic search structure learning in graphical
  models.
\newblock {\em Bayesian Anal.}, 10(2):351--377.

\bibitem[Zhang, 2010]{zhang.2010}
Zhang, C.-H. (2010).
\newblock Nearly unbiased variable selection under minimax concave penalty.
\newblock {\em Ann. Statist.}, 38(2):894--942.

\bibitem[Zolotarev, 1964]{zolotarev.1964}
Zolotarev, V.~M. (1964).
\newblock On the representation of stable laws by integrals.
\newblock {\em Trudy Mat. Inst. Steklov.}, 71:46--50.

\bibitem[Zou, 2006]{zou.2006}
Zou, H. (2006).
\newblock The adaptive lasso and its oracle properties.
\newblock {\em J. Amer. Statist. Assoc.}, 101(476):1418--1429.

\end{thebibliography}

\end{document}